\documentclass[11pt]{article}
\usepackage{amsmath}
\usepackage{amsfonts}
\usepackage{amssymb}
\usepackage{amsthm}
\usepackage{pgfplots}
\usepackage{cleveref}
\usepackage{graphicx}
\usepackage{graphicx}
\usepackage{marvosym}
\usepackage{amsthm}
\usepackage{mathrsfs} 
\usepackage{xfrac} 
\usepackage{tabu}
\usepackage{bbm}

\usepackage{amsmath}
\usepackage{amsfonts}
\usepackage{amssymb}
\usepackage{graphicx}
\usepackage{tikz}
\usetikzlibrary{arrows,shapes,positioning}
\usetikzlibrary{calc}
\usepackage{pgfplots}
\usepackage{cleveref}
\usepackage{graphicx}
\usepackage{subcaption}
\usetikzlibrary{trees,snakes}
\usepackage{verbatim}
\usepackage{mathtools}

\usepackage{mathrsfs} 
\usepackage{xfrac} 
\usepackage{tabu}
\usepackage{color}
\usepackage{tikz}

\usepackage{mathtools}
\usepackage{tikz}
\usetikzlibrary{calc}
\usepackage{pgfplots}
\usepackage{cleveref}

\usepackage{color}
\definecolor{myred}{RGB}{255, 0, 0}
\definecolor{myblue}{RGB}{0, 0, 255}
\newtheorem{theorem}{Theorem}
\newtheorem{lemma}{Lemma}
\newtheorem{proposition}{Proposition}

\newcommand{\EE}{\mathsf{E}}
\newcommand{\RF}{R}
\newcommand{\rr}{R}
\newcommand{\er}{\mathsf{E}_{r}}
\newcommand{\ee}{\mathsf{E}_{e}}
\newcommand{\nn}{\nonumber}
\newcommand{\IND}{\mathbbm{1}}
\newcommand {\Exp} {\mathbb{E}}
\newcommand {\prob} {\mathbb{P}}
\newcommand{\DEF}{\stackrel{\triangle}{=}}

\newcommand{\dfn}{\stackrel{\triangle}{=}}

\newcommand {\lexe} {\stackrel{\cdot} {\le}}

\newcommand {\reals} {{\rm I\!R}}

\newcommand {\hH} {\hat{H}}

\newcommand {\bu} {\boldsymbol{u}}

\newcommand {\bv} {\boldsymbol{v}}

\newcommand {\bU} {\boldsymbol{U}}
\newcommand {\bV} {\boldsymbol{V}}

\newcommand{\calA}{{\cal A}}
\newcommand{\calB}{{\cal B}}
\newcommand{\calC}{{\cal C}}

\newcommand{\calE}{{\cal E}}

\newcommand{\calG}{{\cal G}}

\newcommand{\calI}{{\cal I}}

\newcommand{\calK}{{\cal K}}

\newcommand{\calP}{{\cal P}}
\newcommand{\calQ}{{\cal Q}}

\newcommand{\calS}{{\cal S}}
\newcommand{\calT}{{\cal T}}
\newcommand{\calU}{{\cal U}}
\newcommand{\calV}{{\cal V}}

\newcommand{\calX}{{\cal X}}
\newcommand{\calY}{{\cal Y}}

\begin{document}
\thispagestyle{empty}
\title{Trade-offs Between Error Exponents and Excess-Rate Exponents of Typical Slepian--Wolf Codes\footnote{
		This research was supported by the Israel Science Foundation (ISF) grant no.\ 137/18. This paper was presented in part at the 2019 IEEE Information Theory Workshop, Visby, Gotland, Sweden, 25-28 August, 2019.}\\}
\author{\\ Ran Tamir (Averbuch) and Neri Merhav\\}
\maketitle
\begin{center}
The Andrew \& Erna Viterbi Faculty of Electrical Engineering \\
Technion - Israel Institute of Technology \\
Technion City, Haifa 3200003, ISRAEL \\
\{rans@campus, merhav@ee\}.technion.ac.il
\end{center}
\vspace{1.5\baselineskip}
\setlength{\baselineskip}{1.5\baselineskip}

\begin{abstract}
	Typical random codes (TRC) in a communication scenario of source coding with side information at the decoder is the main subject of this work. 
	We study the semi-deterministic code ensemble, which is a certain variant of the ordinary random binning code ensemble. In this code ensemble, the relatively small type classes of the source are deterministically partitioned into the available bins in a one-to-one manner. As a consequence, the error probability decreases dramatically. 
	The random binning error exponent and the error exponent of the TRC are derived and proved to be equal to one another in a few important special cases.
	We show that the performance under optimal decoding can be attained also by certain universal decoders, e.g., the stochastic likelihood decoder with an empirical entropy metric.
	Moreover, we discuss the trade-offs between the error exponent and the excess--rate exponent for the typical random semi-deterministic code and characterize its optimal rate function. 
	We show that for any pair of correlated information sources, both error and excess--rate probabilities are exponentially vanishing when the blocklength tends to infinity. \\
	
\end{abstract}

\clearpage
\section{Introduction}

As is well known, the random coding error exponent is defined by 
\begin{align}
\label{RCE}
\EE_{\mbox{\tiny r}}(R) = \lim_{n \to \infty} \left\{ - \tfrac{1}{n} \log \mathbb{E} \left[P_{\mbox{\tiny e}}(\calC_{n}) \right] \right\},
\end{align} 
where $R$ is the coding rate, $P_{\mbox{\tiny e}}(\calC_{n})$ is the error probability of a codebook $\calC_{n}$, and the expectation is with respect to (w.r.t.) the randomness of $\calC_{n}$ across the ensemble of codes.
The error exponent of the typical random code (TRC) is defined as \cite{MERHAV_TYPICAL} 
\begin{align} \label{TRC_DEF}
\EE_{\mbox{\tiny trc}}(R) = \lim_{n \to \infty} \left\{- \tfrac{1}{n} \mathbb{E} \left[\log P_{\mbox{\tiny e}}(\calC_{n}) \right] \right\}.
\end{align}
We believe that the error exponent of the TRC is the more relevant performance metric as it captures the most likely error exponent of a randomly selected code, as opposed 
to the random coding error exponent, which is dominated by the relatively poor codes of the ensemble, rather than the channel noise, at relatively low coding rates.
In addition, since in random coding analysis, the code is selected at random and remains fixed, it seems reasonable to study the performance of the very chosen code instead of directly considering the ensemble performance. 

To the best of our knowledge, not much is known on TRCs. In \cite{BargForney}, Barg and Forney considered TRCs with independently and identically distributed codewords as well as typical linear codes, for the special case of the binary symmetric channel with maximum likelihood (ML) decoding. 
It was also shown that at a certain range of low rates, $E_{\mbox{\tiny trc}}(R)$ lies between $E_{\mbox{\tiny r}}(R)$ and the expurgated exponent, $E_{\mbox{\tiny ex}}(R)$. 
In \cite{PRAD2014} Nazari {\em et al.} provided bounds on the error exponents of TRCs for both discrete memoryless channels (DMC) and multiple--access channels. 
In a recent article by Merhav \cite{MERHAV_TYPICAL}, an exact single--letter expression has been derived for the error exponent of typical, random, fixed composition codes, over DMCs, and a wide class of (stochastic) decoders, collectively referred to as the generalized likelihood decoder (GLD).
Later, Merhav has studied error exponents of TRCs for the colored Gaussian channel \cite{MERHAV_GAUSS}, typical random trellis codes \cite{MERHAV_TRELLIS}, and a Lagrange--dual lower bound to the TRC exponent \cite{MERHAV_IID}. Large deviations around the TRC exponent was studied in \cite{TMWG}. 

While originally defined for pure channel coding \cite{BargForney}, \cite{MERHAV_TYPICAL}, \cite{PRAD2014}, the notion of TRCs has natural analogues in other settings as well, like source coding with side information at the decoder \cite{SW}. 
Typical random Slepian--Wolf (SW) codes of a certain variant of the ordinary variable--rate random binning code ensemble are the main theme of this work. 
The random coding error exponent of SW coding, based on fixed--rate (FR) random binning, was first addressed by Gallager in \cite{GAL76}, and improved later on by the expurgated bound in \cite{GoodCodes} and \cite{CKgraph}.
Variable--rate (VR) SW coding received less attention in the literature; VR codes under average rate constraint have been studied in \cite{SW2} and proved to outperform FR codes in terms of error exponents. 
Optimum trade-offs between the error exponent and the excess--rate exponent in VR coding were analyzed in \cite{SW9}. 
Sphere-packing upper bounds for source coding with side information in the FR and VR regimes have been studied in \cite{GAL76} and \cite{SW2}, respectively.
More works where exponential error bounds in source coding have been studied are \cite{CK1980}, \cite{C1982}, \cite{OH1994}, \cite{KW2011}, and \cite{KW2012}.   

It turns out that both the FR and VR ensembles suffer from an intrinsic deficiency, caused by statistical fluctuations in the sizes of the bins that are populated by the relatively small type classes of the source.    
This fundamental problem of the ordinary ensembles is alleviated in some variant of the ordinary VR ensemble -- the semi--deterministic (SD) code ensemble, which has already been proposed and studied in its FR version in \cite{OH1994}. 
In the SD code ensemble, 
for source type classes which are exponentially larger than the amount of available bins, we just randomly assign each source sequence into one of the bins, as being done in ordinary random binning.  
Otherwise, for relatively small type classes, we deterministically order each source sequence into a different bin, which provides a one--to--one mapping. 
This way, all these relatively small source type classes do not contribute to the probability of error.
The main results concerning the SD code are the following:
\begin{enumerate}
	\item The random binning error exponent and the error exponent of the TRC are derived in Theorems \ref{Thm_random_binning_SD_VR} and \ref{SW_THM}, respectively, and proved in Theorem \ref{Thm_comparison} to be equal to one another in a few important special cases, that includes the matched likelihood decoder, the MAP decoder, and the universal minimum entropy decoder. To the best of our knowledge, this phenomenon has not been seen elsewhere before, since the TRC exponent usually improves upon the random coding exponent.  
	As a byproduct, we are able to provide a relatively simple expression for the TRC exponent. 
	\item We prove in Theorem \ref{Thm_Universal} that the error exponent of the TRC under MAP decoding is also attained by two universal decoders: the minimum entropy decoder and the stochastic entropy decoder, which is a GLD with an empirical conditional entropy metric. As far as we know, this result is first of its kind in source coding; in other scenarios, the random coding bound is attained also by universal decoders, but here, we find that the TRC exponent is also universally achievable. Moreover, while the likelihood decoder and the MAP decoder have similar error exponents \cite{LCV2017}, here we prove a similar result, but for two universal decoders (one stochastic and one deterministic) that share the same metric.    
	\item We discuss the trade-offs between the error exponent and the excess--rate exponent for a typical random SD code, similarly to \cite{SW9}, but with a different notion of the excess--rate event, which takes into account the available side information. 
	In Theorem \ref{RF_UB}, we provide an expression for the optimal rate function that guarantees a required level for the error exponent of the typical random SD code. 
	Analogously, Theorem \ref{RF_LB_RC23} proposes an expression for the optimal rate function that guarantees a required level for the excess-rate exponent.    
	Furthermore, we find that for any pair of correlated information sources, the typical random SD code attains both exponentially vanishing error and excess--rate probabilities.     
\end{enumerate}        
    
The remaining part of the paper is organized as follows. 
In Section 2, we establish notation conventions. 
In Section 3, we formalize the model, the coding technique, the main objectives of this work, and we review some background. 
In Section 4, we provide the main results concerning error exponents and universal decoding in the SD ensemble, and in Section 5, we discuss the trade-offs between the error exponent and the excess-rate exponent.

\section{Notation Conventions}


Throughout the paper, random variables will be denoted by capital letters, realizations will be denoted by the corresponding lower case letters, and their alphabets will be denoted by calligraphic letters. 
Random vectors and their realizations will be denoted, 
respectively, by boldface capital and lower case letters. 
Their alphabets will be superscripted by their dimensions. 
Sources and channels will be subscripted by the names of the relevant random 
variables/vectors and their conditionings, whenever applicable, 
following the standard notation conventions, e.g., $Q_{U}$, $Q_{V|U}$, and so on. 
When there is no room for ambiguity, these subscripts will be omitted. For a generic joint 
distribution $Q_{UV} = \{Q_{UV}(u,v), u \in \mathcal{U}, v \in \mathcal{V} \}$, which will often 
be abbreviated by $Q$, information measures will be denoted in the conventional manner, but
with a subscript $Q$, that is, $H_{Q}(U)$ 
is the marginal entropy of $U$, $H_{Q}(U|V)$ is the conditional entropy of $U$ given $V$, and 
$I_{Q}(U;V) = H_{Q}(U) - H_{Q}(U|V)$ is the mutual information between $U$ and $V$.
The Kullback--Leibler divergence between two probability distributions, $Q_{UV}$ and $P_{UV}$, is defined as
\begin{align}
D(Q_{UV}\|P_{UV}) = \sum_{(u,v) \in \calU \times \calV} Q_{UV}(u,v) \log \frac{Q_{UV}(u,v)}{P_{UV}(u,v)},
\end{align}
where logarithms, here and throughout the sequel, are understood to be taken to the natural base.
The probability of an event $\mathcal{E}$ will be denoted by $\prob \{ \mathcal{E} \}$, and the expectation operator w.r.t.\ a 
probability distribution $Q$ will be denoted by $\mathbb{E}_{Q} [\cdot]$, where the subscript will often be omitted. 
For two positive sequences, $\{a_{n}\}$ and $\{b_{n}\}$, the notation $a_{n} \doteq b_{n}$ will stand for equality in the exponential scale, that is, $\lim_{n \to \infty} (1/n) \log \left(a_{n}/b_{n}\right) = 0$. Similarly, $a_{n} \lexe b_{n}$ means that $\limsup_{n \to \infty} (1/n) \log \left(a_{n}/b_{n}\right) \leq 0$, and so on. 
The indicator function of an event $\calA$ 
will be denoted by $\IND\{\calA\}$. 
The notation $[t]_{+}$ will stand for $\max\{0,t\}$. 

The empirical distribution of a sequence $\bu \in \mathcal{U}^{n}$, which will 
be denoted by $\hat{P}_{\bu}$, is the vector of relative frequencies, $\hat{P}_{\bu}(u)$, 
of each symbol $u \in \mathcal{U}$ in $\bu$. 
The type class of $\bu \in \mathcal{U}^{n}$, denoted $\calT(\bu)$, 
is the set of all vectors $\bu'$ with $\hat{P}_{\bu'} = \hat{P}_{\bu}$. 
When we wish to emphasize the dependence of the type class on the empirical 
distribution $\hat{P}$, we will denote it by $\calT(\hat{P})$. 
The set of all types of vectors of length $n$ over $\calU$ will be denoted by $\calP_{n}(\calU)$, and the set of all possible types over $\calU$ will be denoted by $\calP(\calU) \dfn \bigcup_{n=1}^{\infty} \calP_{n}(\calU)$.
Information measures associated with empirical distributions will be denoted with `hats' and will be subscripted 
by the sequences from which they are induced. For example, the entropy associated 
with $\hat{P}_{\bu}$, which is the empirical entropy of $\bu$, will be denoted 
by $\hat{H}_{\bu}(U)$. Similar conventions will apply to the joint empirical distribution, 
the joint type class, the conditional empirical distributions and the conditional type classes 
associated with pairs (and multiples) of sequences of length $n$. 
Accordingly, $\hat{P}_{\bu\bv}$ would be the joint empirical distribution 
of $(\bu, \bv) = \{(u_{i}, v_{i})\}_{i=1}^{n}$, $\calT(\hat{P}_{\bu\bv})$ will denote the joint type class of $(\bu, \bv)$, $\calT( \hat{P}_{\bu|\bv} | \bv)$ will stand for the conditional type class of $\bu$ given $\bv$, 
$\hat{H}_{\bu\bv}(U|V)$ will be the empirical conditional entropy, and so on. 
Likewise, when we wish to emphasize the dependence of empirical information measures upon a given empirical distribution $Q$, 
we denote them using the subscript $Q$, as described above.

\section{Problem Formulation and Background} \label{SEC3}

\subsection{Problem Formulation} \label{SEC3.1}
Let $(\bU,\bV) = \{(U_{t},V_{t})\}_{t=1}^{n}$ be $n$ independent copies of a pair of random variables, $(U,V) \sim P_{UV}$, taking on values in finite alphabets, $\calU$ and $\calV$, respectively. The vector $\bU$ will designate the source vector to be encoded and the vector $\bV$ will serve as correlated side information, available to the decoder. 
In ordinary VR binning, the coding rate is not fixed for every $\bu \in \calU^{n}$, but depends on its empirical distribution. Let us denote a rate function by $\RF(\cdot)$, which is a given continuous function from the probability simplex of $\calU$ to the set of nonnegative reals. In that manner, for every type $Q_{U} \in \calP_{n}(\calU)$, all source sequences in $\calT(Q_{U})$ are randomly partitioned into $e^{n \RF(Q_{U})}$ bins. 
Every source sequence is encoded by its bin index, denoted by $\calB(\bu)$, along with a header that indicates its type index, which requires only a negligible extra rate when $n$ is large enough.

The SD code ensemble is a refinement of the ordinary VR code: 
for types with $H_{Q}(U) \geq \RF(Q_{U})$, i.e., type classes which are exponentially larger than the amount of available bins, we just randomly assign each source sequence into one out of the $e^{n \RF(Q_{U})}$ bins.  
For the other types, we deterministically order each member of $\calT(Q_{U})$ into a different bin. 
This way, all type classes with $H_{Q}(U) < \RF(Q_{U})$ do not contribute to the probability of error.
The entire binning code of source sequences of block--length $n$, i.e., the set $\{\calB(\bu)\}_{\bu \in \calU^{n}}$, is denoted by $\calB_{n}$. A sequence of SW codes, $\{\calB_{n}\}_{n \geq 1}$, indexed by the block length $n$, will be denoted by $\calB$. 

The decoder estimates $\bu$ based on the bin index $\calB(\bu)$, the type index $\calT(\bu)$, and the side information sequence $\bv$, which is a realization of $\bV$. The optimal (MAP) decoder estimates $\bu$ according to
\begin{align}
\hat{\bu} =  \operatorname*{arg\,max}_{\bu' \in \calB(\bu)\cap\calT(\bu)} P(\bu',\bv).
\end{align}
As in \cite{MERHAV2017}, \cite{MERHAV_TYPICAL}, we consider here the GLD. The GLD estimates $\bu$ stochastically, using the bin index $\calB(\bu)$, the type index $\calT(\bu)$, and the SI sequence $\bv$, according to the following posterior distribution
\begin{align}
\label{GLD_DEF}
\prob \left\{ \hat{\bU}=\bu' \middle| \bv, \calB(\bu), \calT(\bu) \right\}  =
\frac{\exp \{n f( \hat{P}_{\bu' \bv} )\}} {\sum_{\tilde{\bu} \in \calB(\bu)\cap\calT(\bu)}  
	\exp \{n f( \hat{P}_{\tilde{\bu} \bv} ) \} } ,
\end{align}
where $\hat{P}_{\bu \bv}$ is the empirical distribution of $(\bu, \bv)$ and $f(\cdot)$ is a given continuous, real valued functional of this empirical distribution.
The GLD provides a unified framework which covers several important special cases, e.g., matched decoding, mismatched decoding, MAP decoding, and universal decoding (similarly to the $\alpha$--decoders described in \cite{CKgraph}). 
A more detailed discussion is given in \cite{MERHAV2017}.

The probability of error is the probability of the event $\{\hat{\bU} \neq \bU\}$.
For a given binning code $\calB_{n}$, the probability of error is given by
\begin{align} \label{DEF_Pe_SD}
P_{\mbox{\tiny e}} (\calB_{n}) 
= \sum_{\bu,\bv} P(\bu,\bv) 
\cdot \IND \left\{ \hat{H}_{\bu}(U) \geq \RF(\hat{P}_{\bu}) \right\} \cdot
\frac{\sum_{\bu' \in \calB(\bu)\cap\calT(\bu), \bu' \neq \bu} \exp\{nf(\hat{P}_{\bu'\bv})\}}{\sum_{\tilde{\bu} \in \calB(\bu)\cap\calT(\bu)} \exp\{nf(\hat{P}_{\tilde{\bu}\bv})\}}.
\end{align}
For a given rate function, we derive the random binning exponent of this ensemble, which is defined by 
\begin{align} \label{FR_random_coding}
\EE_{\mbox{\tiny r}}(\RF(\cdot))
= \lim_{n \to \infty} \left\{ - \frac {\log \mathbb{E} [P_{\mbox{\tiny e}}(\calB_{n})] }{n} \right\},
\end{align}
and compare it to the TRC exponent, which is
\begin{align} \label{FR_TRC}
\EE_{\mbox{\tiny trc}}(\RF(\cdot))
= \lim_{n \to \infty} \left\{ - \frac {\mathbb{E} [\log P_{\mbox{\tiny e}}(\calB_{n})] }{n} \right\}.
\end{align}  
Although it is unclear that the limits in \eqref{FR_random_coding} and \eqref{FR_TRC} exist a priori, it will be evident from the analyses in Appendixes A and B, respectively.

One way to define the excess--rate probability is as $\prob \{\RF(\hat{P}_{\bU}) \geq \rr\}$, where $\rr$ is some target rate \cite{SW9}. Due to the availability of side information at the decoder, it makes sense to require a target rate which depends on the pair $(\bu,\bv)$. 
Since the lowest possible compression rate in this setting is given by $H_{P}(U|V)$ \cite{SW}, then, given $\bU=\bu$ and $\bV=\bv$, it is reasonable to adopt $\hat{H}_{\bu\bv}(U|V)$ as a reference rate.
Hence, an alternative definition of the excess--rate probability of a code $\calB_{n}$, is as $p_{\mbox{\tiny er}}(\calB_{n}, \RF(\cdot),\Delta) = \prob \{\RF(\hat{P}_{\bU}) \geq \hat{H}_{\bU\bV}(U|V) + \Delta\}$, where $\Delta > 0$ is a redundancy threshold\footnote{Note that the entire analysis remains intact if we allow a more general redundancy threshold as $\Delta=\Delta(\hat{P}_{\bu\bv})$. This covers other alternatives for the excess--rate probability, e.g., $\prob \{\RF(\hat{P}_{\bU}) \geq \rr\}$ or $\prob \{\RF(\hat{P}_{\bU}) \geq \alpha \hat{H}_{\bU}(U)\}$.}.
Accordingly, the excess--rate exponent function, achieved by a sequence of codes $\calB$, is defined as
\begin{align} \label{ExcessRateExponent}
\EE_{\mbox{\tiny er}}(\calB,\RF(\cdot),\Delta)
=\liminf_{n \to \infty} -\frac {1}{n} \log p_{\mbox{\tiny er}}(\calB_{n}, \RF(\cdot),\Delta).
\end{align}
The main mission is to characterize the optimal trade--off between the error exponent and the excess--rate exponent for the typical random SD code, and the optimal rate function that attains a prescribed value for the error exponent of the typical random SD code.

\subsection{Background} \label{SEC3.2}
In pure channel coding, 
Merhav \cite{MERHAV_TYPICAL} has derived a single--letter expression for the error exponent of the typical random fixed composition code,
\begin{align}
\EE_{\mbox{\tiny trc}}(R,Q_{X})
= \lim_{n \to \infty} \left\{- \tfrac{1}{n} \mathbb{E} \left[\log P_{\mbox{\tiny e}}(\calC_{n}) \right] \right\}.
\end{align}
In order to present the main result of \cite{MERHAV_TYPICAL}, we define first a few quantities. 
Consider a DMC, $W = \{W(y|x),x \in \calX, y \in \calY\}$, where $\calX$ and $\calY$ are the finite input/output alphabets.  
Define  
\begin{align}
\alpha(R,Q_{Y}) = \max_{\{Q_{\tilde{X}|Y}:~I_{Q}(\tilde{X};Y) \leq R,~ Q_{\tilde{X}}=Q_{X}\}} \{g(Q_{\tilde{X}Y}) - I_{Q}(\tilde{X};Y)\} + R,
\end{align}
where the function $g(\cdot)$, which is the decoding metric, is a continuous function that maps joint probability distributions over $\calX \times \calY$ to real numbers.
Also define 
\begin{align}
\Gamma(Q_{XX'},R) &= \min_{Q_{Y|XX'}} \{ D(Q_{Y|X} \| W |Q_{X}) + I_{Q}(X';Y|X) \nonumber \\
&+ [\max\{g(Q_{XY}), \alpha(R,Q_{Y})\} - g(Q_{X'Y})]_{+} \},
\end{align}
where $D(Q_{Y|X} \| W |Q_{X})$ is the conditional divergence between $Q_{Y|X}$ and $W$, averaged by $Q_{X}$.
A brief intuitive explanation on the term $\Gamma(Q_{XX'},R)$ can be found in \cite[Section 4.1]{TMWG}.
Having defined the above quantities, the error exponent of the TRC is given by \cite{MERHAV_TYPICAL}
\begin{align} \label{TRCexponent}
\EE_{\mbox{\tiny trc}}(R,Q_{X})
= \min_{\{Q_{X'|X}:~I_{Q}(X;X') \leq 2R, ~Q_{X'}=Q_{X}\}} \{\Gamma(Q_{XX'},R) + I_{Q}(X;X') - R\}. 
\end{align}

Returning to the SW model, several articles have been written on error exponents for the FR and the VR codes. Here, we mention only those results that are directly relevant to the current work.
The random binning and expurgated bounds of the FR ensemble in the SW model are given, respectively, by \cite[Sec.\ VI, Th.\ 2]{CKgraph}, \cite[Appendix I, Th.\ 1]{GoodCodes}
\begin{align} \label{SW_RC}
E_{\mbox{\tiny r}}^{\mbox{\tiny fr}}(\rr)
&= \min_{Q_{U}} \left\{D(Q_{U}\|P_{U}) + E_{\mbox{\tiny r}}(Q_{U},P_{V|U},H_{Q}(U)-\rr)  \right\}, \\ \label{SW_Expurgated}
E_{\mbox{\tiny ex}}^{\mbox{\tiny fr}}(\rr)
&= \min_{Q_{U}} \left\{D(Q_{U}\|P_{U}) + E_{\mbox{\tiny ex}}(Q_{U},P_{V|U},H_{Q}(U)-\rr)  \right\},
\end{align} 
where $E_{\mbox{\tiny r}}(Q_{U},P_{V|U},S)$ and $E_{\mbox{\tiny ex}}(Q_{U},P_{V|U},S)$ are, respectively, the random coding and expurgated bounds associated with the channel $P_{V|U}$ w.r.t.\ the ensemble of fixed composition code of rate $S$, 
whose composition is $Q_{U}$. The exponent function $E_{\mbox{\tiny r}}(Q_{U},P_{V|U},S)$ is given by 
\begin{align} \label{Channel_rc}
E_{\mbox{\tiny r}}(Q_{U},P_{V|U},S) = 
\min_{Q_{V|U}} \{D(Q_{V|U} \| P_{V|U} | Q_{U}) + [I_{Q}(U;V) - S]_{+} \},
\end{align}
and $E_{\mbox{\tiny ex}}(Q_{U},P_{V|U},S)$ is given by
\begin{align} \label{Channel_ex}
E_{\mbox{\tiny ex}}(Q_{U},P_{V|U},S)
= \min_{\{Q_{U'|U}:~I_{Q}(U;U')\leq S, ~Q_{U'}=Q_{U}\}} \{\mathbb{E}_{Q_{UU'}}[d_{P_{V|U}}(U,U')] + I_{Q}(U;U') - S\},
\end{align}  
where
\begin{align}
d_{P_{V|U}}(u,u') = - \log \left[\sum_{v \in \calV} 
\sqrt{P_{V|U}(v|u)P_{V|U}(v|u')}\right].
\end{align}
The exact error exponent of VR random binning is given by \cite[eq.\ (34)]{SW9}     
\begin{align} \label{SW_RC_VR}
E_{\mbox{\tiny r}}^{\mbox{\tiny vr}}(\RF(\cdot)) = \min_{Q_{UV}} \left\{D(Q_{UV} \| P_{UV}) +[\RF(Q_{U}) - H_{Q}(U|V)]_{+} \right\}.
\end{align}

\section{Error Exponents and Universal Decoding}
\label{Sec_SD_VR}

To present some of the results, we need a few more definitions. 
The minimum conditional entropy (MCE) decoder estimates $\bu$, using the bin index $\calB(\bu)$ and the SI vector $\bv$, according to
\begin{align}
\hat{\bu} =  \operatorname*{arg\,min}_{\bu' \in \calB(\bu)\cap\calT(\bu)} \hH_{\bu'\bv}(U|V).
\end{align}
The stochastic conditional entropy (SCE) decoder estimates $\bu$ according to the following posterior distribution
\begin{align}
\label{SCE_DEF}
\prob \left\{ \hat{\bU}=\bu' \middle| \bv, \calB(\bu) , \calT(\bu)\right\}  =
\frac{\exp \{-n \hH_{\bu'\bv}(U|V) \}} {\sum_{\tilde{\bu} \in \calB(\bu)\cap\calT(\bu)} \exp \{-n \hH_{\tilde{\bu}\bv}(U|V) \}}.
\end{align}

First, we present random binning error exponents, which are modifications of \eqref{SW_RC_VR} to this ensemble. Define the expression 
\begin{align}
E(Q_{UV},\RF(\cdot)) = \min_{Q_{U'|V}} \left[\RF(Q_{U}) - H_{Q}(U'|V) + [f(Q_{UV}) - f(Q_{U'V})]_{+} \right]_{+}
\end{align}
and the exponent functions: 
\begin{align} \label{SW_RC_VR_SD}
E_{\mbox{\tiny r,GLD}}(\RF(\cdot)) 
&= \min_{\{Q_{UV}:~ H_{Q}(U) \geq \RF(Q_{U})\}} \left\{D(Q_{UV} \| P_{UV}) + E(Q_{UV},\RF(\cdot)) \right\},
\end{align}
and
\begin{align} \label{SW_RC_VR_SD_MAP}
E_{\mbox{\tiny r,MAP}}(\RF(\cdot)) = \min_{\{Q_{UV}:~ H_{Q}(U) \geq \RF(Q_{U})\}} \left\{D(Q_{UV} \| P_{UV}) +[\RF(Q_{U}) - H_{Q}(U|V)]_{+} \right\}.
\end{align}

The following result is proved in Appendix A.
\begin{theorem} \label{Thm_random_binning_SD_VR}
	Let $\RF(\cdot)$ be a given rate function. Then, for the SD ensemble,
	\begin{enumerate}
		\item $\EE_{\mbox{\tiny r}}(\RF(\cdot)) = E_{\mbox{\tiny r,GLD}}(\RF(\cdot))$ for the GLD,
		\item $\EE_{\mbox{\tiny r}}(\RF(\cdot)) = E_{\mbox{\tiny r,MAP}}(\RF(\cdot))$ for the MAP and MCE decoders.
	\end{enumerate}	
\end{theorem}
As a matter of fact, a special case of the second part of Theorem \ref{Thm_random_binning_SD_VR} has already been proved in \cite{OH1994} for the FR regime, while here, we prove a stronger result, according to which, the MCE decoder attains the same random binning error exponent as the MAP decoder, in the VR coding regime too. The first part of Theorem \ref{Thm_random_binning_SD_VR} is completely new; it proposes a single letter expression for the random binning error exponent, for a wide family of stochastic and deterministic decoders. Also, note that an analogous result to the first part of Theorem \ref{Thm_random_binning_SD_VR} has been proved in \cite{MERHAV2017}.
Comparing the expressions in \eqref{SW_RC_VR} and \eqref{SW_RC_VR_SD_MAP}, namely, the random binning error exponents of the ordinary VR and the SD VR ensembles, respectively, we find that they differ at relatively high coding rates, since these minimization problems share the same objective but \eqref{SW_RC_VR_SD_MAP} also has the constraint $H_{Q}(U) \geq \RF(Q_{U})$. The origin of this constraint is the deterministic coding of the relatively small type classes. 

Next, we provide a single--letter expression for the error exponent of the TRCs in this ensemble. 
We define
\begin{align} \label{DEF_gamma}
\gamma(\RF(\cdot),Q_{U},Q_{V}) 
&= \max_{\left\{\substack{Q_{\tilde{U}|V}:~Q_{\tilde{U}}=Q_{U}, \\ H_{Q}(\tilde{U}|V) \geq \RF(Q_{\tilde{U}}) }\right\}} \{f(Q_{\tilde{U}V}) + H_{Q}(\tilde{U}|V)\} - \RF(Q_{\tilde{U}}) 
\end{align}
and
\begin{align} \label{DEF_Psi}
\Psi(\RF(\cdot),Q_{UU'V}) &= \left[\max\{f(Q_{UV}), \gamma(\RF(\cdot),Q_{U},Q_{V})\} - f(Q_{U'V}) \right]_{+}   .
\end{align}
Furthermore, define 
\begin{align} \label{DEF_Lambda}
\Lambda(Q_{UU'},\RF(Q_{U}))
= \min_{Q_{V|UU'}} \left\{ \Psi(\RF(Q_{U}),Q_{UU'V}) - H_{Q}(V|U,U') - \mathbb{E}_{Q} [\log P(V|U)] \right\},
\end{align}
and the following exponent function
\begin{align} \label{VR_SD_SW_EXPRESSION}
E_{\mbox{\tiny trc,GLD}}(\RF(\cdot)) =
\min_{\left\{\substack{Q_{UU'}:~Q_{U'}=Q_{U}, \\ H_{Q}(U) \geq \RF(Q_{U}) } \right\}} \left\{\Lambda(Q_{UU'},\RF(Q_{U})) - \mathbb{E}_{Q} [\log P(U)] - H_{Q}(U,U') + \RF(Q_{U}) \right\}.
\end{align}
Then, the following theorem is proved in Appendix B.
\begin{theorem} \label{SW_THM}
	Let $\RF(\cdot)$ be a given rate function. Then, for the SD ensemble,
	\begin{align}
	\EE_{\mbox{\tiny trc}}(\RF(\cdot))
	= E_{\mbox{\tiny trc,GLD}}(\RF(\cdot)).
	\end{align}	
\end{theorem}
As explained before, an analogous result has already been proved in pure channel coding \cite{MERHAV_TYPICAL}, and one can find a high degree of similarity between the expressions in \eqref{DEF_gamma}-\eqref{VR_SD_SW_EXPRESSION} and the expressions in Subsection \ref{SEC3.2}. 
While in channel coding, the coding rate is fixed, here, on the other hand, we allow the rate to depend on the type class of the source. In order to optimize the rate function, we constrain the problem by introducing the excess--rate exponent \eqref{ExcessRateExponent}, which is the exponential rate of decay of the probability that the compression rate will be higher than some predefined level. A detailed discussion on optimal rate functions and optimal trade--offs between these two exponents can be found in Section \ref{SEC5}.  

The definition of the error exponent of the TRC as in \eqref{FR_TRC} should not be taken for granted. The reason for that is the following. It turns out that the definition in \eqref{FR_TRC} and the value of $-\tfrac{1}{n} \log P_{\mbox{\tiny e}}(\calB_{n})$ for the highly probable codes in the ensemble may not be the same, and they coincide if and only if the ensemble does not contain both zero error probability codes and positive error probability codes.
For example, the FR ensemble in SW coding contains the one-to-one code (which obviously attains $P_{\mbox{\tiny e}}(\calB_{n})=0$) as long as $R \geq \log|\calU|$, but it is definitely not a typical code, at least when ordinary random binning is considered. Hence, in this case, we conclude that $-\tfrac{1}{n} \mathbb{E}[\log P_{\mbox{\tiny e}}(\calB_{n})]=\infty$, while the value of $-\tfrac{1}{n} \log P_{\mbox{\tiny e}}(\calB_{n})$ for the highly probable codes is still finite. 
As for the SD code ensemble, the definition in \eqref{FR_TRC} indeed provides the error exponent of the highly probable codes in the ensemble, which is explained by the following reasoning. 
For any given rate function such that $R(Q_{U}) < H_{Q}(U)$ for at least one type class, then all the type classes with $R(Q_{U}) < H_{Q}(U)$ are encoded by random binning, thus, all the codes in the ensemble have a strictly positive error probability, which implies that the value of $-\tfrac{1}{n} \log P_{\mbox{\tiny e}}(\calB_{n})$ concentrates around the error exponent of the TRC, as defined in \eqref{FR_TRC}. 

The proof of Theorem \ref{SW_THM} follows exactly the same lines as the proof of \cite[Theorem 1]{MERHAV_TYPICAL}, except for one main modification: when we introduce the type class enumerator $N(Q_{UU'})$ (see below) and sum over joint types, the summation set becomes $\{Q_{UU'}:~ Q_{U'}=Q_{U},~H_{Q}(U) \geq \RF(Q_{U})\}$, where the constraint $H_{Q}(U) \geq \RF(Q_{U})$ is due to the indicator function in \eqref{DEF_Pe_SD}. Afterwards, the analysis of the type class enumerator yields the constraint $H_{Q}(U,U') \geq \RF(Q_{U})$, which becomes redundant and thus omitted.      	
This constraint is analogous to the constraint $I_{Q}(X;X') \leq 2R$ in the minimization of \eqref{TRCexponent}. The origin of $H_{Q}(U,U') \geq R$ is the following. 
Define 
\begin{align} \label{DEF_ENUMERATOR}
N(Q_{UU'}) = \sum_{(\bu,\bu') \in \calT(Q_{UU'})} \IND \left\{\calB(\bu') = \calB(\bu) \right\},
\end{align}
which enumerate pairs of source sequences. 
Then, one of the main steps in the proof of Theorem \ref{SW_THM} is deriving the
high probability value of $N(Q_{UU'})$, which is $0$ if $H_{Q}(U,U') < R$ (a relatively small set of source pair and relatively large number of bins) and $\exp\{n[H_{Q}(U,U') - R]\}$ for $H_{Q}(U,U') \geq R$ (a large set of source sequence pair and a small number of bins). One should note that the analysis of $N(Q_{UU'})$ is not trivial, since it is not a binomial random variable, i.e., the enumerator $N(Q_{UU'})$ is given by the sum of dependent binary random variables. 
For a sum $N$ of independent binary random variables, ordinary tools from large deviation theory (e.g., the Chernoff bound) can be invoked for assessing the exponential moments $\mathbb{E} [N^{s}]$, $s \geq 0$, or the large deviation rate function of $\prob \{N \geq e^{n \sigma}\}$, $\sigma \in \reals$.   
For sums of dependent binary random variables, like $N(Q_{UU'})$ in the current problem, 
this can no longer be done by the same techniques, 
and it requires more advanced tools (see, e.g., \cite{MERHAV_TYPICAL}--\cite{MERHAV_IID}). 

It is possible to compare \eqref{SW_RC_VR_SD} and \eqref{VR_SD_SW_EXPRESSION} analytically in the special cases of the matched or the mismatched likelihood decoders and the MCE decoder. 
In the following theorem, the choice $f(Q_{UV}) = \beta \mathbb{E}_{Q} [\log \tilde{P}(U,V)]$, where $\tilde{P}(U,V)$ is a possibly different source distribution than $P(U,V)$, corresponds to a family of stochastic mismatched decoders.
We have the following result, the proof of which is given in Appendix D.
\begin{theorem} \label{Thm_comparison}
	Consider the SD ensemble and a given rate function $\RF(\cdot)$. Then,
	\begin{enumerate}
		\item For a GLD with the decoding metric $f(Q) = \beta \mathbb{E}_{Q} [\log \tilde{P}(U,V)]$, for a given $\beta > 0$,
		\begin{align}
		E_{\mbox{\tiny trc,GLD}}(\RF(\cdot)) = E_{\mbox{\tiny r,GLD}}(\RF(\cdot)).
		\end{align}
		\item For the MCE decoder,
		\begin{align}
		E_{\mbox{\tiny trc,MCE}}(\RF(\cdot)) = E_{\mbox{\tiny r,MCE}}(\RF(\cdot)).
		\end{align} 
	\end{enumerate} 	
\end{theorem}

This result is quite surprising at first glance, since one expects the error exponent of the TRC to be strictly better than the random binning error exponent, as in ordinary channel coding at relatively low coding rates \cite{BargForney}, \cite{MERHAV_TYPICAL}. This phenomenon is due to the fact that part of the source type classes are deterministically partitioned into bins in a one--to--one fashion, and hence do not affect the probability of error (notice that the constraint $H_{Q}(U) \geq \RF(Q_{U})$ appears in both the random binning and the TRC exponents, while in the latter, it made the original constraint $H_{Q}(U,U') \geq \RF(Q_{U})$ redundant). In the cases of FR or ordinary VR binning, these relatively ``thin'' type classes dominated the error probability at relatively high binning rates, but now, by encoding them deterministically into the bins, 
other mechanisms dominate the error event, like the channel noise (between $\bU$ and $\bV$) or the random binning of the type classes with $H_{Q}(U) \geq \RF(Q_{U})$. 
The result of the second part of Theorem \ref{Thm_comparison} is also nontrivial, since it establishes an equality between the error exponent of the TRC and the random binning error exponent, but now for a universal decoder.

Concerning universal decoding, it is already known \cite[Exercise 3.1.6]{CK11}, \cite{SW9} that the random binning error exponents under optimal MAP decoding in both the FR and VR codes, given by \eqref{SW_RC} and \eqref{SW_RC_VR}, respectively, are also attained by the MCE decoder. Furthermore, a similar result for the SD ensemble has been proved here in Theorem \ref{Thm_random_binning_SD_VR}. 
The natural question that arises is whether the error exponent of the TRC is also universally attainable. The following result, which is proved in Appendix E, provides a positive answer to this question.  

\begin{theorem} \label{Thm_Universal}
	Consider the SD ensemble and a given rate function $\RF(\cdot)$. Then, the error exponents of the TRC under the MAP, the MCE, and the SCE decoders are all equal, i.e.,
	\begin{align} \label{3KINGS}
	E_{\mbox{\tiny trc,MAP}}(\RF(\cdot)) 
	= E_{\mbox{\tiny trc,MCE}}(\RF(\cdot))
	= E_{\mbox{\tiny trc,SCE}}(\RF(\cdot)).
	\end{align} 	
\end{theorem}
Theorem \ref{Thm_Universal} asserts that the error exponent of the typical random SD code is not affected if the optimal MAP decoder is replaced by a certain universal decoder, that must not even be deterministic.
While the left hand equality in \eqref{3KINGS} follows immediately from the results of Theorems \ref{Thm_random_binning_SD_VR} and \ref{Thm_comparison}, the right hand equality in \eqref{3KINGS} is far less trivial, since the SCE decoder is both universal and stochastic, and hence, its TRC exponent is expected to be inferior w.r.t.\ the TRC exponent under MAP decoding, but nevertheless, they turn out to be equal.    
Comparing to channel coding, 
it has been recently proved in \cite{TM_MMI} that 
the error exponent of the typical random fixed composition code (given in \eqref{TRCexponent}) is the same for the ML and the maximum mutual information decoder, but on the other hand, numerical evidence shows that a GLD which is based on an empirical mutual information metric attains a strictly lower exponent. 

\section{Optimal Trade-off Functions} \label{SEC5} 
In this section, we study the optimal trade--off between the threshold $\Delta$, the error exponent of the TRC, and the excess--rate exponent. Since both exponents depend on the rate function, we wish to characterize rate functions that are optimal w.r.t.\ this trade--off.
Since a single--letter characterization of the error exponent of the TRC has already been given in \eqref{VR_SD_SW_EXPRESSION}, we next provide a single-letter expression for the excess--rate exponent. Define the following exponent function:
\begin{align} \label{Excess_exponent}
E_{\mbox{\tiny er}}(\RF(\cdot),\Delta) = \min_{\{Q_{UV}:~ \RF(Q_{U}) \geq H_{Q}(U|V) + \Delta\}} D(Q_{UV} \| P_{UV}).
\end{align}
Then, we have the following.
\begin{proposition} \label{Propy1}
	Fix $\Delta>0$ and let $\RF(\cdot)$ be any rate function. Then,
	\begin{align}
	\EE_{\mbox{\tiny er}}(\calB,\RF(\cdot),\Delta)
	= E_{\mbox{\tiny er}}(\RF(\cdot),\Delta).
	\end{align}	
\end{proposition}

\begin{proof}
	The excess--rate probability is given by:
	\begin{align}
	&\prob \{\RF(\hat{P}_{\bU}) \geq \hat{H}_{\bU\bV}(U|V) + \Delta\} \nn \\
	&= \sum_{Q_{UV}} \IND \{\RF(Q_{U}) \geq H_{Q}(U|V) + \Delta \} \cdot \prob \{(\bU,\bV) \in \calT(Q_{UV})\} \\
	&\doteq \sum_{\{Q_{UV}:~ \RF(Q_{U}) \geq H_{Q}(U|V) + \Delta\}} \exp \left\{-n D(Q_{UV} \| P_{UV}) \right\}\\
	&\doteq \exp \left\{-n \cdot \min_{\{Q_{UV}:~ \RF(Q_{U}) \geq H_{Q}(U|V) + \Delta\}} D(Q_{UV} \| P_{UV}) \right\},
	\end{align}
	which proves the desired result.	
\end{proof}

Since Proposition \ref{Propy1} is proved by the method of types \cite{CK11}, we conclude that the excess--rate event is dominated by one specific type class $\calT(Q_{UV})$, whose respective rate $\RF(Q_{U})$ has been chosen too large w.r.t.\ the value of $H_{Q}(U|V) + \Delta$. One extreme case is when the rate function is given by $H_{Q}(U)$, which obviously provides a one-to-one mapping, since the size of each $\calT(Q_{U})$ is upper-bounded by $e^{nH_{Q}(U)}$. In this case, the probability of error is zero, while the excess--rate probability is one, at least when $\Delta$ is not too large. In Subsection \ref{SEC5.2}, we prove that the optimal rate function is indeed upper-bounded by $H_{Q}(U)$, but can also be strictly smaller, especially when the requirement on the error exponent is not too stringent.  
	
One way to explore the trade--off between the error exponent of the TRC and the excess--rate exponent, that will be presented in Subsection \ref{SEC5.1}, is to require the excess--rate exponent to exceed some value $\er>0$, then solve $E_{\mbox{\tiny er}}(\RF(\cdot),\Delta) \geq \er$ for an optimal rate function $\RF^{*}(Q_{U})$, and then to substitute this optimal rate function back into the error exponents in \eqref{SW_RC_VR_SD_MAP} and \eqref{VR_SD_SW_EXPRESSION} to give expressions for the optimal trade--off function $E_{\mbox{\tiny e}}(\er,\Delta)$.
In Subsection \ref{SEC5.2}, we present an alternative option to characterize this trade--off, which is to require the error exponent of the TRC to exceed some value $\ee>0$, to solve $E_{\mbox{\tiny e}}(\RF(\cdot)) \geq \ee$ in order to extract an optimal rate function, and then to substitute it back into the excess--rate exponent in \eqref{Excess_exponent} to provide an expression for the optimal trade--off function $E_{\mbox{\tiny er}}(\ee,\Delta)$.

\subsection{Constrained Excess--Rate Exponent}
\label{SEC5.1}

Relying on the exponent function in \eqref{Excess_exponent}, the following theorem proposes a rate function, whose optimality is proved in Appendix F.

\begin{theorem} \label{RF_UB}
	Let $\er>0$ be fixed. Then, the constraint $E_{\mbox{\tiny er}}(\RF(\cdot),\Delta) \geq \er$ implies that
	\begin{align} \label{J_DEF}
	\RF(Q_{U}) \leq J(Q_{U},\er,\Delta) \dfn  \min_{\{Q_{V|U}:~D(Q_{UV} \| P_{UV}) \leq \er\}} \left\{ H_{Q}(U|V) + \Delta \right\} .
	\end{align} 
\end{theorem} 
This means that we have a dichotomy between two kinds of source types. 
Each type class that is associated with an empirical distribution that is relatively close to the source distribution, i.e., when $D(Q_{UV} \| P_{UV}) \leq \er$ for some $Q_{V|U}$, is partitioned into $e^{n J(Q_{U},\er,\Delta)}$ bins, and the rest of the type classes, those that are relatively distant from $P_{U}$, are encoded by a one--to--one mapping. Two extreme cases should be considered here. First, when $\er$ is relatively small, then only the types closest to $P_{U}$ are encoded with a rate approximately $H_{P}(U|V)+\Delta$, which can be made arbitrarily close to the SW limit \cite{SW}, 
and each a--typical source sequence is allocated with $n \cdot \log_{2} |\calU|$ bits. This coding scheme is the one related to VR coding with an average rate constraint, like the one discussed in \cite{SW2}.  
Second, when $\er$ is extremely large, then each type class is encoded to $\exp\{n \Delta\}$ bins, which is equivalent to FR coding.

Following the first part of Theorem \ref{Thm_comparison}, let us denote the error exponent of the TRC under MAP decoding by $E_{\mbox{\tiny e}}(\cdot)$. 
Upon substituting the optimal rate function of Theorem \ref{RF_UB} back into \eqref{SW_RC_VR_SD_MAP} and \eqref{VR_SD_SW_EXPRESSION} and using the fact that $E_{\mbox{\tiny e}}(\cdot)$ is monotonically increasing, we find that the optimal trade--off function for the typical random SD code is given by
\begin{align} \label{TradeOff_EXPRESSION_SD_r}
E_{\mbox{\tiny e}}(\er,\Delta) = \min_{\left\{Q_{UV}:~ H_{Q}(U) \geq J(Q_{U}) \right\}} \left\{D(Q_{UV} \| P_{UV}) +[J(Q_{U}) - H_{Q}(U|V)]_{+} \right\},
\end{align}
or, alternatively,
\begin{align} \label{TradeOff_EXPRESSION_SD}
E_{\mbox{\tiny e}}(\er,\Delta) 
&= \min_{\left\{\substack{Q_{UU'}:~Q_{U'}=Q_{U}, \\ H_{Q}(U) \geq J(Q_{U})}\right\}} \left\{\Lambda(Q_{UU'},J(Q_{U})) - \mathbb{E}_{Q} [\log P(U)] - H_{Q}(U,U') + J(Q_{U}) \right\},
\end{align}
where $J(Q_{U}) = J(Q_{U},\er,\Delta)$ is given in \eqref{J_DEF}.
The dependence of $E_{\mbox{\tiny e}}(\er,\Delta)$ on $\er$ is as follows. 
Let $Q_{UU'}^{*}(\Delta)$ and $Q_{V|U}^{*}$ be the respective minimizers of the problems which are similar to \eqref{J_DEF} and \eqref{TradeOff_EXPRESSION_SD}, except that the constraint $D(Q_{UV} \| P_{UV}) \leq \er$ is removed from \eqref{J_DEF}. Furthermore, let $Q_{U}^{*}(\Delta)$ be the marginal distribution of $Q_{UU'}^{*}(\Delta)$.
Now, when $\er$ is sufficiently large, i.e., when $\er \geq D(Q_{U}^{*}(\Delta) \times Q_{V|U}^{*} \| P_{UV})$, $E_{\mbox{\tiny e}}(\er,\Delta)$ reaches a plateau and is the lowest possible. It follows from the fact that the stringent requirement on the excess--rate forces the encoder to encode each type class $Q_{U}$ to its target rate $\Delta$, thus all of them affect the error event. 
Otherwise, when $\er < D(Q_{U}^{*}(\Delta) \times Q_{V|U}^{*} \| P_{UV})$, the constraint $D(Q_{UV} \| P_{UV}) \leq \er$ is active and $E_{\mbox{\tiny e}}(\er,\Delta)$ is a monotonically non--increasing function of $\er$. 
The reason for that is the fact that as $\er$ decreases, more and more type classes are encoded with $n \cdot \log_{2} |\calU|$ bits, and hence do not contribute to the error event. 
When $\er=0$, necessarily $Q_{U}=P_{U}$, only the typical set is encoded, and $E_{\mbox{\tiny e}}(0,\Delta)$ is the highest possible. In this case, $J(Q_{U}) = H_{P}(U|V) + \Delta$ and the constraint set in \eqref{TradeOff_EXPRESSION_SD} becomes empty when $\Delta > I_{P}(U;V)$, and then $E_{\mbox{\tiny e}}(0,\Delta) = \infty$.

\subsection{Constrained Error Exponent} \label{SEC5.2}

Based on \eqref{SW_RC_VR_SD_MAP}, the following theorem proposes a rate function, whose optimality is proved in Appendix G.

\begin{theorem} \label{RF_LB_RC23}
	Let $\ee>0$ be fixed. Then, the constraint $E_{\mbox{\tiny e}}(\RF(\cdot)) \geq \ee$ implies that
	\begin{align} \label{RF_LB_RC23_expression}
	\RF(Q_{U}) \geq \Omega(Q_{U},\ee) 
	\DEF \min \left\{ H_{Q}(U), G(Q_{U},\ee) \right\},
	\end{align} 
	where,
	\begin{align} \label{RF_LB_RC_expression}
	G(Q_{U},\ee) = \max_{\{Q_{V|U}:~ D(Q_{UV} \| P_{UV}) \leq \ee\}} \{ H_{Q}(U|V) + \ee - D(Q_{UV} \| P_{UV})\}.
	\end{align}
\end{theorem}
The dependence of $G(Q_{U},\ee)$ on $\ee$ is as follows. 
For any given $Q_{U}$, let $\tilde{Q}_{V|U}$ be the minimizer of $D(Q_{UV} \| P_{UV})$. Then, as long as $\ee < D(Q_{U} \times \tilde{Q}_{V|U} \| P_{UV})$, the constraint set in  
\eqref{RF_LB_RC_expression} is empty, and $\RF(Q_{U})$ can vanish, which practically means that in this range, the entire type class $\calT(Q_{U})$ can be totally ignored, while still achieving $P_{\mbox{\tiny e}} \approx e^{-n \ee}$. 
Only for the unique type $Q_{U}=P_{U}$, $G(P_{U},\ee) > 0$ for all $\ee \geq 0$, and specifically, we find that $G(P_{U},0) = H_{P}(U|V)$.    
Furthermore, let $Q_{V|U}^{*}$ be the maximizer in the unconstrained problem 
\begin{align} \label{Constant1}
\max_{Q_{V|U}} \left\{H_{Q}(U|V) - D(Q_{UV} \| P_{UV})\right\}.
\end{align}
Then, as long as $\ee \in [D(Q_{U} \times \tilde{Q}_{V|U} \| P_{UV}), D(Q_{U} \times Q_{V|U}^{*} \| P_{UV}))$, $G(Q_{U},\ee)$ is a monotonically non--decreasing function of $\ee$. When $\ee \geq D(Q_{U} \times Q_{V|U}^{*} \| P_{UV})$, the maximization in \eqref{RF_LB_RC_expression} reaches its unconstrained optimum, and $G(Q_{U},\ee)$ increases without bound in an affine fashion as $\ee + H_{Q^{*}}(U|V) - D(Q_{U} \times Q_{V|U}^{*} \| P_{UV})$. As can be seen in \eqref{RF_LB_RC23_expression}, $\Omega(Q_{U},\ee)$ finally reaches a plateau at the level of $H_{Q}(U)$.

Upon substituting $\Omega(Q_{U},\ee)$ back into \eqref{Excess_exponent} and using the fact that $E_{\mbox{\tiny er}}(\cdot,\Delta)$ is monotonically non--increasing, we find that the trade--off function is given by
\begin{align} \label{TradeOff_EXPRESSION3}
E_{\mbox{\tiny er}}(\ee,\Delta) = \min_{\{Q_{UV}:~ \Omega(Q_{U},\ee) \geq H_{Q}(U|V) + \Delta\}} D(Q_{UV} \| P_{UV}).
\end{align}
Since $\Omega(Q_{U},\ee)$ is monotonically non--decreasing in $\ee$ for every $Q_{U}$, $E_{\mbox{\tiny er}}(\ee,\Delta)$ is monotonically non--increasing in $\ee$, which is not very surprising. The dependence of $E_{\mbox{\tiny er}}(\ee,\Delta)$ on $\ee$ and $\Delta$ is as follows. At $\ee=0$, notice that $\Omega(Q_{U},0) = -\infty$ for any $Q_{U} \neq P_{U}$ while $\Omega(P_{U},0) = H_{P}(U|V)$.
Thus, $E_{\mbox{\tiny er}}(0,\Delta) = 0$ as long as $\Delta=0$, and it follows from the monotonicity that $E_{\mbox{\tiny er}}(\ee,0) = 0$ everywhere.  
Otherwise, if $\Delta > 0$, $\{Q_{UV}:~ \Omega(Q_{U},\ee) \geq H_{Q}(U|V)+\Delta\}$ is empty as long as $\ee < \ee^{*}(\Delta)$, and then $E_{\mbox{\tiny er}}(\ee,\Delta) = \infty$ in this range\footnote{An expression for $\ee^{*}(\Delta)$ can be found by solving $\max_{Q_{UV}} \{\Omega(Q_{U},\ee) - H_{Q}(U|V)\} \leq \Delta$.}. 
In the other extreme case of a very large $\ee$, $\Omega(Q_{U},\ee)$ reaches a plateau at a level of $H_{Q}(U)$. 
Then, if $\Delta \leq H_{P}(U) - H_{P}(U|V) = I_{P}(U;V)$, $E_{\mbox{\tiny er}}(\ee,\Delta)$ reaches zero for a sufficiently large $\ee$. Else, if $\Delta > I_{P}(U;V)$, $E_{\mbox{\tiny er}}(\ee,\Delta)$ reaches a strictly positive plateau, given by 
\begin{align}
\min_{\{Q_{UV}:~ I_{Q}(U;V) \geq \Delta\}} D(Q_{UV} \| P_{UV}),
\end{align}  
which is a monotonically non--decreasing function of $\Delta$.
Particularly, it means that in this range, the typical random SD code attains both an exponentially vanishing excess--rate probability as well as $P_{\mbox{\tiny e}} \approx 0$.

It is interesting to relate this to the expurgated bound of the FR code in the SW model, which is given by \eqref{SW_Expurgated}.  
Comparing $E_{\mbox{\tiny ex}}^{\mbox{\tiny fr}}(\rr)$ and $E_{\mbox{\tiny e}}(\infty,\Delta)$ analytically is rather difficult. 
Thus, we examined these two exponent functions numerically. Consider the case of a double binary source with alphabets $\calU=\calV=\{0,1\}$, and joint probabilities given by $P_{UV}(0,0)=0.75$, $P_{UV}(0,1)=0.1$, $P_{UV}(1,0)=0$, and $P_{UV}(1,1)=0.15$. 
We already mentioned before, that in the special case of $\er = \infty$, the rate function is given by the threshold $\Delta$, hence we choose $\Delta = \rr$ in order to have a fair comparison.
Graphs of the functions $E_{\mbox{\tiny ex}}^{\mbox{\tiny fr}}(\rr)$ and $E_{\mbox{\tiny e}}(\infty,\rr)$ are presented in Fig.\ \ref{Z-Channel-Numeric4}.

\begin{figure}[ht!]
	\centering
	\begin{tikzpicture}[scale=1.3]
	\draw[->] (0,0) -- (7,0) node[right] {$\rr$};
	\begin{axis}[
	disabledatascaling,
	scaled y ticks=false,
	yticklabel style={/pgf/number format/fixed,
		/pgf/number format/precision=3},
	xmin=0, xmax=0.7,
	ymin=0, ymax=0.9,
	legend pos=south east,
	]
	
	\addplot[smooth,color=cyan,thick]
	table[row sep=crcr] 
	{
0	4.26E-12	\\
0.001	4.26E-12	\\
0.002	4.26E-12	\\
0.003	4.26E-12	\\
0.004	4.26E-12	\\
0.005	4.26E-12	\\
0.006	4.26E-12	\\
0.007	4.26E-12	\\
0.008	4.26E-12	\\
0.009	4.26E-12	\\
0.01	4.26E-12	\\
0.011	4.26E-12	\\
0.012	4.26E-12	\\
0.013	4.26E-12	\\
0.014	4.26E-12	\\
0.015	4.26E-12	\\
0.016	4.26E-12	\\
0.017	4.26E-12	\\
0.018	4.26E-12	\\
0.019	4.26E-12	\\
0.02	4.26E-12	\\
0.021	4.26E-12	\\
0.022	4.26E-12	\\
0.023	4.26E-12	\\
0.024	4.26E-12	\\
0.025	4.26E-12	\\
0.026	4.26E-12	\\
0.027	4.26E-12	\\
0.028	4.26E-12	\\
0.029	4.26E-12	\\
0.03	4.26E-12	\\
0.031	4.26E-12	\\
0.032	4.26E-12	\\
0.033	4.26E-12	\\
0.034	4.26E-12	\\
0.035	4.26E-12	\\
0.036	4.26E-12	\\
0.037	4.26E-12	\\
0.038	4.26E-12	\\
0.039	4.26E-12	\\
0.04	4.26E-12	\\
0.041	4.26E-12	\\
0.042	4.26E-12	\\
0.043	4.26E-12	\\
0.044	4.26E-12	\\
0.045	4.26E-12	\\
0.046	4.26E-12	\\
0.047	4.26E-12	\\
0.048	4.26E-12	\\
0.049	4.26E-12	\\
0.05	4.26E-12	\\
0.051	4.26E-12	\\
0.052	4.26E-12	\\
0.053	4.26E-12	\\
0.054	4.26E-12	\\
0.055	4.26E-12	\\
0.056	4.26E-12	\\
0.057	4.26E-12	\\
0.058	4.26E-12	\\
0.059	4.26E-12	\\
0.06	4.26E-12	\\
0.061	4.26E-12	\\
0.062	4.26E-12	\\
0.063	4.26E-12	\\
0.064	4.26E-12	\\
0.065	4.26E-12	\\
0.066	4.26E-12	\\
0.067	4.26E-12	\\
0.068	4.26E-12	\\
0.069	4.26E-12	\\
0.07	4.26E-12	\\
0.071	4.26E-12	\\
0.072	4.26E-12	\\
0.073	4.26E-12	\\
0.074	4.26E-12	\\
0.075	4.26E-12	\\
0.076	4.26E-12	\\
0.077	4.26E-12	\\
0.078	4.26E-12	\\
0.079	4.26E-12	\\
0.08	4.26E-12	\\
0.081	4.26E-12	\\
0.082	4.26E-12	\\
0.083	4.26E-12	\\
0.084	4.26E-12	\\
0.085	4.26E-12	\\
0.086	4.26E-12	\\
0.087	4.26E-12	\\
0.088	4.26E-12	\\
0.089	4.26E-12	\\
0.09	4.26E-12	\\
0.091	4.26E-12	\\
0.092	4.26E-12	\\
0.093	4.26E-12	\\
0.094	4.26E-12	\\
0.095	4.26E-12	\\
0.096	4.26E-12	\\
0.097	4.26E-12	\\
0.098	4.26E-12	\\
0.099	4.26E-12	\\
0.1	4.26E-12	\\
0.101	4.26E-12	\\
0.102	4.26E-12	\\
0.103	4.26E-12	\\
0.104	4.26E-12	\\
0.105	4.26E-12	\\
0.106	4.26E-12	\\
0.107	4.26E-12	\\
0.108	4.26E-12	\\
0.109	4.26E-12	\\
0.11	4.26E-12	\\
0.111	4.26E-12	\\
0.112	4.26E-12	\\
0.113	4.26E-12	\\
0.114	4.26E-12	\\
0.115	4.26E-12	\\
0.116	4.26E-12	\\
0.117	4.26E-12	\\
0.118	4.26E-12	\\
0.119	4.26E-12	\\
0.12	4.26E-12	\\
0.121	4.26E-12	\\
0.122	4.26E-12	\\
0.123	4.26E-12	\\
0.124	4.26E-12	\\
0.125	4.26E-12	\\
0.126	4.26E-12	\\
0.127	4.26E-12	\\
0.128	4.26E-12	\\
0.129	4.26E-12	\\
0.13	4.26E-12	\\
0.131	4.26E-12	\\
0.132	4.26E-12	\\
0.133	4.26E-12	\\
0.134	4.26E-12	\\
0.135	4.26E-12	\\
0.136	4.26E-12	\\
0.137	4.26E-12	\\
0.138	4.26E-12	\\
0.139	4.26E-12	\\
0.14	4.26E-12	\\
0.141	4.26E-12	\\
0.142	4.26E-12	\\
0.143	4.26E-12	\\
0.144	4.26E-12	\\
0.145	4.26E-12	\\
0.146	4.26E-12	\\
0.147	4.26E-12	\\
0.148	4.26E-12	\\
0.149	4.26E-12	\\
0.15	4.26E-12	\\
0.151	4.26E-12	\\
0.152	4.26E-12	\\
0.153	4.26E-12	\\
0.154	4.26E-12	\\
0.155	4.26E-12	\\
0.156	4.26E-12	\\
0.157	4.26E-12	\\
0.158	4.26E-12	\\
0.159	4.26E-12	\\
0.16	4.26E-12	\\
0.161	4.26E-12	\\
0.162	4.26E-12	\\
0.163	4.26E-12	\\
0.164	4.26E-12	\\
0.165	4.26E-12	\\
0.166	4.26E-12	\\
0.167	4.26E-12	\\
0.168	4.26E-12	\\
0.169	2.94E-06	\\
0.17	1.61E-05	\\
0.171	3.97E-05	\\
0.172	7.40E-05	\\
0.173	0.000118577	\\
0.174	0.000173889	\\
0.175	0.000239198	\\
0.176	0.00031516	\\
0.177	0.00040153	\\
0.178	0.000498345	\\
0.179	0.00060542	\\
0.18	0.000722851	\\
0.181	0.00085056	\\
0.182	0.000988613	\\
0.183	0.001136908	\\
0.184	0.00129547	\\
0.185	0.001464955	\\
0.186	0.001645872	\\
0.187	0.001832836	\\
0.188	0.002031713	\\
0.189	0.002241168	\\
0.19	0.002460849	\\
0.191	0.002692088	\\
0.192	0.002930345	\\
0.193	0.003179888	\\
0.194	0.003439909	\\
0.195	0.003711499	\\
0.196	0.003989243	\\
0.197	0.004279553	\\
0.198	0.004578607	\\
0.199	0.004888215	\\
0.2	0.005208674	\\
0.201	0.005536919	\\
0.202	0.005876229	\\
0.203	0.006228565	\\
0.204	0.006584154	\\
0.205	0.006952369	\\
0.206	0.00733108	\\
0.207	0.007718634	\\
0.208	0.008116549	\\
0.209	0.008523761	\\
0.21	0.008940906	\\
0.211	0.009368427	\\
0.212	0.009806377	\\
0.213	0.010250847	\\
0.214	0.010705542	\\
0.215	0.011170618	\\
0.216	0.011647499	\\
0.217	0.012130171	\\
0.218	0.012623188	\\
0.219	0.013126332	\\
0.22	0.013640774	\\
0.221	0.014161483	\\
0.222	0.014693034	\\
0.223	0.015233637	\\
0.224	0.015784891	\\
0.225	0.016343534	\\
0.226	0.016914501	\\
0.227	0.017490935	\\
0.228	0.018080598	\\
0.229	0.018675615	\\
0.23	0.019281896	\\
0.231	0.019899175	\\
0.232	0.020522207	\\
0.233	0.021156179	\\
0.234	0.021799391	\\
0.235	0.022454326	\\
0.236	0.023117005	\\
0.237	0.023785236	\\
0.238	0.024464621	\\
0.239	0.025153211	\\
0.24	0.025851167	\\
0.241	0.026558373	\\
0.242	0.027274677	\\
0.243	0.02800073	\\
0.244	0.02873476	\\
0.245	0.02948749	\\
0.246	0.030230424	\\
0.247	0.030991925	\\
0.248	0.031762325	\\
0.249	0.032541779	\\
0.25	0.033330957	\\
0.251	0.034127533	\\
0.252	0.034933824	\\
0.253	0.035750392	\\
0.254	0.036573141	\\
0.255	0.037406413	\\
0.256	0.038248019	\\
0.257	0.039098991	\\
0.258	0.039958375	\\
0.259	0.040828558	\\
0.26	0.041704097	\\
0.261	0.042590197	\\
0.262	0.04348511	\\
0.263	0.044388801	\\
0.264	0.045301248	\\
0.265	0.046222501	\\
0.266	0.047152487	\\
0.267	0.048091249	\\
0.268	0.049038723	\\
0.269	0.049994991	\\
0.27	0.050959827	\\
0.271	0.051933431	\\
0.272	0.052915739	\\
0.273	0.053906717	\\
0.274	0.054905455	\\
0.275	0.055905455	\\
0.276	0.056905455	\\
0.277	0.057905455	\\
0.278	0.058905455	\\
0.279	0.059905455	\\
0.28	0.060905455	\\
0.281	0.061905455	\\
0.282	0.062905455	\\
0.283	0.063905455	\\
0.284	0.064905455	\\
0.285	0.065905455	\\
0.286	0.066905455	\\
0.287	0.067905455	\\
0.288	0.068905455	\\
0.289	0.069905455	\\
0.29	0.070905455	\\
0.291	0.071905455	\\
0.292	0.072905455	\\
0.293	0.073905455	\\
0.294	0.074905455	\\
0.295	0.075905455	\\
0.296	0.076905455	\\
0.297	0.077905455	\\
0.298	0.078905455	\\
0.299	0.079905455	\\
0.3	0.080905455	\\
0.301	0.081905455	\\
0.302	0.082905455	\\
0.303	0.083905455	\\
0.304	0.084905455	\\
0.305	0.085905455	\\
0.306	0.086905455	\\
0.307	0.087905455	\\
0.308	0.088905455	\\
0.309	0.089905455	\\
0.31	0.090905455	\\
0.311	0.091905455	\\
0.312	0.092905455	\\
0.313	0.093905455	\\
0.314	0.094905455	\\
0.315	0.095905455	\\
0.316	0.096905455	\\
0.317	0.097905455	\\
0.318	0.098905455	\\
0.319	0.099905455	\\
0.32	0.100905455	\\
0.321	0.101905455	\\
0.322	0.102905455	\\
0.323	0.103905455	\\
0.324	0.104905455	\\
0.325	0.105905455	\\
0.326	0.106905455	\\
0.327	0.107905455	\\
0.328	0.108905455	\\
0.329	0.109905455	\\
0.33	0.110905455	\\
0.331	0.111905455	\\
0.332	0.112905455	\\
0.333	0.113905455	\\
0.334	0.114905455	\\
0.335	0.115905455	\\
0.336	0.116905455	\\
0.337	0.117905455	\\
0.338	0.118905455	\\
0.339	0.119905455	\\
0.34	0.120905455	\\
0.341	0.121905455	\\
0.342	0.122905455	\\
0.343	0.123905455	\\
0.344	0.124905455	\\
0.345	0.125905455	\\
0.346	0.126905455	\\
0.347	0.127905455	\\
0.348	0.128905455	\\
0.349	0.129905455	\\
0.35	0.130905455	\\
0.351	0.131905455	\\
0.352	0.132905455	\\
0.353	0.133905455	\\
0.354	0.134905455	\\
0.355	0.135905455	\\
0.356	0.136905455	\\
0.357	0.137905455	\\
0.358	0.138905455	\\
0.359	0.139905455	\\
0.36	0.140905455	\\
0.361	0.141905455	\\
0.362	0.142905455	\\
0.363	0.143905455	\\
0.364	0.144905455	\\
0.365	0.145905455	\\
0.366	0.146905455	\\
0.367	0.147905455	\\
0.368	0.148905455	\\
0.369	0.149905455	\\
0.37	0.150905455	\\
0.371	0.151905455	\\
0.372	0.152905455	\\
0.373	0.153905455	\\
0.374	0.154905455	\\
0.375	0.155905455	\\
0.376	0.156905455	\\
0.377	0.157905455	\\
0.378	0.158905455	\\
0.379	0.159905455	\\
0.38	0.160905455	\\
0.381	0.161905455	\\
0.382	0.162905455	\\
0.383	0.163905455	\\
0.384	0.164905455	\\
0.385	0.165905455	\\
0.386	0.166905455	\\
0.387	0.167905455	\\
0.388	0.168905455	\\
0.389	0.169905455	\\
0.39	0.170905455	\\
0.391	0.171905455	\\
0.392	0.172905455	\\
0.393	0.173905455	\\
0.394	0.174905455	\\
0.395	0.175905455	\\
0.396	0.176905455	\\
0.397	0.177905455	\\
0.398	0.178905455	\\
0.399	0.179905455	\\
0.4	0.180905455	\\
0.401	0.181905455	\\
0.402	0.182905455	\\
0.403	0.183905455	\\
0.404	0.184905455	\\
0.405	0.185905455	\\
0.406	0.186905455	\\
0.407	0.187905455	\\
0.408	0.188905455	\\
0.409	0.189905455	\\
0.41	0.190905455	\\
0.411	0.191905455	\\
0.412	0.192905455	\\
0.413	0.193905455	\\
0.414	0.194905455	\\
0.415	0.195905455	\\
0.416	0.196905455	\\
0.417	0.197905455	\\
0.418	0.198905455	\\
0.419	0.199905455	\\
0.42	0.200905455	\\
0.421	0.201905455	\\
0.422	0.202905455	\\
0.423	0.203905455	\\
0.424	0.204905455	\\
0.425	0.205905455	\\
0.426	0.206905455	\\
0.427	0.207905455	\\
0.428	0.208905455	\\
0.429	0.209905455	\\
0.43	0.210905455	\\
0.431	0.211905455	\\
0.432	0.212905455	\\
0.433	0.213905455	\\
0.434	0.214905455	\\
0.435	0.215905455	\\
0.436	0.216905455	\\
0.437	0.217905455	\\
0.438	0.218905455	\\
0.439	0.219905455	\\
0.44	0.220905455	\\
0.441	0.221905455	\\
0.442	0.222905455	\\
0.443	0.223905455	\\
0.444	0.224905455	\\
0.445	0.225905455	\\
0.446	0.226905455	\\
0.447	0.227905455	\\
0.448	0.228905455	\\
0.449	0.229905455	\\
0.45	0.230905455	\\
0.451	0.231905455	\\
0.452	0.232905455	\\
0.453	0.233905455	\\
0.454	0.234905455	\\
0.455	0.235905455	\\
0.456	0.236905455	\\
0.457	0.237905455	\\
0.458	0.238905455	\\
0.459	0.239905455	\\
0.46	0.240905455	\\
0.461	0.241905455	\\
0.462	0.242905455	\\
0.463	0.243905455	\\
0.464	0.244905455	\\
0.465	0.245905455	\\
0.466	0.246905455	\\
0.467	0.247905455	\\
0.468	0.248905455	\\
0.469	0.249905455	\\
0.47	0.250905455	\\
0.471	0.251905455	\\
0.472	0.252905455	\\
0.473	0.253905455	\\
0.474	0.254905455	\\
0.475	0.255905455	\\
0.476	0.256905455	\\
0.477	0.257905455	\\
0.478	0.258905455	\\
0.479	0.259905455	\\
0.48	0.260905455	\\
0.481	0.261905455	\\
0.482	0.262905455	\\
0.483	0.263905455	\\
0.484	0.264905455	\\
0.485	0.265905455	\\
0.486	0.266905455	\\
0.487	0.267905455	\\
0.488	0.268905455	\\
0.489	0.269905455	\\
0.49	0.270905455	\\
0.491	0.271905455	\\
0.492	0.272905455	\\
0.493	0.273905455	\\
0.494	0.274905455	\\
0.495	0.275905455	\\
0.496	0.276905455	\\
0.497	0.277905455	\\
0.498	0.278905455	\\
0.499	0.279905455	\\
0.5	0.280905455	\\
0.501	0.281905455	\\
0.502	0.282905455	\\
0.503	0.283905455	\\
0.504	0.284905455	\\
0.505	0.285905455	\\
0.506	0.286905455	\\
0.507	0.287905455	\\
0.508	0.288905455	\\
0.509	0.289905455	\\
0.51	0.290905455	\\
0.511	0.291905455	\\
0.512	0.292905455	\\
0.513	0.293905455	\\
0.514	0.294905455	\\
0.515	0.295905455	\\
0.516	0.296905455	\\
0.517	0.297905455	\\
0.518	0.298905455	\\
0.519	0.299905455	\\
0.52	0.300905455	\\
0.521	0.301905455	\\
0.522	0.302905455	\\
0.523	0.303905455	\\
0.524	0.304905455	\\
0.525	0.305905455	\\
0.526	0.30690618	\\
0.527	0.307911479	\\
0.528	0.308921906	\\
0.529	0.309937447	\\
0.53	0.310958321	\\
0.531	0.311984382	\\
0.532	0.313015784	\\
0.533	0.314052397	\\
0.534	0.315094645	\\
0.535	0.316142177	\\
0.536	0.317195251	\\
0.537	0.318253928	\\
0.538	0.319317948	\\
0.539	0.320387646	\\
0.54	0.321463085	\\
0.541	0.322544329	\\
0.542	0.323631441	\\
0.543	0.324724483	\\
0.544	0.325823037	\\
0.545	0.326927596	\\
0.546	0.328038224	\\
0.547	0.329154982	\\
0.548	0.330277936	\\
0.549	0.331406534	\\
0.55	0.332542041	\\
0.551	0.333683267	\\
0.552	0.334830889	\\
0.553	0.335985691	\\
0.554	0.337146323	\\
0.555	0.338313541	\\
0.556	0.339487411	\\
0.557	0.340667168	\\
0.558	0.341854504	\\
0.559	0.343048682	\\
0.56	0.344249767	\\
0.561	0.345457822	\\
0.562	0.346671951	\\
0.563	0.347894113	\\
0.564	0.349123439	\\
0.565	0.350359991	\\
0.566	0.351602767	\\
0.567	0.35285394	\\
0.568	0.354112535	\\
0.569	0.355378615	\\
0.57	0.356652245	\\
0.571	0.357933489	\\
0.572	0.359223647	\\
0.573	0.360520343	\\
0.574	0.361824848	\\
0.575	0.363138544	\\
0.576	0.364460234	\\
0.577	0.365789983	\\
0.578	0.367127858	\\
0.579	0.36847535	\\
0.58	0.369829699	\\
0.581	0.371193852	\\
0.582	0.372566448	\\
0.583	0.373947553	\\
0.584	0.375338798	\\
0.585	0.376738741	\\
0.586	0.378147446	\\
0.587	0.379566632	\\
0.588	0.380994769	\\
0.589	0.382431925	\\
0.59	0.383879904	\\
0.591	0.38533709	\\
0.592	0.386805347	\\
0.593	0.388283004	\\
0.594	0.389770127	\\
0.595	0.391268665	\\
0.596	0.392778773	\\
0.597	0.394298666	\\
0.598	0.395830379	\\
0.599	0.397372071	\\
0.6	0.398925837	\\
0.601	0.400489776	\\
0.602	0.402066042	\\
0.603	0.403654794	\\
0.604	0.405254042	\\
0.605	0.406868209	\\
0.606	0.408493129	\\
0.607	0.410131109	\\
0.608	0.411780039	\\
0.609	0.413444586	\\
0.61	0.415120345	\\
0.611	0.416812104	\\
0.612	0.418515335	\\
0.613	0.420232531	\\
0.614	0.421963854	\\
0.615	0.423711952	\\
0.616	0.425472052	\\
0.617	0.427249317	\\
0.618	0.429041431	\\
0.619	0.430848557	\\
0.62	0.432670864	\\
0.621	0.434511193	\\
0.622	0.436367101	\\
0.623	0.438241498	\\
0.624	0.440131876	\\
0.625	0.442038407	\\
0.626	0.4439641	\\
0.627	0.445909225	\\
0.628	0.447871149	\\
0.629	0.449852984	\\
0.63	0.451855004	\\
0.631	0.453877486	\\
0.632	0.455920707	\\
0.633	0.457984947	\\
0.634	0.460067377	\\
0.635	0.462177609	\\
0.636	0.464306596	\\
0.637	0.466457735	\\
0.638	0.468637813	\\
0.639	0.470837474	\\
0.64	0.473063547	\\
0.641	0.475316433	\\
0.642	0.477596535	\\
0.643	0.479900829	\\
0.644	0.482236551	\\
0.645	0.484600721	\\
0.646	0.486993753	\\
0.647	0.48941965	\\
0.648	0.491875328	\\
0.649	0.494364873	\\
0.65	0.496885132	\\
0.651	0.499444018	\\
0.652	0.502038425	\\
0.653	0.504668915	\\
0.654	0.507339916	\\
0.655	0.510052123	\\
0.656	0.512806238	\\
0.657	0.51560297	\\
0.658	0.518447068	\\
0.659	0.521335312	\\
0.66	0.524276671	\\
0.661	0.527267988	\\
0.662	0.530314357	\\
0.663	0.53341681	\\
0.664	0.536580693	\\
0.665	0.539802861	\\
0.666	0.543097574	\\
0.667	0.546457471	\\
0.668	0.549892794	\\
0.669	0.55340512	\\
0.67	0.557000647	\\
0.671	0.560681162	\\
0.672	0.564462587	\\
0.673	0.568337754	\\
0.674	0.57232797	\\
0.675	0.57643115	\\
0.676	0.580659823	\\
0.677	0.585032054	\\
0.678	0.58955157	\\
0.679	0.594237552	\\
0.68	0.59911008	\\
0.681	0.604184941	\\
0.682	0.609489318	\\
0.683	0.615057179	\\
0.684	0.620929795	\\
0.685	0.627156315	\\
0.686	0.633800093	\\
0.687	0.640962911	\\
0.688	0.648765132	\\
0.689	0.657407424	\\
0.69	0.667228323	\\
0.691	0.678840368	\\
0.692	0.693775763	\\
0.693	0.72009009	\\
	};
	\legend{}
	\addlegendentry{$E_{\mbox{\tiny e}}(\rr)$}

	\addplot[smooth,color=olive,thick,dash pattern={on 4pt off 2pt}]
	table[row sep=crcr] 
	{
0	-2.19E-01	\\
0.001	-2.18E-01	\\
0.002	-2.17E-01	\\
0.003	-2.16E-01	\\
0.004	-2.15E-01	\\
0.005	-2.14E-01	\\
0.006	-2.13E-01	\\
0.007	-2.12E-01	\\
0.008	-2.11E-01	\\
0.009	-2.10E-01	\\
0.01	-2.09E-01	\\
0.011	-2.08E-01	\\
0.012	-2.07E-01	\\
0.013	-2.06E-01	\\
0.014	-2.05E-01	\\
0.015	-2.04E-01	\\
0.016	-2.03E-01	\\
0.017	-2.02E-01	\\
0.018	-0.201094545	\\
0.019	-0.200094545	\\
0.02	-0.199094545	\\
0.021	-0.198094545	\\
0.022	-0.197094545	\\
0.023	-0.196094545	\\
0.024	-0.195094545	\\
0.025	-0.194094545	\\
0.026	-0.193094545	\\
0.027	-0.192094545	\\
0.028	-0.191094545	\\
0.029	-0.190094545	\\
0.03	-0.189094545	\\
0.031	-0.188094545	\\
0.032	-0.187094545	\\
0.033	-0.186094545	\\
0.034	-0.185094545	\\
0.035	-0.184094545	\\
0.036	-0.183094545	\\
0.037	-0.182094545	\\
0.038	-0.181094545	\\
0.039	-0.180094545	\\
0.04	-0.179094545	\\
0.041	-0.178094545	\\
0.042	-0.177094545	\\
0.043	-0.176094545	\\
0.044	-0.175094545	\\
0.045	-0.174094545	\\
0.046	-0.173094545	\\
0.047	-0.172094545	\\
0.048	-0.171094545	\\
0.049	-0.170094545	\\
0.05	-0.169094545	\\
0.051	-0.168094545	\\
0.052	-0.167094545	\\
0.053	-0.166094545	\\
0.054	-0.165094545	\\
0.055	-0.164094545	\\
0.056	-0.163094545	\\
0.057	-0.162094545	\\
0.058	-0.161094545	\\
0.059	-0.160094545	\\
0.06	-0.159094545	\\
0.061	-0.158094545	\\
0.062	-0.157094545	\\
0.063	-0.156094545	\\
0.064	-0.155094545	\\
0.065	-0.154094545	\\
0.066	-0.153094545	\\
0.067	-0.152094545	\\
0.068	-0.151094545	\\
0.069	-0.150094545	\\
0.07	-0.149094545	\\
0.071	-0.148094545	\\
0.072	-0.147094545	\\
0.073	-0.146094545	\\
0.074	-0.145094545	\\
0.075	-0.144094545	\\
0.076	-0.143094545	\\
0.077	-0.142094545	\\
0.078	-0.141094545	\\
0.079	-0.140094545	\\
0.08	-0.139094545	\\
0.081	-0.138094545	\\
0.082	-0.137094545	\\
0.083	-0.136094545	\\
0.084	-0.135094545	\\
0.085	-0.134094545	\\
0.086	-0.133094545	\\
0.087	-0.132094545	\\
0.088	-0.131094545	\\
0.089	-0.130094545	\\
0.09	-0.129094545	\\
0.091	-0.128094545	\\
0.092	-0.127094545	\\
0.093	-0.126094545	\\
0.094	-0.125094545	\\
0.095	-0.124094545	\\
0.096	-0.123094545	\\
0.097	-0.122094545	\\
0.098	-0.121094545	\\
0.099	-0.120094545	\\
0.1	-0.119094545	\\
0.101	-0.118094545	\\
0.102	-0.117094545	\\
0.103	-0.116094545	\\
0.104	-0.115094545	\\
0.105	-0.114094545	\\
0.106	-0.113094545	\\
0.107	-0.112094545	\\
0.108	-0.111094545	\\
0.109	-0.110094545	\\
0.11	-0.109094545	\\
0.111	-0.108094545	\\
0.112	-0.107094545	\\
0.113	-0.106094545	\\
0.114	-0.105094545	\\
0.115	-0.104094545	\\
0.116	-0.103094545	\\
0.117	-0.102094545	\\
0.118	-0.101094545	\\
0.119	-0.100094545	\\
0.12	-0.099094545	\\
0.121	-0.098094545	\\
0.122	-0.097094545	\\
0.123	-0.096094545	\\
0.124	-0.095094545	\\
0.125	-0.094094545	\\
0.126	-0.093094545	\\
0.127	-0.092094545	\\
0.128	-0.091094545	\\
0.129	-0.090094545	\\
0.13	-0.089094545	\\
0.131	-0.088094545	\\
0.132	-0.087094545	\\
0.133	-0.086094545	\\
0.134	-0.085094545	\\
0.135	-0.084094545	\\
0.136	-0.083094545	\\
0.137	-0.082094545	\\
0.138	-0.081094545	\\
0.139	-0.080094545	\\
0.14	-0.079094545	\\
0.141	-0.078094545	\\
0.142	-0.077094545	\\
0.143	-0.076094545	\\
0.144	-0.075094545	\\
0.145	-0.074094545	\\
0.146	-0.073094545	\\
0.147	-0.072094545	\\
0.148	-0.071094545	\\
0.149	-0.070094545	\\
0.15	-0.069094545	\\
0.151	-0.068094545	\\
0.152	-0.067094545	\\
0.153	-0.066094545	\\
0.154	-0.065094545	\\
0.155	-0.064094545	\\
0.156	-0.063094545	\\
0.157	-0.062094545	\\
0.158	-0.061094545	\\
0.159	-0.060094545	\\
0.16	-0.059094545	\\
0.161	-0.058094545	\\
0.162	-0.057094545	\\
0.163	-0.056094545	\\
0.164	-0.055094545	\\
0.165	-0.054094545	\\
0.166	-0.053094545	\\
0.167	-0.052094545	\\
0.168	-0.051094545	\\
0.169	-0.050094545	\\
0.17	-0.049094545	\\
0.171	-0.048094545	\\
0.172	-0.047094545	\\
0.173	-0.046094545	\\
0.174	-0.045094545	\\
0.175	-0.044094545	\\
0.176	-0.043094545	\\
0.177	-0.042094545	\\
0.178	-0.041094545	\\
0.179	-0.040094545	\\
0.18	-0.039094545	\\
0.181	-0.038094545	\\
0.182	-0.037094545	\\
0.183	-0.036094545	\\
0.184	-0.035094545	\\
0.185	-0.034094545	\\
0.186	-0.033094545	\\
0.187	-0.032094545	\\
0.188	-0.031094545	\\
0.189	-0.030094545	\\
0.19	-0.029094545	\\
0.191	-0.028094545	\\
0.192	-0.027094545	\\
0.193	-0.026094545	\\
0.194	-0.025094545	\\
0.195	-0.024094545	\\
0.196	-0.023094545	\\
0.197	-0.022094545	\\
0.198	-0.021094545	\\
0.199	-0.020094545	\\
0.2	-0.019094545	\\
0.201	-0.018094545	\\
0.202	-0.017094545	\\
0.203	-0.016094545	\\
0.204	-0.015094545	\\
0.205	-0.014094545	\\
0.206	-0.013094545	\\
0.207	-0.012094545	\\
0.208	-0.011094545	\\
0.209	-0.010094545	\\
0.21	-0.009094545	\\
0.211	-0.008094545	\\
0.212	-0.007094545	\\
0.213	-0.006094545	\\
0.214	-0.005094545	\\
0.215	-0.004094545	\\
0.216	-0.003094545	\\
0.217	-0.002094545	\\
0.218	-0.001094545	\\
0.219	-9.45E-05	\\
0.22	0.000905455	\\
0.221	0.001905455	\\
0.222	0.002905455	\\
0.223	0.003905455	\\
0.224	0.004905455	\\
0.225	0.005905455	\\
0.226	0.006905455	\\
0.227	0.007905455	\\
0.228	0.008905455	\\
0.229	0.009905455	\\
0.23	0.010905455	\\
0.231	0.011905455	\\
0.232	0.012905455	\\
0.233	0.013905455	\\
0.234	0.014905455	\\
0.235	0.015905455	\\
0.236	0.016905455	\\
0.237	0.017905455	\\
0.238	0.018905455	\\
0.239	0.019905455	\\
0.24	0.020905455	\\
0.241	0.021905455	\\
0.242	0.022905455	\\
0.243	0.023905455	\\
0.244	0.024905455	\\
0.245	0.025905455	\\
0.246	0.026905455	\\
0.247	0.027905455	\\
0.248	0.028905455	\\
0.249	0.029905455	\\
0.25	0.030905455	\\
0.251	0.031905455	\\
0.252	0.032905455	\\
0.253	0.033905455	\\
0.254	0.034905455	\\
0.255	0.035905455	\\
0.256	0.036905455	\\
0.257	0.037905455	\\
0.258	0.038905455	\\
0.259	0.039905455	\\
0.26	0.040905455	\\
0.261	0.041905455	\\
0.262	0.042905455	\\
0.263	0.043905455	\\
0.264	0.044905455	\\
0.265	0.045905455	\\
0.266	0.046905455	\\
0.267	0.047905455	\\
0.268	0.048905455	\\
0.269	0.049905455	\\
0.27	0.050905455	\\
0.271	0.051905455	\\
0.272	0.052905455	\\
0.273	0.053905455	\\
0.274	0.054905455	\\
0.275	0.055905455	\\
0.276	0.056905455	\\
0.277	0.057905455	\\
0.278	0.058905455	\\
0.279	0.059905455	\\
0.28	0.060905455	\\
0.281	0.061905455	\\
0.282	0.062905455	\\
0.283	0.063905455	\\
0.284	0.064905455	\\
0.285	0.065905455	\\
0.286	0.066905455	\\
0.287	0.067905455	\\
0.288	0.068905455	\\
0.289	0.069905455	\\
0.29	0.070905455	\\
0.291	0.071905455	\\
0.292	0.072905455	\\
0.293	0.073905455	\\
0.294	0.074905455	\\
0.295	0.075905455	\\
0.296	0.076905455	\\
0.297	0.077905455	\\
0.298	0.078905455	\\
0.299	0.079905455	\\
0.3	0.080905455	\\
0.301	0.081905455	\\
0.302	0.082905455	\\
0.303	0.083905455	\\
0.304	0.084905455	\\
0.305	0.085905455	\\
0.306	0.086905455	\\
0.307	0.087905455	\\
0.308	0.088905455	\\
0.309	0.089905455	\\
0.31	0.090905455	\\
0.311	0.091905455	\\
0.312	0.092905455	\\
0.313	0.093905455	\\
0.314	0.094905455	\\
0.315	0.095905455	\\
0.316	0.096905455	\\
0.317	0.097905455	\\
0.318	0.098905455	\\
0.319	0.099905455	\\
0.32	0.100905455	\\
0.321	0.101905455	\\
0.322	0.102905455	\\
0.323	0.103905455	\\
0.324	0.104905455	\\
0.325	0.105905455	\\
0.326	0.106905455	\\
0.327	0.107905455	\\
0.328	0.108905455	\\
0.329	0.109905455	\\
0.33	0.110905455	\\
0.331	0.111905455	\\
0.332	0.112905455	\\
0.333	0.113905455	\\
0.334	0.114905455	\\
0.335	0.115905455	\\
0.336	0.116905455	\\
0.337	0.117905455	\\
0.338	0.118905455	\\
0.339	0.119905455	\\
0.34	0.120905455	\\
0.341	0.121905455	\\
0.342	0.122905455	\\
0.343	0.123905455	\\
0.344	0.124905455	\\
0.345	0.125905455	\\
0.346	0.126905455	\\
0.347	0.127905455	\\
0.348	0.128905455	\\
0.349	0.129905455	\\
0.35	0.130905455	\\
0.351	0.131905455	\\
0.352	0.132905455	\\
0.353	0.133905455	\\
0.354	0.134905455	\\
0.355	0.135905455	\\
0.356	0.136905455	\\
0.357	0.137905455	\\
0.358	0.138905455	\\
0.359	0.139905455	\\
0.36	0.140905455	\\
0.361	0.141905455	\\
0.362	0.142905455	\\
0.363	0.143905455	\\
0.364	0.144905455	\\
0.365	0.145905455	\\
0.366	0.146905455	\\
0.367	0.147905455	\\
0.368	0.148905455	\\
0.369	0.149905455	\\
0.37	0.150905455	\\
0.371	0.151905455	\\
0.372	0.152905455	\\
0.373	0.153905455	\\
0.374	0.154905455	\\
0.375	0.155905455	\\
0.376	0.156905455	\\
0.377	0.157905455	\\
0.378	0.158905455	\\
0.379	0.159905455	\\
0.38	0.160905455	\\
0.381	0.161905455	\\
0.382	0.162905455	\\
0.383	0.163905455	\\
0.384	0.164905455	\\
0.385	0.165905455	\\
0.386	0.166905455	\\
0.387	0.167905455	\\
0.388	0.168905455	\\
0.389	0.169905455	\\
0.39	0.170905455	\\
0.391	0.171905455	\\
0.392	0.172905455	\\
0.393	0.173905455	\\
0.394	0.174905455	\\
0.395	0.175905455	\\
0.396	0.176905455	\\
0.397	0.177905455	\\
0.398	0.178905455	\\
0.399	0.179905455	\\
0.4	0.180905455	\\
0.401	0.181905455	\\
0.402	0.182905455	\\
0.403	0.183905455	\\
0.404	0.184905455	\\
0.405	0.185905455	\\
0.406	0.186905455	\\
0.407	0.187905455	\\
0.408	0.188905455	\\
0.409	0.189905455	\\
0.41	0.190905455	\\
0.411	0.191905455	\\
0.412	0.192905455	\\
0.413	0.193905455	\\
0.414	0.194905455	\\
0.415	0.195905455	\\
0.416	0.196905455	\\
0.417	0.197905455	\\
0.418	0.198905455	\\
0.419	0.199905455	\\
0.42	0.200905455	\\
0.421	0.201905455	\\
0.422	0.202905455	\\
0.423	0.203905455	\\
0.424	0.204905455	\\
0.425	0.205905455	\\
0.426	0.206905455	\\
0.427	0.207905455	\\
0.428	0.208905455	\\
0.429	0.209905455	\\
0.43	0.210905455	\\
0.431	0.211905455	\\
0.432	0.212905455	\\
0.433	0.213905455	\\
0.434	0.214905455	\\
0.435	0.215905455	\\
0.436	0.216905455	\\
0.437	0.217905455	\\
0.438	0.218905455	\\
0.439	0.219905455	\\
0.44	0.220905455	\\
0.441	0.221905455	\\
0.442	0.222905455	\\
0.443	0.223905455	\\
0.444	0.224905455	\\
0.445	0.225905455	\\
0.446	0.226905455	\\
0.447	0.227906335	\\
0.448	0.228910854	\\
0.449	0.229919226	\\
0.45	0.230931476	\\
0.451	0.23194763	\\
0.452	0.23296772	\\
0.453	0.233991896	\\
0.454	0.235019722	\\
0.455	0.236051693	\\
0.456	0.237088152	\\
0.457	0.238129577	\\
0.458	0.239173225	\\
0.459	0.240220233	\\
0.46	0.241272335	\\
0.461	0.242328635	\\
0.462	0.243389194	\\
0.463	0.244454099	\\
0.464	0.245523214	\\
0.465	0.246596701	\\
0.466	0.247674577	\\
0.467	0.248757578	\\
0.468	0.24984364	\\
0.469	0.250934558	\\
0.47	0.252028139	\\
0.471	0.253126665	\\
0.472	0.254230812	\\
0.473	0.255339663	\\
0.474	0.256451169	\\
0.475	0.257567428	\\
0.476	0.258688256	\\
0.477	0.259813791	\\
0.478	0.260944051	\\
0.479	0.26207903	\\
0.48	0.263218739	\\
0.481	0.264363174	\\
0.482	0.265512273	\\
0.483	0.266666857	\\
0.484	0.267825827	\\
0.485	0.268989274	\\
0.486	0.270156336	\\
0.487	0.271329314	\\
0.488	0.272524479	\\
0.489	0.273690996	\\
0.49	0.274878387	\\
0.491	0.276070533	\\
0.492	0.277268285	\\
0.493	0.278472182	\\
0.494	0.279679285	\\
0.495	0.280891864	\\
0.496	0.282110206	\\
0.497	0.283334024	\\
0.498	0.284561801	\\
0.499	0.285796011	\\
0.5	0.287036367	\\
0.501	0.288278912	\\
0.502	0.289531032	\\
0.503	0.290832115	\\
0.504	0.292093736	\\
0.505	0.293310819	\\
0.506	0.294581813	\\
0.507	0.295860338	\\
0.508	0.297141565	\\
0.509	0.298431114	\\
0.51	0.299723802	\\
0.511	0.301024706	\\
0.512	0.302329	\\
0.513	0.303642737	\\
0.514	0.304962489	\\
0.515	0.306280537	\\
0.516	0.30761231	\\
0.517	0.308952742	\\
0.518	0.310423333	\\
0.519	0.311743353	\\
0.52	0.312994701	\\
0.521	0.314346102	\\
0.522	0.31571028	\\
0.523	0.317081489	\\
0.524	0.31845957	\\
0.525	0.319844917	\\
0.526	0.321234617	\\
0.527	0.322631968	\\
0.528	0.324035804	\\
0.529	0.325446126	\\
0.53	0.32686395	\\
0.531	0.328286466	\\
0.532	0.329923916	\\
0.533	0.331153866	\\
0.534	0.332597836	\\
0.535	0.334052862	\\
0.536	0.335509755	\\
0.537	0.336971251	\\
0.538	0.338443285	\\
0.539	0.339922216	\\
0.54	0.341408419	\\
0.541	0.34290563	\\
0.542	0.344402883	\\
0.543	0.345911084	\\
0.544	0.34742676	\\
0.545	0.34935036	\\
0.546	0.350481003	\\
0.547	0.35201984	\\
0.548	0.353565598	\\
0.549	0.355124192	\\
0.55	0.356681361	\\
0.551	0.358251844	\\
0.552	0.359829445	\\
0.553	0.361415104	\\
0.554	0.363009266	\\
0.555	0.364612907	\\
0.556	0.36622251	\\
0.557	0.368410793	\\
0.558	0.369472616	\\
0.559	0.371105748	\\
0.56	0.372750799	\\
0.561	0.374404877	\\
0.562	0.376067913	\\
0.563	0.37773927	\\
0.564	0.37941993	\\
0.565	0.381110375	\\
0.566	0.382813436	\\
0.567	0.384518083	\\
0.568	0.38713003	\\
0.569	0.388019865	\\
0.57	0.389809224	\\
0.571	0.391572306	\\
0.572	0.39320601	\\
0.573	0.394969199	\\
0.574	0.396751939	\\
0.575	0.398532954	\\
0.576	0.400341082	\\
0.577	0.402174909	\\
0.578	0.403966297	\\
0.579	0.405785137	\\
0.58	0.407638236	\\
0.581	0.409511874	\\
0.582	0.411337696	\\
0.583	0.413210782	\\
0.584	0.415094863	\\
0.585	0.416992984	\\
0.586	0.418902375	\\
0.587	0.420962278	\\
0.588	0.422842178	\\
0.589	0.424692615	\\
0.59	0.426650205	\\
0.591	0.428616574	\\
0.592	0.430601018	\\
0.593	0.432591831	\\
0.594	0.434601621	\\
0.595	0.436626312	\\
0.596	0.438653628	\\
0.597	0.44096621	\\
0.598	0.442763337	\\
0.599	0.444834964	\\
0.6	0.446919235	\\
0.601	0.449021652	\\
0.602	0.451140517	\\
0.603	0.453275445	\\
0.604	0.455417499	\\
0.605	0.457582662	\\
0.606	0.460262816	\\
0.607	0.461954103	\\
0.608	0.46416425	\\
0.609	0.466389827	\\
0.61	0.468627228	\\
0.611	0.470886493	\\
0.612	0.473162996	\\
0.613	0.475454714	\\
0.614	0.477820197	\\
0.615	0.48009469	\\
0.616	0.482442359	\\
0.617	0.484807292	\\
0.618	0.487192085	\\
0.619	0.489597296	\\
0.62	0.49202182	\\
0.621	0.494477788	\\
0.622	0.496932584	\\
0.623	0.499413728	\\
0.624	0.50192809	\\
0.625	0.504456437	\\
0.626	0.506999338	\\
0.627	0.509592314	\\
0.628	0.51233953	\\
0.629	0.514857808	\\
0.63	0.517436943	\\
0.631	0.52010287	\\
0.632	0.522796838	\\
0.633	0.52551589	\\
0.634	0.528265901	\\
0.635	0.531468959	\\
0.636	0.533842096	\\
0.637	0.536672719	\\
0.638	0.539527075	\\
0.639	0.542419894	\\
0.64	0.545337995	\\
0.641	0.548288157	\\
0.642	0.551888956	\\
0.643	0.554290006	\\
0.644	0.557341811	\\
0.645	0.560429252	\\
0.646	0.563554603	\\
0.647	0.566714314	\\
0.648	0.569917291	\\
0.649	0.573154362	\\
0.65	0.5764357	\\
0.651	0.579759592	\\
0.652	0.58320038	\\
0.653	0.586549275	\\
0.654	0.590002664	\\
0.655	0.59351306	\\
0.656	0.597074921	\\
0.657	0.601111542	\\
0.658	0.604359187	\\
0.659	0.608084963	\\
0.66	0.61187163	\\
0.661	0.615721483	\\
0.662	0.619635529	\\
0.663	0.623618942	\\
0.664	0.627673895	\\
0.665	0.631803945	\\
0.666	0.636015209	\\
0.667	0.640314178	\\
0.668	0.644697921	\\
0.669	0.649177126	\\
0.67	0.653849252	\\
0.671	0.658422715	\\
0.672	0.663236687	\\
0.673	0.668172197	\\
0.674	0.673172768	\\
0.675	0.678356883	\\
0.676	0.68370228	\\
0.677	0.689197971	\\
0.678	0.694871926	\\
0.679	0.700867566	\\
0.68	0.706847501	\\
0.681	0.713188034	\\
0.682	0.719803887	\\
0.683	0.726755338	\\
0.684	0.734031912	\\
0.685	0.74174058	\\
0.686	0.749953859	\\
0.687	0.758749997	\\
0.688	0.768299893	\\
0.689	0.778861191	\\
0.69	0.790807086	\\
0.691	0.804881292	\\
0.692	0.822846822	\\
0.693	0.855168077	\\
	};
	\addlegendentry{$E_{\mbox{\tiny ex}}^{\mbox{\tiny fr}}(\rr)$}	
		
	\end{axis}
	
	\end{tikzpicture}
	\caption{Graphs of the functions $E_{\mbox{\tiny ex}}^{\mbox{\tiny fr}}(\rr)$ and $E_{\mbox{\tiny e}}(\infty,\rr)$.} \label{Z-Channel-Numeric4}
\end{figure}

As can be seen in Fig.\ \ref{Z-Channel-Numeric4}, both $E_{\mbox{\tiny ex}}^{\mbox{\tiny fr}}(\rr)$ and $E_{\mbox{\tiny e}}(\infty,\rr)$ tend to infinity as $R$ tends to $\log 2 \approx 0.693$. For relatively high binning rates, $E_{\mbox{\tiny ex}}^{\mbox{\tiny fr}}(\rr)$ is strictly higher than $E_{\mbox{\tiny e}}(\infty,\rr)$, which can be explained in the following way:
Referring to the analogy between SW coding and channel coding, one can think of each bin as containing a channel code. In general, a channel code behaves well if it does not contain pairs of relatively ``close'' codewords. Since we randomly assign the source vectors into the bins (even if the populations of the bins are totally equal, which can be attained by randomly partitioning each type class into $\exp\{n\rr\}$ subsets), it is reasonable to assume that some bins will contain relatively bad codebooks. 
On the other hand, in the expurgated SW code \cite{CKgraph}, each type class $\calT(Q_{U})$ is partitioned into $\exp\{n\rr\}$ ``balanced'' subsets in some sense (referring to the enumerators $N(Q_{UU'})$ in \eqref{DEF_ENUMERATOR}, they are equally populated in all of the bins), such that the codebooks contained in the bins have approximately equal error probabilities. Moreover, we conclude from \eqref{SW_Expurgated} that each bin contains a codebook with a quality of an expurgated channel code.
This code is certainly better than the TRCs in the SD ensemble.       

In channel coding, it is known \cite{SSG} that the random Gilbert--Varshamov ensemble has an exact random coding error exponent which is as high as the maximum between \eqref{Channel_rc} and \eqref{Channel_ex}.
In SW source coding, on the other hand, it seems to be a more challenging problem to define an ensemble, such that the error exponent of its TRCs is as high as $E_{\mbox{\tiny ex}}^{\mbox{\tiny fr}}(\rr)$ of \eqref{SW_Expurgated}.  
Since the gap between $E_{\mbox{\tiny ex}}^{\mbox{\tiny fr}}(\rr)$ and $E_{\mbox{\tiny e}}(\infty,\rr)$ is not necessarily very significant, as can be seen in Fig.\ \ref{Z-Channel-Numeric4}, we conclude that the SD ensemble may be more attractive because the amount of computations needed for drawing a code from it are much lower than the amount of computations required for having an expurgated SW code. In addition, it is important to note that the probability of drawing a SD code with an exponent much lower than $E_{\mbox{\tiny e}}(\infty,\rr)$ decays exponentially fast, in analogy to the result in pure channel coding \cite{TMWG}.

\section*{Appendix A}
\renewcommand{\theequation}{A.\arabic{equation}}
\setcounter{equation}{0}  
\subsection*{Proof of Theorem \ref{Thm_random_binning_SD_VR}}

By definition, we have
\begin{align}
\mathbb{E} [P_{\mbox{\tiny e}} (\calB_{n})] 
= \mathbb{E} \left[ \frac{\sum_{\bu' \in \calB(\bU), \bu' \neq \bU} \exp\{nf(\hat{P}_{\bu'\bV})\}}{\sum_{\tilde{\bu} \in \calB(\bU)} \exp\{nf(\hat{P}_{\tilde{\bu}\bV})\}} \right].
\end{align}

\subsubsection*{Step 1: Averaging Over the Random Code}
We first condition on the true source sequences $(\bU=\bu, \bV=\bv)$ and take the expectation only w.r.t.\ the random binning. 
We get
\begin{align}
&\mathbb{E} [ P_{\mbox{\tiny e}} (\calB_{n}) | \bu,\bv ] \nn \\
&= \mathbb{E} \left[ \frac{\sum_{\bu' \in \calB(\bu), \bu' \neq \bu} \exp\{nf(\hat{P}_{\bu'\bv})\}} {\exp \{n \cdot f( \hat{P}_{\bu \bv} ) \} + \sum_{\bu' \in \calB(\bu), \bu' \neq \bu} \exp\{nf(\hat{P}_{\bu'\bv})\} } \right] \\
\label{ToRef0}
&= \int_{0}^{1} \prob \left\{ \frac{\sum_{\bu' \in \calB(\bu), \bu' \neq \bu} \exp\{nf(\hat{P}_{\bu'\bv})\}} {\exp \{n \cdot f( \hat{P}_{\bu \bv} ) \} + \sum_{\bu' \in \calB(\bu), \bu' \neq \bu} \exp\{nf(\hat{P}_{\bu'\bv})\} } \geq s \right\} ds \\
\label{ToRef1}
&= \int_{0}^{\infty} ne^{-n \xi} \cdot \prob \left\{ \frac{\sum_{\bu' \in \calB(\bu), \bu' \neq \bu} \exp\{nf(\hat{P}_{\bu'\bv})\}} {\exp \{n \cdot f( \hat{P}_{\bu \bv} ) \} + \sum_{\bu' \in \calB(\bu), \bu' \neq \bu} \exp\{nf(\hat{P}_{\bu'\bv})\} } \geq e^{-n \xi} \right\} d\xi \\
&= \int_{0}^{\infty} ne^{-n \xi} \cdot  \prob \left\{ (1- e^{-n \xi}) \sum_{\bu' \in \calB(\bu), \bu' \neq \bu} \exp\{nf(\hat{P}_{\bu'\bv})\} \geq e^{-n \xi}  \exp\{n  f(\hat{P}_{\bu\bv}) \}  \right\} d\xi \\
\label{step0}
&\doteq \int_{0}^{\infty} ne^{-n \xi} \cdot  \prob \left\{ \sum_{\bu' \in \calB(\bu), \bu' \neq \bu} \exp\{nf(\hat{P}_{\bu'\bv})\} \geq \exp\{n [ f(\hat{P}_{\bu\bv}) - \xi ]\}  \right\} d\xi,
\end{align}
where \eqref{ToRef1} follows by changing the integration variable in \eqref{ToRef0} according to $s = e^{-n \xi}$.
Define 
\begin{align}
N_{\bu,\bv}(Q_{U|V}) = \sum_{\bu' \in \calB(\bu), \bu' \neq \bu} \IND \{\bu' \in \calT(Q_{U|V}|\bv)\},
\end{align}
such that the probability in \eqref{step0} is given by
\begin{align}
&\prob \left\{ \sum_{\bu' \in \calB(\bu), \bu' \neq \bu} \exp\{nf(\hat{P}_{\bu'\bv})\} \geq \exp\{n [ f(\hat{P}_{\bu\bv}) - \xi ]\}  \right\} \nn \\
&=\prob \left\{ \sum_{Q_{U'|V}} N_{\bu,\bv}(Q_{U'|V}) \exp\{nf(Q_{U'V})\} \geq \exp\{n [ f(\hat{P}_{\bu\bv}) - \xi ]\}  \right\} \\
&\doteq \prob \left\{ \max_{Q_{U'|V}} N_{\bu,\bv}(Q_{U'|V}) \exp\{nf(Q_{U'V})\} \geq \exp\{n [f(\hat{P}_{\bu\bv}) - \xi ]\}  \right\} \\
&=\prob \bigcup_{Q_{U'|V}} \left\{N_{\bu,\bv}(Q_{U'|V}) \exp\{nf(Q_{U'V})\} \geq \exp\{n [f(\hat{P}_{\bu\bv}) - \xi ]\}  \right\} \\
&\doteq \sum_{Q_{U'|V}} \prob \left\{N_{\bu,\bv}(Q_{U'|V}) \geq \exp\{n [ f(\hat{P}_{\bu\bv}) - f(Q_{U'V}) - \xi ]\}  \right\},
\end{align}
where $Q_{U'V} = Q_{U'|V} \times \hat{P}_{\bv}$. Let us denote $B_{0} = f(\hat{P}_{\bu\bv}) - f(Q_{U'V})$.
Now, given $\bu$ and $\bv$, $N_{\bu,\bv}(Q_{U'|V})$ is a binomial sum of $|\calT(Q_{U'|V}|\bv)| \doteq e^{n H_{Q}(U'|V)}$ trials and success rate of the exponential order of $e^{-n \RF(Q_{U})}$. Therefore, using the techniques of \cite[Section 6.3]{MERHAV09},
\begin{align}
&-\frac{1}{n} \log \prob \left\{ N_{\bu,\bv}(Q_{U'|V})  \geq  \exp\{n[B_{0} - \xi ]\}  \right\} \nonumber \\
&= \left\{   
\begin{array}{l l}
\left[\RF(Q_{U})-H_{Q}(U'|V)\right]_{+}    & \quad \text{  $\left[H_{Q}(U'|V) - \RF(Q_{U})\right]_{+} \geq B_{0} - \xi$  }   \\
\infty    & \quad \text{  $\left[H_{Q}(U'|V) - \RF(Q_{U})\right]_{+} < B_{0} - \xi$  }
\end{array} \right. \\
&= \left\{   
\begin{array}{l l}
\left[\RF(Q_{U})-H_{Q}(U'|V)\right]_{+}    & \quad \text{  $ \xi \geq B_{0} - \left[H_{Q}(U'|V) - \RF(Q_{U})\right]_{+}$  }   \\
\infty    & \quad \text{  $\xi < B_{0} - \left[H_{Q}(U'|V) - \RF(Q_{U})\right]_{+}$  }
\end{array} \right. ,
\end{align}
and so,
\begin{align}
&\int_{0}^{\infty} e^{-n \xi} \cdot \prob \left\{ N_{\bu,\bv}(Q_{U|V})  \geq  \exp\{n[B_{0} - \xi ]\}  \right\} d\xi \nonumber \\
&\doteq \int_{\left[B_{0} - [H_{Q}(U'|V) - \RF(Q_{U})]_{+}\right]_{+}}^{\infty} e^{-n \xi} \cdot e^{-n [\RF(Q_{U})-H_{Q}(U'|V)]_{+}} d\xi \\
&\doteq \exp \left\{-n \left([\RF(Q_{U})-H_{Q}(U'|V)]_{+} + \left[B_{0} - [H_{Q}(U'|V) - \RF(Q_{U})]_{+}\right]_{+}  \right) \right\} \\
&= \exp \left\{-n  
\left(   
\begin{array}{l l}
\RF(Q_{U})-H_{Q}(U'|V) + [B_{0} ]_{+} & \quad \text{$\RF(Q_{U}) \geq H_{Q}(U'|V)$} \\
\left[\RF(Q_{U}) - H_{Q}(U'|V) + B_{0} \right]_{+} & \quad \text{$\RF(Q_{U}) < H_{Q}(U'|V)$}
\end{array} \right)
\right\}  \\
&= \exp \left\{-n  
\left(   
\begin{array}{l l}
\left[\RF(Q_{U})-H_{Q}(U'|V) + [B_{0}]_{+} \right]_{+} & \quad \text{$\RF(Q_{U}) \geq H_{Q}(U'|V)$} \\
\left[\RF(Q_{U}) - H_{Q}(U'|V) + [B_{0}]_{+} \right]_{+} & \quad \text{$\RF(Q_{U}) < H_{Q}(U'|V)$}
\end{array} \right)
\right\}  \\
&= \exp \left\{-n \cdot \left[\RF(Q_{U}) - H_{Q}(U'|V) + [B_{0}]_{+} \right]_{+}
\right\} .
\end{align}
Finally, we have that
\begin{align}
&\sum_{Q_{U'|V}} \int_{0}^{\infty} e^{-n \xi} \cdot \prob \left\{ N_{\bu,\bv}(Q_{U|V})  \geq  \exp\{n[B_{0} - \xi ]\}  \right\} d\xi  \\
&\doteq \max_{Q_{U'|V}} \exp \left\{-n \cdot \left[\RF(Q_{U}) - H_{Q}(U'|V) + [B_{0}]_{+} \right]_{+}
\right\} \\
&= \exp \left\{-n \cdot \min_{Q_{U'|V}} \left[\RF(Q_{U}) - H_{Q}(U'|V) + [B_{0}]_{+} \right]_{+}
\right\},
\end{align}
thus,
\begin{align}
E(\bu,\bv)
= \min_{Q_{U'|V}} \left[\RF(Q_{U}) - H_{Q}(U'|V) + [f(\hat{P}_{\bu\bv}) - f(Q_{U'V})]_{+} \right]_{+}.
\end{align}

\subsubsection*{Step 2: Averaging Over $\bU$ and $\bV$}
Notice that the exponent function $E(\bu,\bv)$ depends on $(\bu,\bv)$ only via the empirical distribution $\hat{P}_{\bu\bv}$.
Averaging over the source and the SI sequences, now yields
\begin{align}
\mathbb{E} \left\{ P_{\mbox{\tiny e}} (\calB_{n}) \right\} 
&= \sum_{\bu,\bv} P(\bu,\bv) 
\cdot \IND \left\{ \hat{H}_{\bu}(U) \geq \RF(\hat{P}_{\bu}) \right\} \cdot
\exp\left\{-n \cdot E(\hat{P}_{\bu\bv})\right\} \\
&\doteq \sum_{\{Q_{UV}:~ H_{Q}(U) \geq \RF(Q_{U})\}} e^{-n \cdot D(Q_{UV} \| P_{UV})} \cdot \exp\left\{-n \cdot E(Q_{UV})\right\} \\
\label{FinalExponent0}
&\doteq \exp\left\{-n \cdot \min_{\{Q_{UV}:~ H_{Q}(U) \geq \RF(Q_{U})\}}  \left[ D(Q_{UV} \| P_{UV}) +  E(Q_{UV}) \right] \right\} ,
\end{align}
which proves the first point of Theorem \ref{Thm_random_binning_SD_VR}.

\subsubsection*{Step 3: Moving from Stochastic to Deterministic Decoding}

In order to transform the GLD into the general deterministic decoder of 
\begin{align}
\label{DEF_Deterministic_Decoder}
\hat{\bu} = \operatorname*{arg\,max}_{\bu' \in \calB(\bu)}
f( \hat{P}_{\bu'\bv} ),
\end{align}
we just have to multiply $f(\cdot)$, in
\begin{align}
E(Q_{UV})
= \min_{Q_{U'|V}} \left[\RF(Q_{U}) - H_{Q}(U'|V) + [f(Q_{UV}) - f(Q_{U'V})]_{+} \right]_{+},
\end{align} 
by $\beta \geq 0$, and then let $\beta \to \infty$. We find that the overall error exponent of the SD ensemble with the general deterministic decoder of \eqref{DEF_Deterministic_Decoder} is given by  
\begin{align}
E(P) 
&= \min_{\{Q_{UV}:~ H_{Q}(U) \geq \RF(Q_{U})\}}  \left[ D(Q_{UV} \| P_{UV}) +  \tilde{E}(Q_{UV}) \right],
\end{align}
where,
\begin{align}
\label{Minimum1}
\tilde{E}(Q_{UV}) 
&= \min_{\{Q_{U'|V}:~ f(Q_{U'V}) \geq f(Q_{UV})\}} \left[ \RF(Q_{U}) - H_{Q}(U'|V) \right]_{+} .
\end{align}

\subsubsection*{Step 4: A Fundamental Limitation on the Error Exponent}

Note that the minimum in \eqref{Minimum1} can be upper--bounded by choosing a specific distribution in the feasible set.  
In \eqref{Minimum1}, we take $Q_{U'|V} = Q_{U|V}$ and then
\begin{align}
\tilde{E}(Q_{UV}) 
&\leq \left[ \RF(Q_{U}) - H_{Q}(U|V) \right]_{+}. 
\end{align}
Hence, the overall error exponent is upper--bounded as 
\begin{align}
\label{Upper_Bound}
E(P) &\leq \min_{\{Q_{UV}:~ H_{Q}(U) \geq \RF(Q_{U})\}}  \left[ D(Q_{UV} \| P_{UV}) +  \left[ \RF(Q_{U}) - H_{Q}(U|V) \right]_{+} \right].
\end{align}

\subsubsection*{Step 5: An Optimal Universal Decoder}

We prove that the upper bound of \eqref{Upper_Bound} is attainable by choosing the universal decoding metric $f(Q_{UV}) = - H_{Q}(U|V)$.
Now, we get for \eqref{Minimum1} 
\begin{align}
\tilde{E}(Q_{UV}) 
&= \min_{\{Q_{U'|V}:~ f(Q_{U'V}) \geq f(Q_{UV})\}} \left[ \RF(Q_{U}) - H_{Q}(U'|V) \right]_{+} \\
&= \min_{\{Q_{U'|V}:~ H_{Q}(U|V) \geq H_{Q}(U'|V) \}} \left[ \RF(Q_{U}) - H_{Q}(U'|V) \right]_{+} \\
&= \left[ \RF(Q_{U}) - H_{Q}(U|V) \right]_{+} ,
\end{align}
which completes the proof of Theorem \ref{Thm_random_binning_SD_VR}.

\section*{Appendix B}
\renewcommand{\theequation}{B.\arabic{equation}}
\setcounter{equation}{0}  
\subsection*{Proof of Theorem \ref{SW_THM}}
\subsubsection*{Lower Bound on the Error Exponent}
Our starting point is the following inequality, for any $\rho > 0$,
\begin{align} \label{StartingPoint}
\mathbb{E} [\log P_{\mbox{\tiny e}} (\calB_{n}) ]
\leq \log \left( \mathbb{E} [P_{\mbox{\tiny e}} (\calB_{n})]^{1/\rho}\right)^{\rho},
\end{align}
which is due to the following considerations. 
First, for a positive random variable $X$, the function
\begin{align}
f(\rho) = \log \left( \mathbb{E} \left[X^{1/\rho}\right] \right)^{\rho}
\end{align}
is monotonically decreasing, and second, by L'Hospital's rule,
\begin{align}
\lim_{\rho \to \infty} \log \left( \mathbb{E} \left[X^{1/\rho}\right]   \right)^{\rho} = \mathbb{E}[\log X].
\end{align} 
Recall that the error probability is given by 
\begin{align}
P_{\mbox{\tiny e}} (\calB_{n}) 
&= \sum_{\bu,\bv} P(\bu,\bv) 
\cdot \IND \left\{ \hat{H}_{\bu}(U) \geq \RF(\hat{P}_{\bu}) \right\} \cdot
\frac{\sum_{\bu' \in \calB(\bu)\cap\calT(\bu), \bu' \neq \bu} \exp\{nf(\hat{P}_{\bu'\bv})\}}{\sum_{\tilde{\bu} \in \calB(\bu)\cap\calT(\bu)} \exp\{nf(\hat{P}_{\tilde{\bu}\bv})\}}. 
\end{align}
Let
\begin{align}
\label{Z_DEF}
Z_{\bu}(\bv) = \sum_{\tilde{\bu} \in \calB(\bu)\cap\calT(\bu), \tilde{\bu} \neq \bu} \exp\{nf(\hat{P}_{\tilde{\bu}\bv})\},
\end{align}
fix $\epsilon>0$ arbitrarily small, and for every $\bu \in \calU^{n}$ and $\bv \in \calV^{n}$, define the set
\begin{align}
\label{B_DEF1}
B_{\epsilon}(\bu,\bv) = \left\{\calB_{n}:~ Z_{\bu}(\bv) \leq \exp\{n \alpha(R+\epsilon,\hat{P}_{\bu},\hat{P}_{\bv})\} \right\}.
\end{align}
Following the result of \cite[Appendix B]{MERHAV2017}, we prove the following modification in Appendix C.
\begin{lemma} \label{Large_Deviations_Z}
	Let $\epsilon>0$ be arbitrarily small. Then, for every $\bu \in \calU^{n}$ and $\bv \in \calV^{n}$,
	\begin{equation}
	\prob \left\{Z_{\bu}(\bv) \leq \exp\{n \alpha(R+\epsilon,\hat{P}_{\bu},\hat{P}_{\bv})\}  \right\} \leq \exp \{-e^{n\epsilon} +n\epsilon +1\}.
	\end{equation}
\end{lemma}
So, by the union bound,
\begin{align}
\label{B_UNION_DEF}
\prob \left\{\bigcup_{\bu \in \calU^{n}} \bigcup_{\bv \in \calV^{n}} B_{\epsilon}(\bu,\bv)\right\}
\DEF  \prob \left\{B_{\epsilon}\right\} 
&\leq \sum_{\bu \in \calU^{n}} \sum_{\bv \in \calV^{n}} \prob \left\{B_{\epsilon}(\bu,\bv)\right\} \\
&\leq \sum_{\bu \in \calU^{n}} \sum_{\bv \in \calV^{n}} \exp\{-e^{n\epsilon} + n\epsilon + 1\} \\
&= |\calU \times \calV|^{n} \cdot \exp\{-e^{n\epsilon} + n\epsilon + 1\},
\end{align}
which still decays double--exponentially fast. 
Recall that $\calQ = \{Q_{UU'}:~ Q_{U}=Q_{U'}\}$. Then, for any $\rho \geq 1$
\begin{align}
&\mathbb{E} \left\{ \left[P_{\mbox{\tiny e}} (\calB_{n}) \right]^{1/ \rho} \right\}  \nonumber \\
&=\Exp \left\{P_{\mbox{\tiny e}}(\calB_{n})^{1/\rho} \cdot \IND\{B_{\epsilon}^{\mbox{\tiny c}} \} \right\} + \Exp \left\{P_{\mbox{\tiny e}}(\calB_{n})^{1/\rho} \cdot \IND\{B_{\epsilon} \} \right\} \\
&\leq \mathbb{E} \left\{ \left[
\sum_{\bu,\bv} P(\bu,\bv) \IND \left\{ \hat{H}_{\bu}(U) \geq \RF(\hat{P}_{\bu}) \right\} 
\frac{\sum_{\bu' \in \calB(\bu)\cap\calT(\bu), \bu' \neq \bu} \exp\{nf(\hat{P}_{\bu'\bv})\}}{\exp\{nf(\hat{P}_{\bu\bv})\} + Z_{\bu}(\bv)} \right]^{1/ \rho} \IND\{B_{\epsilon}^{\mbox{\tiny c}} \} \right\} \nn \\
&~~+ \prob\{B_{\epsilon} \} \\
\label{E1}
&\leq \mathbb{E} \left\{ \left[
\sum_{\bu,\bv} \sum_{\bu' \in \calB(\bu)\cap\calT(\bu), \bu' \neq \bu} P(\bu,\bv) \IND \left\{ \hat{H}_{\bu}(U) \geq \RF(\hat{P}_{\bu}) \right\}
\right. \right. \nonumber \\ & \left. \left. ~~~~~ \times
\min \left\{1, \frac{\exp\{nf(\hat{P}_{\bu'\bv})\}}{\exp\{nf(\hat{P}_{\bu\bv})\} + \exp\{n \alpha(R+\epsilon,\hat{P}_{\bu},\hat{P}_{\bv})\}} \right\}  \right]^{1/ \rho} \right\} \nn \\
&~~+ |\calU \times \calV|^{n} \cdot \exp\{-e^{n\epsilon} + n\epsilon + 1\} \\
\label{Use_the_MOT_2}
&\doteq \mathbb{E} \left\{ \left[
\sum_{\bu} \sum_{\bu' \in \calB(\bu)\cap\calT(\bu), \bu' \neq \bu} P(\bu) \IND \left\{ \hat{H}_{\bu}(U) \geq \RF(\hat{P}_{\bu}) \right\} \right. \right. \nonumber \\
& \left. \left. ~~ \times \sum_{\bv} P(\bv|\bu) \exp \left\{-n \cdot \big[\max\{f(\hat{P}_{\bu\bv}), \alpha(R+\epsilon,\hat{P}_{\bu},\hat{P}_{\bv})\} - f(\hat{P}_{\bu'\bv}) \big]_{+} \right\}  \right]^{1/ \rho} \right\} \\
\label{E2}
&\doteq \mathbb{E} \left\{ \left[
\sum_{\bu} \sum_{\bu' \in \calB(\bu)\cap\calT(\bu), \bu' \neq \bu} P(\bu) \cdot \IND \left\{ \hat{H}_{\bu}(U) \geq \RF(\hat{P}_{\bu}) \right\} \cdot
\exp \left\{-n \cdot \Lambda(\hat{P}_{\bu\bu'},R+\epsilon) \right\}  
\right]^{1/ \rho} \right\} \\
\label{E3}
&= \mathbb{E} \left\{ \left[
\sum_{\{Q_{UU'}\in \calQ:~H_{Q}(U)\geq R(Q_{U})\}} N(Q_{UU'}) \cdot  e^{n \mathbb{E}_{Q} [\log P(U)]} \cdot 
\exp \left\{-n \cdot \Lambda(Q_{UU'},R+\epsilon) \right\}  
\right]^{1/ \rho} \right\} \\
\label{continue4}
&\doteq \sum_{\{Q_{UU'}\in \calQ:~H_{Q}(U)\geq R(Q_{U})\}} \mathbb{E} \left\{ \left[ N(Q_{UU'}) \right]^{1/ \rho} \right\} \cdot  e^{n (\mathbb{E}_{Q} [\log P(U)])/ \rho} \cdot 
\exp \left\{-n \cdot \Lambda(Q_{UU'},R+\epsilon) / \rho \right\} ,
\end{align}
where \eqref{E1} is due to Lemma \ref{Large_Deviations_Z}, \eqref{E2} is by the method of types and the definition of $\Lambda(Q_{UU'},R)$ in \eqref{DEF_Lambda}, and in \eqref{E3} we used the definition of $N(Q_{UU'})$ in \eqref{DEF_ENUMERATOR}.
Therefore, our next task is to evaluate the $1/ \rho$--th moment of $N(Q_{UU'})$.
Let us define
\begin{align}
N_{\bu}(Q_{U'|U}) = \sum_{\bu' \in \calT(Q_{U'|U}|\bu)} \IND \left\{\calB(\bu') = \calB(\bu) \right\} .
\end{align}
For a given $\rho \geq 1$, let $s \in [1, \rho]$. Then, 
\begin{align}
\mathbb{E} \left\{ \left[N(Q_{UU'})\right]^{1 / \rho}  \right\} 
&=  \mathbb{E} \left\{ \left[ \sum_{\bu \in \calT(Q_{U})} N_{\bu}(Q_{U'|U}) \right]^{1 / \rho}   \right\}    \\
&=  \mathbb{E} \left\{ \left( \left[ \sum_{\bu \in \calT(Q_{U})} N_{\bu}(Q_{U'|U}) \right]^{1 / s} \right)^{s / \rho}   \right\}    \\
\label{Exp3}
&\leq \mathbb{E} \left\{ \left( \sum_{\bu \in \calT(Q_{U})} \left[ N_{\bu}(Q_{U'|U}) \right]^{1 / s} \right)^{s / \rho}   \right\}    \\
\label{Exp4}
&\leq \left( \mathbb{E} \left\{  \sum_{\bu \in \calT(Q_{U})} \left[ N_{\bu}(Q_{U'|U}) \right]^{1 / s}  \right\} \right)^{s / \rho}    \\
\label{Stage23}
&= \left(  \sum_{\bu \in \calT(Q_{U})} \mathbb{E} \left\{ \left[ N_{\bu}(Q_{U'|U}) \right]^{1 / s} \right\}  \right)^{s / \rho}, 
\end{align}
where \eqref{Exp4} follows from Jensen's inequality.
Now, $N_{\bu}(Q_{U'|U})$ is a binomial random variable with $|\calT(Q_{U'|U}|\bu)| \doteq e^{n H_{Q}(U'|U)}$ trials and success rate which is of the exponential order of $e^{-nR}$.
We have that \cite[Sec.\ 6.3]{MERHAV09} 
\begin{align}
\mathbb{E} \left\{ \left[ N_{\bu}(Q_{U'|U}) \right]^{1/s}  \right\} \doteq    
\left\{   
\begin{array}{l l}
\exp\{n[H_{Q}(U'|U) - R]/s\}   & \quad \text{  $H_{Q}(U'|U) \geq R$  }\\
\exp\{n[H_{Q}(U'|U) - R]\}   & \quad \text{  $H_{Q}(U'|U) < R$  } 
\end{array} \right.  ,
\end{align}
and so,
\begin{align}
\mathbb{E} \left\{ \left[N(Q_{UU'})\right]^{1 / \rho}  \right\}  
&\leq e^{n H_{Q}(U) \cdot  s/ \rho} \cdot  \left( \mathbb{E} \left\{ \left[ N_{\bu}(Q_{U'|U}) \right]^{1/s}  \right\} \right)^{s/ \rho}  \\
&\doteq e^{n H_{Q}(U) \cdot s/ \rho} \cdot
\left\{   
\begin{array}{l l}
\exp\{n [H_{Q}(U'|U) - R] / \rho\}   & \quad \text{  $H_{Q}(U'|U) \geq R$  }\\
\exp\{n [H_{Q}(U'|U) - R] s/ \rho \}   & \quad \text{  $H_{Q}(U'|U) < R$  } 
\end{array} \right. \\
&= 
\left\{   
\begin{array}{l l}
\exp\{n [H_{Q}(U) \cdot s + H_{Q}(U'|U) - R]/ \rho \}   & \quad \text{  $H_{Q}(U'|U) \geq R$  }\\
\exp\{n  [H_{Q}(U) + H_{Q}(U'|U) - R]s/ \rho \}   & \quad \text{  $H_{Q}(U'|U) < R$  } 
\end{array} \right. \\
&= 
\left\{   
\begin{array}{l l}
\exp\{n [H_{Q}(U) \cdot s + H_{Q}(U'|U) - R]/ \rho \}   & \quad \text{  $H_{Q}(U'|U) \geq R$  }\\
\exp\{n  [H_{Q}(U,U') - R]s/ \rho \}   & \quad \text{  $H_{Q}(U'|U) < R$  } 
\end{array} \right. .
\end{align}
After optimizing over $s$, we get
\begin{align}
&\frac{1}{n} \log \mathbb{E} \left\{ \left[N(Q_{UU'})\right]^{1 / \rho}  \right\} \nonumber \\
&\leq  \min_{1 \leq s \leq \rho}             \left\{   
\begin{array}{l l}
\left[H_{Q}(U) \cdot s + H_{Q}(U'|U) - R\right]/ \rho   & \quad \text{  $H_{Q}(U'|U) \geq R$  }\\
\left[H_{Q}(U,U') - R\right]s/ \rho    & \quad \text{  $H_{Q}(U'|U) < R, H_{Q}(U,U') \geq R$  }   \\
\left[H_{Q}(U,U') - R\right]s/ \rho    & \quad \text{  $H_{Q}(U'|U) < R, H_{Q}(U,U') < R$  }
\end{array} \right.  \\
&=      \left\{   
\begin{array}{l l}
\left[H_{Q}(U) + H_{Q}(U'|U) - R\right]/ \rho   & \quad \text{  $H_{Q}(U'|U) \geq R$  }\\
\left[H_{Q}(U,U') - R\right]/ \rho   & \quad \text{  $H_{Q}(U'|U) < R, H_{Q}(U,U') \geq R$  }   \\
\left[H_{Q}(U,U') - R\right] \rho  / \rho    & \quad \text{  $H_{Q}(U'|U) < R, H_{Q}(U,U') < R$  }
\end{array} \right.  \\
\label{upup19}
&=      \left\{   
\begin{array}{l l}
\left[H_{Q}(U,U') - R\right] / \rho    & \quad \text{  $H_{Q}(U,U') \geq R$  }   \\
\left[H_{Q}(U,U') - R\right]    & \quad \text{  $H_{Q}(U,U') < R$  }
\end{array} \right. , 
\end{align}
which gives, after raising to the $\rho$--th power,
\begin{align}
\left( \mathbb{E} \left\{ \left[N(Q_{UU'})\right]^{1 / \rho}  \right\} \right)^{\rho} 
&\leq     \left\{   
\begin{array}{l l}
\exp\{n\left[H_{Q}(U,U') - R\right] \}   & \quad \text{  $H_{Q}(U,U') \geq R$  }   \\
\exp\{n\left[H_{Q}(U,U') - R\right] \cdot \rho \}   & \quad \text{  $H_{Q}(U,U') < R$  }
\end{array} \right. \\
\label{ToUse2}
&= \exp\{n([H_{Q}(U,U') - R]_{+} - \rho [R-H_{Q}(U,U')]_{+}) \}. 
\end{align}
Let us denote $F(Q,R,\rho) = [H_{Q}(U,U') - R]_{+} - \rho [R-H_{Q}(U,U')]_{+}$.
Continuing now from \eqref{continue4},
\begin{align}
&\left( \mathbb{E} \left\{ \left[P_{\mbox{\tiny e}} (\calB_{n}) \right]^{1/ \rho} \right\} \right)^{\rho} \nonumber \\ 
&\lexe \left( \sum_{\{Q_{UU'}\in \calQ:~ H_{Q}(U) \geq R(Q_{U})\}} \mathbb{E} \left\{ \left[ N(Q_{UU'}) \right]^{1/ \rho} \right\} \cdot  e^{n [\mathbb{E}_{Q} \log P(U)]/ \rho} \cdot \exp \left\{-n \cdot \Lambda(Q_{UU'},R+\epsilon) / \rho \right\} \right)^{\rho} \\
\label{E4}
&\doteq  \sum_{\{Q_{UU'}\in \calQ:~ H_{Q}(U) \geq R(Q_{U})\}} \left( \mathbb{E} \left\{ \left[ N(Q_{UU'}) \right]^{1/ \rho} \right\} \right)^{\rho} \cdot  e^{n \mathbb{E}_{Q} \log P(U)} \cdot \exp \left\{-n \cdot \Lambda(Q_{UU'},R+\epsilon) \right\} \\
\label{E5}
&\leq  \sum_{\{Q_{UU'}\in \calQ:~ H_{Q}(U) \geq R(Q_{U})\}} \exp\{n(F(Q,R,\rho) + \mathbb{E}_{Q} [\log P(U)] - \Lambda(Q_{UU'},R+\epsilon)) \}  \\
&\doteq  \exp\left\{-n \cdot \min_{\{Q_{UU'}\in \calQ:~ H_{Q}(U) \geq R(Q_{U})\}} (\Lambda(Q_{UU'},R+\epsilon) - F(Q,R,\rho) - \mathbb{E}_{Q} [\log P(U)] ) \right\} .
\end{align}
where \eqref{E5} follows from \eqref{ToUse2}.
Finally, it follows by \eqref{StartingPoint} that
\begin{align}
&\liminf_{n \to \infty} -\frac{1}{n} \mathbb{E} [\log P_{\mbox{\tiny e}} (\calB_{n}) ] \nn \\
&\geq \liminf_{n \to \infty} -\frac{1}{n} \log \left( \mathbb{E} [P_{\mbox{\tiny e}} (\calB_{n})]^{1/\rho}\right)^{\rho} \\
&\geq \min_{\{Q_{UU'}\in \calQ:~ H_{Q}(U) \geq R(Q_{U})\}} (\Lambda(Q_{UU'},R+\epsilon) - F(Q,R,\rho) - \mathbb{E}_{Q} [\log P(U)] ).
\end{align}  
Letting $\rho$ grow without bound yields that
\begin{align}
&\liminf_{n \to \infty} -\frac{1}{n} \mathbb{E} [\log P_{\mbox{\tiny e}} (\calB_{n}) ] \nn \\
&\geq \min_{\{Q_{UU'}\in \calQ:~ H_{Q}(U) \geq R(Q_{U}),~H_{Q}(U,U') \geq R(Q_{U})\}} (\Lambda(Q_{UU'},R+\epsilon) - H_{Q}(U,U') + R(Q_{U}) - \mathbb{E}_{Q} [\log P(U)] ) \\
&= \min_{\{Q_{UU'}\in \calQ:~ H_{Q}(U) \geq R(Q_{U})\}} (\Lambda(Q_{UU'},R+\epsilon) - H_{Q}(U,U') + R(Q_{U}) - \mathbb{E}_{Q} [\log P(U)] ).
\end{align}
Due to the arbitrariness of $\epsilon>0$, we have proved that
\begin{align}
&\liminf_{n \to \infty} -\frac{1}{n} \mathbb{E} [\log P_{\mbox{\tiny e}} (\calB_{n}) ] \nn \\
&\geq \min_{\{Q_{UU'}\in \calQ:~ H_{Q}(U) \geq R(Q_{U})\}} (\Lambda(Q_{UU'},R) - H_{Q}(U,U') + R(Q_{U}) - \mathbb{E}_{Q} [\log P(U)] ).
\end{align}
completing half of the proof of Theorem \ref{SW_THM}.

\subsubsection*{Upper Bound on the Error Exponent}

Consider a joint distribution $Q_{UU'}$, that satisfies $H_{Q}(U,U') > R$, and define the event $\calE(Q_{UU'})=\{\calB_{n}:~ N(Q_{UU'})< \exp\{n[H_{Q}(U,U') - R - \epsilon]\}\}$. We want to show that $\prob\{\calE(Q_{UU'})\}$ is small. Consider the following:
\begin{align}
\prob\{\calE(Q_{UU'})\} 
&= \prob \{ N(Q_{UU'})< \exp\{n[H_{Q}(U,U') - R - \epsilon]\} \} \\
&= \prob \{ N(Q_{UU'})< e^{-n \epsilon} \cdot \mathbb{E} \{N(Q_{UU'})\} \} \\
&= \prob \left\{ \frac{N(Q_{UU'})}{\mathbb{E} \{N(Q_{UU'})\}} - 1< -(1-e^{-n \epsilon}) \right\} \\
&\leq \prob \left\{ \left[\frac{N(Q_{UU'}) - \mathbb{E} \{N(Q_{UU'})\}}{\mathbb{E} \{N(Q_{UU'})\}} \right]^{2} >  (1-e^{-n \epsilon})^{2} \right\} \\
&\leq \frac{\mbox{Var}\{N(Q_{UU'})\}}{(1-e^{-n \epsilon})^{2} \cdot \mathbb{E}^{2} \{N(Q_{UU'})\}}.
\end{align}
Let us use the shorthand notations $\calI(\bu,\bu')=\IND \left\{\calB(\bu') = \calB(\bu) \right\}$, $K = |\calT(Q_{UU'})|$, and $p = e^{-nR}$. Concerning the variance of $N(Q_{UU'})$, we have the following
\begin{align}
&\mbox{Var}\{N(Q_{UU'})\} \nonumber \\
&= \mathbb{E}\{N^{2}(Q_{UU'})\} - \mathbb{E}^{2}\{N(Q_{UU'})\} \\
&= \mathbb{E} \left\{ \left[ \sum_{(\bu,\bu') \in \calT(Q_{UU'})} \calI(\bu,\bu') \right] \times \left[ \sum_{(\tilde{\bu},\hat{\bu}) \in \calT(Q_{UU'})} \calI(\tilde{\bu},\hat{\bu}) \right] \right\} - (Kp)^{2} \\
&= \sum_{(\bu,\bu') \in \calT(Q_{UU'})} \sum_{(\tilde{\bu},\hat{\bu}) \in \calT(Q_{UU'})} \mathbb{E} \left\{\calI(\bu,\bu') \calI(\tilde{\bu},\hat{\bu})    
\right\} - (Kp)^{2} \\
&= \sum_{(\bu,\bu') \in \calT(Q_{UU'})}  \mathbb{E} \left\{\calI^{2}(\bu,\bu') \right\} +
\sum_{\substack{(\bu,\bu'),(\tilde{\bu},\hat{\bu}) \in \calT(Q_{UU'}) \\ (\bu,\bu') \neq (\tilde{\bu},\hat{\bu})}} \mathbb{E} \left\{\calI(\bu,\bu') \calI(\tilde{\bu},\hat{\bu})    
\right\} - (Kp)^{2} \\
&= Kp + K (K-1) p^{2} - (Kp)^{2} \\
&= K p (1-p) \\
&\doteq \exp\{n[H_{Q}(U,U') - R]\},
\end{align}
and hence,
\begin{align}
\prob\{\calE(Q_{UU'})\} 
&\lexe \frac{\exp\{n[H_{Q}(U,U') - R]\}}{\exp\{n[2H_{Q}(U,U') - 2R]\}} \\
&= \exp\{-n[H_{Q}(U,U') - R]\},
\end{align}
which decays to zero since we have assumed that $H_{Q}(U,U') > R$. Furthermore, if $H_{Q}(U,U') \geq R + \epsilon$, then $\prob\{\calE(Q_{UU'})\}$ tends to zero at least as fast as $e^{-n \epsilon}$. Now, for a given $\epsilon > 0$, and a given joint type $Q_{UU'V}$, such that $H_{Q}(U,U') \geq R + \epsilon$, let us define
\begin{align}
Z_{\bu\bu'}(\bv) = \sum_{\tilde{\bu} \in \calB(\bu) \cap \calT(\bu), \tilde{\bu} \neq \bu,\bu'} \exp\{nf(\hat{P}_{\tilde{\bu}\bv})\},
\end{align}   
and
\begin{align}
\calG_{n}(Q_{UU'V}) = \Bigg\{\calB_{n}:&~ \sum_{(\bu,\bu') \in \calT(Q_{UU'})} \IND \{\calB(\bu') = \calB(\bu)\} \times \nonumber \\
& \sum_{\bv \in \calT(Q_{V|UU'}|\bu,\bu')} \IND \left\{Z_{\bu\bu'}(\bv) 
\leq e^{n[\alpha(R-2\epsilon,Q_{U},Q_{V}) + \epsilon]} \right\} \geq \nonumber \\
&\exp\{n[H_{Q}(U,U') - R - 3\epsilon/2]\} \cdot |\calT(Q_{V|UU'}|\bu,\bu')| \Bigg\},
\end{align}
where $(\bu,\bu')$ in the expression $|\calT(Q_{V|UU'}|\bu,\bu')|$ should be understood as any pair of source sequences in $\calT(Q_{UU'})$. Next, we define 
\begin{align}
\calG_{n} = \bigcap_{\{Q_{UU'V}:~ H_{Q}(U,U')\geq R + \epsilon\}}
[\calG_{n}(Q_{UU'V}) \cap \calE^{\mbox{\tiny c}}(Q_{UU'})].
\end{align}
We start by proving that $\prob\{\calG_{n}\} \to 1$ as $n \to \infty$, or equivalently, that $\prob\{\calG_{n}^{\mbox{\tiny c}}\} \to 0$ as $n \to \infty$. Now,  
\begin{align}
\prob\{\calG_{n}^{\mbox{\tiny c}}\} 
&= \prob \left\{ \bigcup_{\{Q_{UU'V}:~ H_{Q}(U,U')\geq R + \epsilon\}} [\calG_{n}^{\mbox{\tiny c}}(Q_{UU'V}) \cup \calE(Q_{UU'})] \right\} \\
&\leq  \sum_{\{Q_{UU'V}:~ H_{Q}(U,U')\geq R + \epsilon\}} \prob \left\{ \calG_{n}^{\mbox{\tiny c}}(Q_{UU'V}) \cup \calE(Q_{UU'}) \right\} \\
&=  \sum_{\{Q_{UU'V}:~ H_{Q}(U,U')\geq R + \epsilon\}} 
[\prob \left\{\calE(Q_{UU'}) \right\} + \prob \left\{ \calG_{n}^{\mbox{\tiny c}}(Q_{UU'V}) \cap \calE^{\mbox{\tiny c}}(Q_{UU'}) \right\}].
\end{align}
The last summation contains a polynomial number of terms. If we prove that the summand tends to zero exponentially with $n$, then $\prob\{\calG_{n}^{\mbox{\tiny c}}\} \to 0$ as $n \to \infty$. The first term in the summand, $\prob \left\{\calE(Q_{UU'}) \right\}$, has already been proved to be upper bounded by $e^{-n \epsilon}$. Concerning the second term, we have the following 
\begin{align}
&\prob \left\{ \calG_{n}^{\mbox{\tiny c}}(Q_{UU'V}) \cap \calE^{\mbox{\tiny c}}(Q_{UU'}) \right\} \nonumber  \\
&= \prob \Bigg[ \sum_{(\bu,\bu') \in \calT(Q_{UU'})} \IND \{\calB(\bu') = \calB(\bu)\} \cdot \sum_{\bv \in \calT(Q_{V|UU'}|\bu,\bu')} \IND \left\{Z_{\bu\bu'}(\bv) 
\leq e^{n[\alpha(R-2\epsilon,Q_{U},Q_{V}) + \epsilon]} \right\} < \nonumber \\
&~~~~~~\exp\{n[H_{Q}(U,U') - R - 3\epsilon/2]\} \cdot |\calT(Q_{V|UU'}|\bu,\bu')|, \nonumber \\
&~~~~~~~~~~~~~~~~~~~~~~~~~~~~~~~~~~~~~~~~~~~~~~~~~~~~~~~~~~ N(Q_{UU'}) \geq \exp\{n[H_{Q}(U,U') - R - \epsilon]\}\Bigg] \\
\label{E7}
&= \prob \Bigg[ \sum_{(\bu,\bu') \in \calT(Q_{UU'})} \IND \{\calB(\bu') = \calB(\bu)\} \cdot \sum_{\bv \in \calT(Q_{V|UU'}|\bu,\bu')} \IND \left\{Z_{\bu\bu'}(\bv) 
> e^{n[\alpha(R-2\epsilon,Q_{U},Q_{V}) + \epsilon]} \right\} > \nonumber \\
&~~~~~~[N(Q_{UU'}) - \exp\{n[H_{Q}(U,U') - R - 3\epsilon/2]\}] \cdot |\calT(Q_{V|UU'}|\bu,\bu')|, \nonumber \\
&~~~~~~~~~~~~~~~~~~~~~~~~~~~~~~~~~~~~~~~~~~~~~~~~~~~~~~~~~~ N(Q_{UU'}) \geq \exp\{n[H_{Q}(U,U') - R - \epsilon]\}\Bigg] \\
\label{E8}
&\leq \prob \Bigg[ \sum_{(\bu,\bu') \in \calT(Q_{UU'})} \IND \{\calB(\bu') = \calB(\bu)\} \cdot \sum_{\bv \in \calT(Q_{V|UU'}|\bu,\bu')} \IND \left\{Z_{\bu\bu'}(\bv) 
> e^{n[\alpha(R-2\epsilon,Q_{U},Q_{V}) + \epsilon]} \right\} > \nonumber \\
&~~~~~~[\exp\{n[H_{Q}(U,U') - R - \epsilon]\} - \exp\{n[H_{Q}(U,U') - R - 3\epsilon/2]\}] \cdot |\calT(Q_{V|UU'}|\bu,\bu')|, \nonumber \\
&~~~~~~~~~~~~~~~~~~~~~~~~~~~~~~~~~~~~~~~~~~~~~~~~~~~~~~~~~~ N(Q_{UU'}) \geq \exp\{n[H_{Q}(U,U') - R - \epsilon]\}\Bigg] \\
\label{E9}
&\leq \prob \Bigg[ \sum_{(\bu,\bu') \in \calT(Q_{UU'})} \IND \{\calB(\bu') = \calB(\bu)\} \cdot \sum_{\bv \in \calT(Q_{V|UU'}|\bu,\bu')} \IND \left\{Z_{\bu\bu'}(\bv) 
> e^{n[\alpha(R-2\epsilon,Q_{U},Q_{V}) + \epsilon]} \right\} > \nonumber \\
&~~~~[\exp\{n[H_{Q}(U,U') - R - \epsilon]\} - \exp\{n[H_{Q}(U,U') - R - 3\epsilon/2]\}] \cdot |\calT(Q_{V|UU'}|\bu,\bu')| \Bigg] \\
\label{E10}
&\leq \frac{\mathbb{E} \left\{ \sum_{(\bu,\bu') \in \calT(Q_{UU'})} \IND \{\calB(\bu') = \calB(\bu)\} \cdot \sum_{\bv \in \calT(Q_{V|UU'}|\bu,\bu')} \IND \left\{Z_{\bu\bu'}(\bv) 
	> e^{n[\alpha(R-2\epsilon,Q_{U},Q_{V}) + \epsilon]} \right\} \right\}}{[\exp\{n[H_{Q}(U,U') - R - \epsilon]\} - \exp\{n[H_{Q}(U,U') - R - 3\epsilon/2]\}] \cdot |\calT(Q_{V|UU'}|\bu,\bu')|} \\
\label{E11}
&\lexe \frac{ |\calT(Q_{UU'})| \cdot |\calT(Q_{V|UU'}|\bu,\bu')| \cdot \prob \left\{\calB(\bu') = \calB(\bu),~Z_{\bu\bu'}(\bv) 
	> e^{n[\alpha(R-2\epsilon,Q_{U},Q_{V}) + \epsilon]} \right\}}{\exp\{n[H_{Q}(U,U') - R - \epsilon]\} \cdot |\calT(Q_{V|UU'}|\bu,\bu')|} \\
\label{E12}
&\doteq \frac{ \exp\{nH_{Q}(U,U')\} \cdot \prob \left\{\calB(\bu') = \calB(\bu) \right\} \cdot \prob \left\{Z_{\bu\bu'}(\bv) > e^{n[\alpha(R-2\epsilon,Q_{U},Q_{V}) + \epsilon]} \right\} }{\exp\{n[H_{Q}(U,U') - R - \epsilon]\}} \\
\label{E13}
&= e^{n \epsilon} \cdot \prob \left\{Z_{\bu\bu'}(\bv) > e^{n[\alpha(R-2\epsilon,Q_{U},Q_{V}) + \epsilon]} \right\},
\end{align}
where \eqref{E8} follows by using the second event $N(Q_{UU'}) \geq \exp\{n[H_{Q}(U,U') - R - \epsilon]\}$ to increase the first event inside the probability in \eqref{E7}, \eqref{E9} is true since the second event in \eqref{E8} was omitted, \eqref{E10} follows from Markov's inequality, and \eqref{E12} is due to the independence between the two events inside the probability in \eqref{E11}.   
As for the probability in \eqref{E13},
\begin{align}
&\prob \left\{Z_{\bu\bu'}(\bv) > e^{n[\alpha(R-2\epsilon,Q_{U},Q_{V}) + \epsilon]} \right\} \nonumber \\
&= \prob \left\{\sum_{Q_{U|V}} N(Q_{UV}) e^{nf(Q_{UV})} > e^{n[\alpha(R-2\epsilon,Q_{U},Q_{V}) + \epsilon]} \right\} \\
&\doteq \max_{Q_{U|V}} \prob \left\{ N(Q_{UV}) > \exp\{n[\alpha(R-2\epsilon,Q_{U},Q_{V}) + \epsilon - f(Q_{UV})]\} \right\} \\
&\doteq e^{-n E},
\end{align}
where $N(Q_{UV})$ is the number of source sequences within $\calB(\bu)$, other than $\bu$ and $\bu'$, that fall in the conditional type class $\calT(Q_{U|V}|\bv)$, which is a binomial random variable with $e^{nH_{Q}(U|V)}-2$ trials and success rate of exponential order $e^{-nR}$, and hence,
\begin{align}
E
&= \min_{Q_{U|V}} \left\{ 
\begin{array}{l l}
[R - H_{Q}(U|V)]_{+} & \text{$f(Q_{UV}) + [H_{Q}(U|V) - R]_{+}  \geq \alpha(R-2\epsilon,Q_{U},Q_{V}) + \epsilon$} \\
\infty     & \text{$f(Q_{UV}) + [H_{Q}(U|V) - R]_{+}  < \alpha(R-2\epsilon,Q_{U},Q_{V}) + \epsilon$} 
\end{array} \right. \\
&= \min_{\{Q_{U|V}:~  f(Q_{UV}) + [H_{Q}(U|V) - R]_{+}  \geq \alpha(R-2\epsilon,Q_{U},Q_{V}) + \epsilon\}} [R - H_{Q}(U|V)]_{+}.
\end{align}
By definition of the function $\alpha(R,Q_{U},Q_{V})$, the set $\{Q_{U|V}:~  f(Q_{UV}) + [H_{Q}(U|V) - R]_{+}  \geq \alpha(R-2\epsilon,Q_{U},Q_{V}) + \epsilon\}$ is a subset of $\{Q_{U|V}:~  H_{Q}(U|V) \leq R - 2\epsilon \}$. Thus,
\begin{align}
E \geq \min_{\{Q_{U|V}:~  H_{Q}(U|V) \leq R - 2\epsilon \}} [R - H_{Q}(U|V)]_{+}
\geq 2\epsilon,
\end{align}  
and hence, $\prob \left\{Z_{\bu\bu'}(\bv) > e^{n[\alpha(R-2\epsilon,Q_{U},Q_{V}) + \epsilon]} \right\} \lexe e^{-2n \epsilon}$, which provides
\begin{align}
\prob \left\{ \calG_{n}^{\mbox{\tiny c}}(Q_{UU'V}) \cap \calE^{\mbox{\tiny c}}(Q_{UU'}) \right\} \lexe e^{n \epsilon} \cdot e^{-2n \epsilon} = e^{-n \epsilon},
\end{align}
which proves that $\prob\{\calG_{n}\} \to 1$ as $n \to \infty$. 
Now, for a given $\calB_{n} \in \calG_{n}(Q_{UU'V})$, we define the set
\begin{align}
\calK(\calB_{n},Q_{UU'V}) = \{(\bu,\bu',\bv):~ Z_{\bu\bu'}(\bv) \leq \exp\{n[\alpha(R-2\epsilon,Q_{U},Q_{V}) + \epsilon]\}\},
\end{align}
as well as 
\begin{align} \label{DEF_K}
\calK(\calB_{n},Q_{UU'V}|\bu,\bu') = \{\bv:~ (\bu,\bu',\bv) \in \calK(\calB_{n},Q_{UU'V})\}.
\end{align}
Then, by definition, for any $\calB_{n} \in \calG_{n}(Q_{UU'V})$,
\begin{align} \label{ToHelp}
&\sum_{(\bu,\bu') \in \calT(Q_{UU'})} \IND \{\calB(\bu') = \calB(\bu)\} \cdot
\frac{|\calT(Q_{V|UU'}|\bu,\bu') \cap \calK(\calB_{n},Q_{UU'V}|\bu,\bu')|}{|\calT(Q_{V|UU'}|\bu,\bu')|} \nonumber \\
&\geq \exp\{n[H_{Q}(U,U') - R - 3\epsilon/2]\},
\end{align}
where we have used the fact that $\calT(Q_{V|UU'}|\bu,\bu')$ has exponentially the same cardinality for all $(\bu,\bu') \in \calT(Q_{UU'})$. 
Wrapping all up, we get that for any $\calB_{n} \in \calG_{n}$, 
\begin{align}
&P_{\mbox{\tiny e}} (\calB_{n})  \nn \\
&= \sum_{\bu,\bv} P(\bu,\bv) \IND \left\{ \hat{H}_{\bu}(U) \geq \RF(\hat{P}_{\bu}) \right\} 
\frac{\sum_{\bu' \in \calB(\bu) \cap \calT(\bu), \bu' \neq \bu} \exp\{nf(\hat{P}_{\bu'\bv})\}}{\exp\{nf(\hat{P}_{\bu\bv})\} + \exp\{nf(\hat{P}_{\bu'\bv})\} + Z_{\bu\bu'}(\bv)} \\
&\geq \sum_{\{Q_{UU'}:~ H_{Q}(U,U') \geq R + \epsilon,H_{Q}(U) \geq R(Q_{U})\}} \sum_{(\bu,\bu') \in \calT(Q_{UU'})} \IND \{\calB(\bu') = \calB(\bu)\} \cdot  \exp\{n \mathbb{E}_{Q} \log P(U)\}  \nonumber \\
& ~~~~ \times \sum_{Q_{V|UU'}} \sum_{\bv \in \calT(Q_{V|UU'}|\bu,\bu') \cap \calK(\calB_{n},Q_{UU'V}|\bu,\bu')} \exp\{n \mathbb{E}_{Q} \log P(V|U)\}  \nn \\
& ~~~~~~~~~~~~ \times \frac{\exp\{nf(Q_{U'V})\}}{\exp\{nf(Q_{UV})\} + \exp\{nf(Q_{U'V})\} + Z_{\bu\bu'}(\bv)}   \\
\label{E14}
&\geq \sum_{\{Q_{UU'}:~ H_{Q}(U,U') \geq R + \epsilon,H_{Q}(U) \geq R(Q_{U})\}} \sum_{(\bu,\bu') \in \calT(Q_{UU'})} \IND \{\calB(\bu') = \calB(\bu)\} \cdot  \exp\{n \mathbb{E}_{Q} \log P(U)\}  \nn \\
& ~~~~ \times \sum_{Q_{V|UU'}} \sum_{\bv \in \calT(Q_{V|UU'}|\bu,\bu') \cap \calK(\calB_{n},Q_{UU'V}|\bu,\bu')} \exp\{n \mathbb{E}_{Q} \log P(V|U)\}  \nonumber \\
& ~~~~~~~~~~~~ \times \frac{\exp\{nf(Q_{U'V})\}}{\exp\{nf(Q_{UV})\} + \exp\{nf(Q_{U'V})\} + \exp\{n[\alpha(R-2\epsilon,Q_{U},Q_{V}) + \epsilon]\}}   \\
&\doteq \sum_{\{Q_{UU'}:~ H_{Q}(U,U') \geq R + \epsilon,H_{Q}(U) \geq R(Q_{U})\}} \sum_{(\bu,\bu') \in \calT(Q_{UU'})} \IND \{\calB(\bu') = \calB(\bu)\}  \nonumber \\
& ~~~~ \times \sum_{Q_{V|UU'}} \frac{|\calT(Q_{V|UU'}|\bu,\bu') \cap \calK(\calB_{n},Q_{UU'V}|\bu,\bu')|}{|\calT(Q_{V|UU'}|\bu,\bu')|} \cdot |\calT(Q_{V|UU'}|\bu,\bu')| \cdot e^{n \mathbb{E}_{Q} \log P(U,V)}  \nonumber \\
& ~~~~ \times  \exp\{-n \cdot [\max\{f(Q_{UV}),\alpha(R-2\epsilon,Q_{U},Q_{V}) + \epsilon\} - f(Q_{U'V})]_{+}\}  \\
&\doteq \sum_{\{Q_{UU'V}:~ H_{Q}(U,U') \geq R + \epsilon,H_{Q}(U) \geq R(Q_{U})\}} \sum_{(\bu,\bu') \in \calT(Q_{UU'})} \IND \{\calB(\bu') = \calB(\bu)\}  \nonumber \\
& ~~~~ \times \frac{|\calT(Q_{V|UU'}|\bu,\bu') \cap \calK(\calB_{n},Q_{UU'V}|\bu,\bu')|}{|\calT(Q_{V|UU'}|\bu,\bu')|} \cdot e^{n H_{Q}(V|U,U')} \cdot e^{n \mathbb{E}_{Q} \log P(U,V)}  \nonumber \\
& ~~~~ \times  \exp\{-n \cdot [\max\{f(Q_{UV}),\alpha(R-2\epsilon,Q_{U},Q_{V}) + \epsilon\} - f(Q_{U'V})]_{+}\}  \\
\label{E15}
&\geq \sum_{\{Q_{UU'V}:~ H_{Q}(U,U') \geq R + \epsilon,H_{Q}(U) \geq R(Q_{U})\}} \exp\{n[H_{Q}(U,U') - R - 3\epsilon/2]\} \cdot e^{n H_{Q}(V|U,U')} \nn \\
& ~~ \times  e^{n \mathbb{E}_{Q} \log P(U,V)} \cdot \exp\{-n \cdot [\max\{f(Q_{UV}),\alpha(R-2\epsilon,Q_{U},Q_{V}) + \epsilon\} - f(Q_{U'V})]_{+}\}  \\
&\doteq \exp\left\{-n \cdot \min_{\{Q_{UU'V}:~ H_{Q}(U,U') \geq R + \epsilon,H_{Q}(U) \geq R(Q_{U})\}}  \{-H_{Q}(U,U') + R + 3\epsilon/2 - H_{Q}(V|U,U') \right. \nn \\
&\left. ~~~~~~~~~~- \mathbb{E}_{Q} [\log P(U,V)] + [\max\{f(Q_{UV}),\alpha(R-2\epsilon,Q_{U},Q_{V}) + \epsilon\} - f(Q_{U'V})]_{+}\}\right\} \\
&\dfn \exp\{-n E_{\mbox{\tiny trc}}(R,\epsilon)\}, 
\end{align}
where \eqref{E14} follows from the definition of the set $\calK(\calB_{n},Q_{UU'V}|\bu,\bu')$ in \eqref{DEF_K} and \eqref{E15} is due to \eqref{ToHelp}.
Consider the following:
\begin{align}
&\mathbb{E} \left[-\frac{1}{n} \log P_{\mbox{\tiny e}} (\calB_{n}) \right] \nn \\
&= \sum_{\calB_{n}} \prob\{\calB_{n}\} \left(-\frac{1}{n} \log P_{\mbox{\tiny e}} (\calB_{n}) \right) \\
&= \sum_{\calB_{n} \in \calG_{n}} \prob\{\calB_{n}\} \left(-\frac{1}{n} \log P_{\mbox{\tiny e}} (\calB_{n}) \right) + \sum_{\calB_{n} \in \calG_{n}^{\mbox{\tiny c}}} \prob\{\calB_{n}\} \left(-\frac{1}{n} \log P_{\mbox{\tiny e}} (\calB_{n}) \right) \\
&\leq \sum_{\calB_{n} \in \calG_{n}} \prob\{\calB_{n}\} \left(-\frac{1}{n} \log e^{-n E_{\mbox{\tiny trc}}(R,\epsilon)} \right) + \sum_{\calB_{n} \in \calG_{n}^{\mbox{\tiny c}}} \prob\{\calB_{n}\} \left(-\frac{1}{n} \log e^{-n E_{\mbox{\tiny sp}}(R)} \right) \\
&= \prob\{\calG_{n}\} E_{\mbox{\tiny trc}}(R,\epsilon) + \prob\{\calG_{n}^{\mbox{\tiny c}}\} E_{\mbox{\tiny sp}}(R),
\end{align}
which implies that 
\begin{align}
\limsup_{n \to \infty} \mathbb{E} \left[-\frac{1}{n} \log P_{\mbox{\tiny e}} (\calB_{n}) \right] 
\leq E_{\mbox{\tiny trc}}(R,\epsilon).
\end{align}
It follows from the arbitrariness of $\epsilon$ that
\begin{align}
&\limsup_{n \to \infty} \mathbb{E} \left\{-\frac{1}{n}\log \left[P_{\mbox{\tiny e}} (\calB_{n}) \right] \right\}  \nonumber \\
&\leq \min_{\{Q_{UU'V}:~ H_{Q}(U,U') \geq R, H_{Q}(U) \geq R(Q_{U})\}}  \{-H_{Q}(U,U') + R - H_{Q}(V|U,U')  \nn \\
& ~~~~~~~~~~- \mathbb{E}_{Q} [\log P(U,V)] + [\max\{f(Q_{UV}),\alpha(R,Q_{U},Q_{V})\} - f(Q_{U'V})]_{+}\} \\
&= \min_{\{Q_{UU'V}:~ H_{Q}(U) \geq R(Q_{U})\}}  \{-H_{Q}(U,U') + R - H_{Q}(V|U,U')  \nn \\
& ~~~~~~~~~~- \mathbb{E}_{Q} [\log P(U,V)] + [\max\{f(Q_{UV}),\alpha(R,Q_{U},Q_{V})\} - f(Q_{U'V})]_{+}\} \\
&= \min_{\{Q_{UU'}\in \calQ:~ H_{Q}(U) \geq R(Q_{U})\}} \{\Lambda(Q_{UU'},R) - H_{Q}(U,U') + R(Q_{U}) - \mathbb{E}_{Q} [\log P(U)]\},
\end{align} 
which completes the proof of Theorem \ref{SW_THM}.

\section*{Appendix C}
\renewcommand{\theequation}{C.\arabic{equation}}
\setcounter{equation}{0}  
\subsection*{Proof of Lemma \ref{Large_Deviations_Z}}

Let $N(\calT(Q_{U|V}|\bv),\calB(\bu))$ be defined as
\begin{align}
N(\calT(Q_{U|V}|\bv),\calB(\bu)) = \sum_{\bu' \in \calT(Q_{U|V}|\bv)}  
\IND \left\{ \calB(\bu') = \calB(\bu) \right\}.
\end{align}
First, note that
\begin{align}
Z_{\bu}(\bv) = \sum_{\tilde{\bu} \in \calB(\bu)\cap\calT(\bu), \tilde{\bu} \neq \bu} \exp\{nf(\hat{P}_{\tilde{\bu}\bv})\} 
= \sum_{Q_{U|V} \in \calS(\hat{P}_{\bu},\hat{P}_{\bv})} N(\calT(Q_{U|V}|\bv),\calB(\bu)) e^{nf(Q_{UV})} ,
\end{align}
where $\calS(\hat{P}_{\bu},\hat{P}_{\bv}) = \{Q_{U|V}:~ (\hat{P}_{\bv} \times Q_{U|V})_{U}=\hat{P}_{\bu}\}$.
Thus, taking the randomness of $\{\calB(\bu)\}_{\bu \in \calU^{n}}$ into account,
\begin{align}
&\prob \left\{ Z_{\bv}(\bu) \leq \exp \{n \alpha(R+\epsilon,\hat{P}_{\bu},\hat{P}_{\bv}) \} \right\} \nonumber \\
&=  \prob \left\{ \sum_{Q_{U|V} \in \calS(\hat{P}_{\bu},\hat{P}_{\bv})} N(\calT(Q_{U|V}|\bv),\calB(\bu)) e^{nf(Q_{UV})} \leq \exp \{n \alpha(R+\epsilon,\hat{P}_{\bu},\hat{P}_{\bv}) \}   \right\} \\
&\leq  \prob \left\{ \max_{Q_{U|V} \in \calS(\hat{P}_{\bu},\hat{P}_{\bv})} N(\calT(Q_{U|V}|\bv),\calB(\bu)) e^{nf(Q_{UV})} \leq \exp \{n \alpha(R+\epsilon,\hat{P}_{\bu},\hat{P}_{\bv}) \}   \right\} \\
&=  \prob \bigcap_{Q_{U|V} \in \calS(\hat{P}_{\bu},\hat{P}_{\bv})} \left\{ N(\calT(Q_{U|V}|\bv),\calB(\bu)) e^{nf(Q_{UV})} \leq \exp \{n \alpha(R+\epsilon,\hat{P}_{\bu},\hat{P}_{\bv}) \}   \right\} \\
&=  \prob \bigcap_{Q_{U|V} \in \calS(\hat{P}_{\bu},\hat{P}_{\bv})} \left\{ N(\calT(Q_{U|V}|\bv),\calB(\bu))  \leq \exp \{n [ \alpha(R+\epsilon,\hat{P}_{\bu},\hat{P}_{\bv}) - f(Q_{UV})] \}   \right\} .
\end{align}
Now, $N(\calT(Q_{U|V}|\bv),\calB(\bu))$ is a binomial random variable with $|\calT(Q_{U|V}|\bv)| \doteq e^{n H_{Q}(U|V)}$ trials and success rate which is of the exponential order of $e^{-nR}$. We prove that by the very definition of the function $\alpha(R+\epsilon,\hat{P}_{\bu},\hat{P}_{\bv})$, there must exist some conditional distribution $Q_{U|V}^{*} \in \calS(\hat{P}_{\bu},\hat{P}_{\bv})$ such that for $Q_{UV}^{*} = \hat{P}_{\bv} \times Q_{U|V}^{*}$, the two inequalities $H_{Q^{*}}(U|V) \geq R + \epsilon$ and $H_{Q^{*}}(U|V) - R - \epsilon  \geq \alpha(R+\epsilon,\hat{P}_{\bu},\hat{P}_{\bv}) - f(Q_{UV}^{*})$ hold. To show that, we assume conversely, i.e., that for every conditional distribution $Q_{U|V} \in \calS(\hat{P}_{\bu},\hat{P}_{\bv})$, which defines $Q_{UV} = \hat{P}_{\bv} \times Q_{U|V}$, either $H_{Q}(U|V) < R + \epsilon$ or $H_{Q}(U|V) - R - \epsilon  < \alpha(R+\epsilon,\hat{P}_{\bu},\hat{P}_{\bv}) - f(Q_{UV})$, which means that for every distribution $Q_{U|V} \in \calS(\hat{P}_{\bu},\hat{P}_{\bv})$
\begin{align}
H_{Q}(U|V) - \epsilon &< \max \{R, R + \alpha(R+\epsilon,\hat{P}_{\bu},\hat{P}_{\bv}) - f(Q_{UV}) \}  \\
&= R + [\alpha(R+\epsilon,\hat{P}_{\bu},\hat{P}_{\bv}) - f(Q_{UV})]_{+}.
\end{align}
Writing it slightly differently, for every $Q_{U|V} \in \calS(\hat{P}_{\bu},\hat{P}_{\bv})$ there exists some real number $t \in [0,1]$ such that
\begin{align}
H_{Q}(U|V) - \epsilon &< R + t[\alpha(R+\epsilon,\hat{P}_{\bu},\hat{P}_{\bv}) - f(Q_{UV})],
\end{align}
or equivalently, 
\begin{align}
\alpha(R+\epsilon,\hat{P}_{\bu},\hat{P}_{\bv}) &>
\max_{Q_{U|V} \in \calS(\hat{P}_{\bu},\hat{P}_{\bv})} \min_{t \in [0,1]}  f(Q_{UV}) + \frac{H_{Q}(U|V) - R - \epsilon}{t} \\
&= \max_{Q_{U|V} \in \calS(\hat{P}_{\bu},\hat{P}_{\bv})} \left\{ 
\begin{array}{l l}
f(Q_{UV}) + H_{Q}(U|V) - R - \epsilon   &  \text{  $H_{Q}(U|V) \geq R + \epsilon$  }\\
-\infty  &  \text{ $H_{Q}(U|V) < R + \epsilon$  } 
\end{array} \right. \\
&= \max_{\{Q_{U|V} \in \calS(\hat{P}_{\bu},\hat{P}_{\bv}):~ H_{Q}(U|V) \geq R + \epsilon\}} [f(Q_{UV}) + H_{Q}(U|V)] - R - \epsilon \\
&\equiv \alpha(R+\epsilon,\hat{P}_{\bu},\hat{P}_{\bv}),
\end{align}
which is a contradiction. Let the conditional distribution $Q_{U|V}^{*}$ be as defined above. Then,
\begin{align}
&\prob \bigcap_{Q_{U|V} \in \calS(\hat{P}_{\bu},\hat{P}_{\bv})} \Big\{ N(\calT(Q_{U|V}|\bv),\calB(\bu))  \leq \exp \{n [ \alpha(R+\epsilon,\hat{P}_{\bu},\hat{P}_{\bv}) - f(Q_{UV})] \}   \Big\} \\
\label{up7}
&\leq \prob \Big\{ N(\calT(Q_{U|V}^{*}|\bv),\calB(\bu))  \leq \exp \{n [ \alpha(R+\epsilon,\hat{P}_{\bu},\hat{P}_{\bv}) - f(Q_{UV}^{*})] \}  \Big\}.
\end{align}
Now, we know that both of the inequalities $H_{Q^{*}}(U|V) \geq R + \epsilon$ and $H_{Q^{*}}(U|V) - R - \epsilon  \geq \alpha(R+\epsilon,\hat{P}_{\bu},\hat{P}_{\bv}) - f(Q_{UV}^{*})$ hold. By the Chernoff bound, the probability of (\ref{up7}) is upper bounded by
\begin{align}
\exp \Big\{ -e^{nH_{Q^{*}}(U|V)} D( e^{-an} \| e^{-bn} ) \Big\},
\end{align}
where $a = H_{Q^{*}}(U|V) + f(Q_{UV}^{*}) - \alpha(R+\epsilon,\hat{P}_{\bu},\hat{P}_{\bv})$ and $b = R$, and where $D( \alpha \| \beta )$, for $\alpha,\beta \in [0,1]$, is the binary divergence function, that is 
\begin{align}
D( \alpha \| \beta ) = \alpha \log \frac{\alpha}{\beta} + (1-\alpha) \log \frac{1-\alpha}{1-\beta}.  
\end{align}
Since $a - b \geq \epsilon$, the binary divergence is lower bounded as follows (\cite[Sec.\ 6.3]{MERHAV09}):
\begin{align}
D( e^{-an} \| e^{-bn} ) &\geq e^{-bn} \left\{ 1- e^{-(a-b)n} [1 + n(a-b)]  \right\} \\
&\geq e^{-nR} [ 1- e^{-n \epsilon} (1 + n \epsilon)  ] ,
\end{align}
where in the second inequality, we invoked the decreasing monotonicity of the function $f(t) = (1+t) e^{-t}$ for $t \geq 0$. Finally, we get that
\begin{align}
& \prob \Big\{ N(\calT(Q_{U|V}^{*}|\bv),\calB(\bu))  \leq \exp \{n [ \alpha(R+\epsilon,\hat{P}_{\bu},\hat{P}_{\bv}) - f(Q_{UV}^{*})] \}  \Big\} \\
&\leq \exp \Big\{ -e^{nH_{Q^{*}}(U|V)} \cdot  e^{-nR} [ 1- e^{-n \epsilon} (1 + n \epsilon)  ]  \Big\}  \\
&\leq \exp \big\{ -e^{n \epsilon}  [ 1- e^{-n \epsilon} (1 + n \epsilon)  ]  \big\} \\
&= \exp \big\{ -e^{n \epsilon}  + n \epsilon + 1  \big\} .
\end{align}
This completes the proof of Lemma \ref{Large_Deviations_Z}.

\section*{Appendix D}
\renewcommand{\theequation}{D.\arabic{equation}}
\setcounter{equation}{0}  
\subsection*{Proof of Theorem \ref{Thm_comparison}}

By definition of the error exponents, it follows that $E_{\mbox{\tiny trc,GLD}}(\RF(\cdot)) \geq E_{\mbox{\tiny r,GLD}}(\RF(\cdot))$. We now prove the other direction. 
The expression in \eqref{VR_SD_SW_EXPRESSION} can also be written as
\begin{align} 
&E_{\mbox{\tiny trc,GLD}}(\RF(\cdot)) \nn \\
&= \min_{\left\{\substack{Q_{UU'}:~Q_{U'}=Q_{U}, \\ H_{Q}(U) \geq \RF(Q_{U}) } \right\}} \left\{\Lambda(Q_{UU'},\RF(Q_{U})) - \mathbb{E}_{Q} [\log P(U)] - H_{Q}(U,U') + \RF(Q_{U}) \right\} \\
&= \min_{\left\{\substack{Q_{UU'}:~Q_{U'}=Q_{U}, \\ H_{Q}(U) \geq \RF(Q_{U}) } \right\}} \left\{\min_{Q_{V|UU'}} \left\{ \Psi(\RF(Q_{U}),Q_{UU'V}) - H_{Q}(V|U,U') - \mathbb{E}_{Q} [\log P(V|U)] \right\} \right. \nn \\
& \left.~~~~~~~~~~~~~~~~~~~~~~~~~~~~~~~~~~~~ - \mathbb{E}_{Q} [\log P(U)] - H_{Q}(U,U') + \RF(Q_{U}) \right\} \\
&= \min_{\left\{\substack{Q_{UU'V}:~Q_{U'}=Q_{U}, \\ H_{Q}(U) \geq \RF(Q_{U}) } \right\}} \left\{  \Psi(\RF(Q_{U}),Q_{UU'V}) - H_{Q}(U,U',V) - \mathbb{E}_{Q} [\log P(U,V)] + \RF(Q_{U}) \right\} \\
&= \min_{\left\{\substack{Q_{UU'V}:~Q_{U'}=Q_{U}, \\ H_{Q}(U) \geq \RF(Q_{U}) } \right\}} \left\{  \Psi(\RF(Q_{U}),Q_{UU'V}) + D(Q_{UV} \| P_{UV})  - H_{Q}(U'|U,V) + \RF(Q_{U}) \right\} \\
\label{MAP_exponent_SD}
&=\min_{\calQ} \left\{ D(Q_{UV} \| P_{UV}) + \RF(Q_{U}) - H_{Q}(U'|U,V) \right. \nn \\
&~~~~~~~~~~~ \left. + \left[\max\{f(Q_{UV}), \gamma(\RF(Q_{U}),Q_{U},Q_{V})\} - f(Q_{U'V}) \right]_{+} \right\},
\end{align} 	    
with the set $\calQ$ given by $\calQ = \{ Q_{UU'V}:~Q_{U'}=Q_{U},~H_{Q}(U) \geq \RF(Q_{U}) \}$,
and where,
\begin{align} \label{gamma_def} 
\gamma(\RF(\cdot),Q_{U},Q_{V}) 
&= \max_{\left\{\substack{Q_{\tilde{U}|V}:~Q_{\tilde{U}}=Q_{U}, \\ H_{Q}(\tilde{U}|V) \geq \RF(Q_{\tilde{U}}) }\right\}} \{f(Q_{\tilde{U}V}) + H_{Q}(\tilde{U}|V)\} - \RF(Q_{U}). 
\end{align}  
We upper--bound the minimum in \eqref{MAP_exponent_SD} by decreasing the feasible set; we add to $\calQ$ the constraint that $U \leftrightarrow V \leftrightarrow U'$ form a Markov chain in that order and denote the new feasible set by $\tilde{\calQ}$. We get that 
\begin{align} 
E_{\mbox{\tiny trc,GLD}}(\RF(\cdot)) 
&\leq \min_{\tilde{\calQ}} \left\{ D(Q_{UV} \| P_{UV}) + \RF(Q_{U}) - H_{Q}(U'|U,V) \right. \nn \\
&~~~~~~~~~~~ \left. + \left[\max\{f(Q_{UV}), \gamma(\RF(Q_{U}),Q_{U},Q_{V})\} - f(Q_{U'V}) \right]_{+} \right\} \\
&= \min_{\tilde{\calQ}} \left\{ D(Q_{UV} \| P_{UV}) + \RF(Q_{U}) - H_{Q}(U'|V) \right. \nn \\
&~~~~~~~~~~~ \left. + \left[\max\{f(Q_{UV}), \gamma(\RF(Q_{U}),Q_{U},Q_{V})\} - f(Q_{U'V}) \right]_{+} \right\} \\
\label{Inner_minimum}
&= \min_{\{Q_{UV}:~ H_{Q}(U) \geq \RF(Q_{U}) \}} \left\{ D(Q_{UV} \| P_{UV}) + \min_{Q_{U'|V} \in \hat{\calQ}} \{\RF(Q_{U}) - H_{Q}(U'|V) \right. \nn \\
&\left. ~~~~~~~~~~~+ \left[\max\{f(Q_{UV}), \gamma(\RF(Q_{U}),Q_{U},Q_{V})\} - f(Q_{U'V}) \right]_{+} \} \right\},
\end{align}
where $\hat{\calQ} = \{ Q_{U'|V}:~Q_{U'}=Q_{U} \}.$
In order to upper--bound the inner minimum in \eqref{Inner_minimum}, we split into two cases, according to the maximum between $f(Q_{UV})$ and $ \gamma(\RF(Q_{U}),Q_{U},Q_{V})$. This is legitimate when the inner minimum and this maximum can be interchanged, which is possible at least in the special cases of the matched/mismatched decoding metrics $f(Q) = \beta \mathbb{E}_{Q} [\log \tilde{P}(U,V)]$ for some $\beta > 0$, since if $f(Q)$ is linear, then the entire expression inside the inner minimum in \eqref{Inner_minimum} is convex in $Q_{U'|V}$. 
On the one hand, if the maximum is given by $f(Q_{UV})$, then the inner minimum in \eqref{Inner_minimum} is just
\begin{align} \label{ToUse0}
\min_{Q_{U'|V} \in \hat{\calQ}} \left\{ \RF(Q_{U}) - H_{Q}(U'|V) + \left[ f(Q_{UV}) - f(Q_{U'V}) \right]_{+} \right\}. 
\end{align}
On the other hand, if the maximum is given by $\gamma(\RF(Q_{U}),Q_{U},Q_{V})$, let $Q^{*}=Q_{\tilde{U}|V}^{*}$ be the maximizer in \eqref{gamma_def}, and then 
\begin{align}
&\min_{Q_{U'|V} \in \hat{\calQ}} \left\{ \RF(Q_{U}) - H_{Q}(U'|V) + \left[ \gamma(\RF(Q_{U}),Q_{U},Q_{V}) - f(Q_{U'V}) \right]_{+} \right\} \nn \\ 
&=\min_{Q_{U'|V} \in \hat{\calQ}} \left\{ \RF(Q_{U}) - H_{Q}(U'|V) + \left[ f(Q_{\tilde{U}V}^{*}) + H_{Q^{*}}(\tilde{U}|V) - \RF(Q_{U}) - f(Q_{U'V}) \right]_{+} \right\} \\ 
\label{ToExp0}
&\leq \RF(Q_{U}) - H_{Q^{*}}(U'|V) + \left[ f(Q_{\tilde{U}V}^{*}) + H_{Q^{*}}(\tilde{U}|V) - \RF(Q_{U}) - f(Q_{U'V}^{*}) \right]_{+}  \\
&= \RF(Q_{U}) - H_{Q^{*}}(U'|V) + \left[ H_{Q^{*}}(\tilde{U}|V) - \RF(Q_{U}) \right]_{+}  \\
\label{ToExp1}
&= \RF(Q_{U}) - H_{Q^{*}}(U'|V) +  H_{Q^{*}}(\tilde{U}|V) - \RF(Q_{U})   \\
\label{ToUse1}
&= 0,  
\end{align}
where \eqref{ToExp0} is because we choose $Q_{U'|V}^{*}=Q_{\tilde{U}|V}^{*}$ instead of minimizing over all $Q_{U'|V} \in \hat{\calQ}$ and \eqref{ToExp1} is true since $H_{Q^{*}}(\tilde{U}|V) \geq \RF(Q_{U})$ by the definition of $\gamma(\RF(Q_{U}), Q_{U},Q_{V})$. Combining \eqref{ToUse0} and \eqref{ToUse1}, we find that \eqref{Inner_minimum} is upper--bounded by
\begin{align} 
E_{\mbox{\tiny trc,GLD}}(\RF(\cdot))   
&\leq \min_{\left\{Q_{UV}:~ H_{Q}(U) \geq \RF(Q_{U}) \right\}} \left\{ D(Q_{UV} \| P_{UV}) \right. \nn \\
&~~~ \left. +  \max\left\{\min_{Q_{U'|V} \in \hat{\calQ}} \left\{ \RF(Q_{U}) - H_{Q}(U'|V) + \left[ f(Q_{UV}) - f(Q_{U'V}) \right]_{+} \right\} ,0 \right\} \right\}\\
&= \min_{\left\{Q_{UV}:~ H_{Q}(U) \geq \RF(Q_{U}) \right\}} \left\{ D(Q_{UV} \| P_{UV}) \right. \nn \\
&~~~ \left. +  \left[ \min_{Q_{U'|V} \in \hat{\calQ}} \left\{ \RF(Q_{U}) - H_{Q}(U'|V) + \left[ f(Q_{UV}) - f(Q_{U'V}) \right]_{+} \right\} \right]_{+} \right\} \\
&= \min_{\left\{Q_{UV}:~ H_{Q}(U) \geq \RF(Q_{U}) \right\}} \left\{ D(Q_{UV} \| P_{UV}) \right. \nn \\
&~~~ \left. + \min_{Q_{U'|V} \in \hat{\calQ}} \left\{ \left[ \RF(Q_{U}) - H_{Q}(U'|V) + \left[ f(Q_{UV}) - f(Q_{U'V}) \right]_{+} \right]_{+} \right\} \right\} \\	 
&= E_{\mbox{\tiny r,GLD}}(\RF(\cdot)),
\end{align}
which proves the first point of the theorem.
Moving forward, consider the following:
\begin{align} \label{Beautiful}
E_{\mbox{\tiny trc,MAP}}(\RF(\cdot)) 
\overset{\mbox{\small (a)}}{=} E_{\mbox{\tiny r,MAP}}(\RF(\cdot))
\overset{\mbox{\small (b)}}{=} E_{\mbox{\tiny r,MCE}}(\RF(\cdot))
\overset{\mbox{\small (c)}}{\leq} E_{\mbox{\tiny trc,MCE}}(\RF(\cdot))
\overset{\mbox{\small (d)}}{\leq} E_{\mbox{\tiny trc,MAP}}(\RF(\cdot)),
\end{align}
where $\mbox{(a)}$ follows from the first point in this theorem by using the matched decoding metric $f(Q) = \beta \mathbb{E}_{Q} [\log P(U,V)]$ and letting $\beta \to \infty$. Equality $\mbox{(b)}$ is due to the second point of Theorem \ref{Thm_random_binning_SD_VR}, which ensures that the random binning error exponents of the MAP and the MCE decoders are equal.
Passage $\mbox{(c)}$ is thanks to the fact that for any decoder, the error exponent of the typical random code is always at least as high as the random coding error exponent and $\mbox{(d)}$ is due to the fact that the MAP decoder is optimal. Finally, the leftmost and the rightmost sides of \eqref{Beautiful} are the same, which implies that passages $\mbox{(c)}$ and $\mbox{(d)}$ must hold with equalities. 
The equality in passage $\mbox{(c)}$ concludes the second point of the theorem.

\section*{Appendix E}
\renewcommand{\theequation}{E.\arabic{equation}}
\setcounter{equation}{0}  
\subsection*{Proof of Theorem \ref{Thm_Universal}}

The left equality in \eqref{3KINGS} is implied by the proved equality in passage $\mbox{(d)}$ in \eqref{Beautiful}.  
In order to prove the right equality in \eqref{3KINGS}, first note that $E_{\mbox{\tiny trc,SCE}}(\RF(\cdot)) \leq E_{\mbox{\tiny trc,MAP}}(\RF(\cdot))$ by the optimality of the MAP decoder. For the other direction,   
consider the universal decoding metric of $f(Q_{UV}) = -H_{Q}(U|V)$. Then, trivially,
\begin{align}
\gamma(\RF(\cdot),Q_{U},Q_{V}) 
&= \max_{\left\{\substack{Q_{\tilde{U}|V}:~Q_{\tilde{U}}=Q_{U}, \\ H_{Q}(\tilde{U}|V) \geq \RF(Q_{\tilde{U}}) }\right\}} \{f(Q_{\tilde{U}V}) + H_{Q}(\tilde{U}|V)\} - \RF(Q_{U}) = - \RF(Q_{U}), 
\end{align} 
as well as
\begin{align}
\Psi(\RF(\cdot),Q_{UU'V}) 
&= \left[\max\{f(Q_{UV}), \gamma(\RF(\cdot),Q_{U},Q_{V})\} - f(Q_{U'V}) \right]_{+} \\ 
&= \left[\max\{-H_{Q}(U|V), -\RF(Q_{U})\} + H_{Q}(U'|V) \right]_{+}  \\
&= \left[H_{Q}(U'|V) - \min\{H_{Q}(U|V), \RF(Q_{U})\} \right]_{+}  \\
&\geq \left[H_{Q}(U'|U,V) - \min\{H_{Q}(U|V), \RF(Q_{U})\} \right]_{+}.
\end{align}
We have the following
\begin{align} 
&E_{\mbox{\tiny trc,SCE}}(\RF(\cdot)) \nn \\
&=\min_{\calQ} \left\{ D(Q_{UV} \| P_{UV}) + \RF(Q_{U}) - H_{Q}(U'|U,V) \right. \nn \\
&~~~~~~~~~~~ \left. + \left[\max\{f(Q_{UV}), \gamma(\RF(Q_{U}),Q_{U},Q_{V})\} - f(Q_{U'V}) \right]_{+} \right\} \\
&\geq \min_{\calQ} \left\{ D(Q_{UV} \| P_{UV}) + \RF(Q_{U}) - H_{Q}(U'|U,V) \right. \nn \\
&~~~~~~~~~~~ \left. + \left[H_{Q}(U'|U,V) - \min\{H_{Q}(U|V), \RF(Q_{U})\} \right]_{+} \right\} \\
&= \min_{\calQ} \left\{ D(Q_{UV} \| P_{UV})  - \min\{H_{Q}(U|V),H_{Q}(U'|U,V), \RF(Q_{U})\} + \RF(Q_{U}) \right\} \\
&\geq \min_{\calQ} \left\{ D(Q_{UV} \| P_{UV})  - \min\{H_{Q}(U|V),H_{Q}(U'), \RF(Q_{U})\} + \RF(Q_{U}) \right\} \\
&= \min_{\{Q_{UV}:~ H_{Q}(U) \geq \RF(Q_{U})\}} \left\{ D(Q_{UV} \| P_{UV})  - \min\{H_{Q}(U|V), H_{Q}(U), \RF(Q_{U}) \} + \RF(Q_{U}) \right\} \\
&= \min_{\{Q_{UV}:~ H_{Q}(U) \geq \RF(Q_{U})\}} \left\{ D(Q_{UV} \| P_{UV})  - \min\{H_{Q}(U|V), \RF(Q_{U}) \} + \RF(Q_{U}) \right\} \\
&= \min_{\{Q_{UV}:~ H_{Q}(U) \geq \RF(Q_{U})\}} \left\{ D(Q_{UV} \| P_{UV})  + \max\{\RF(Q_{U}) - H_{Q}(U|V),0\} \right\} \\
&= \min_{\{Q_{UV}:~ H_{Q}(U) \geq \RF(Q_{U})\}} \left\{D(Q_{UV} \| P_{UV})  + [\RF(Q_{U}) - H_{Q}(U|V)]_{+} \right\} \\
&= E_{\mbox{\tiny trc,MAP}}(\RF(\cdot)),
\end{align}
which completes the proof of the theorem.

\section*{Appendix F}
\renewcommand{\theequation}{F.\arabic{equation}}
\setcounter{equation}{0}  
\subsection*{Proof of Theorem \ref{RF_UB}}

We start by writing the expression in \eqref{Excess_exponent} in a slightly different way using $\min_{\{Q:~ g(Q) \leq 0\}} f(Q) = \min_{Q} \sup_{s \geq 0} \{f(Q) + s \cdot g(Q)\}$:
\begin{align} 
E_{\mbox{\tiny er}}(\RF(\cdot),\Delta) 
&= \min_{\{Q_{UV}:~ \RF(Q_{U}) \geq H_{Q}(U|V) + \Delta\}} D(Q_{UV} \| P_{UV}) \\ 
&= \min_{Q_{UV}} \sup_{\sigma \geq 0} \{D(Q_{UV} \| P_{UV}) + \sigma \cdot (H_{Q}(U|V) + \Delta - \RF(Q_{U}))\}.
\end{align}	
Now, the requirement $E_{\mbox{\tiny er}}(\RF(\cdot),\Delta) \geq \er$ is equivalent to 
\begin{align}
&\min_{Q_{UV}} \sup_{\sigma \geq 0} \{D(Q_{UV} \| P_{UV}) + \sigma \cdot (H_{Q}(U|V) + \Delta - \RF(Q_{U}))\} \geq \er 
\end{align}
or,
\begin{align}
&\forall Q_{UV},~ \exists \sigma \geq 0,~ D(Q_{UV} \| P_{UV}) + \sigma \cdot (H_{Q}(U|V) + \Delta - \RF(Q_{U})) \geq \er 
\end{align}
or,
\begin{align}
&\forall Q_{U},~\forall Q_{V|U},~ \exists \sigma \geq 0,~   \RF(Q_{U}) \leq H_{Q}(U|V) + \Delta + \frac{D(Q_{UV} \| P_{UV}) - \er}{\sigma} 
\end{align}
or that for any $Q_{U} \in \calP(\calU)$,
\begin{align}
\RF(Q_{U}) 
&\leq \min_{Q_{V|U}} \sup_{\sigma \geq 0} \left\{H_{Q}(U|V) + \Delta + \frac{D(Q_{UV} \| P_{UV}) - \er}{\sigma} \right\} \\
&= 
\min_{Q_{V|U}} \left\{   
\begin{array}{l l}
H_{Q}(U|V) + \Delta  & \quad \text{ $D(Q_{UV} \| P_{UV}) \leq \er$  }\\
\infty  & \quad \text{ $D(Q_{UV} \| P_{UV}) > \er$  } 
\end{array} \right. \\
&=  \min_{\{Q_{V|U}:~D(Q_{UV} \| P_{UV}) \leq \er\}} \left\{ H_{Q}(U|V) + \Delta \right\},
\end{align}
with the understanding that a minimum over an empty set equals infinity.

\section*{Appendix G}
\renewcommand{\theequation}{G.\arabic{equation}}
\setcounter{equation}{0}  
\subsection*{Proof of Theorem \ref{RF_LB_RC23}}

It follows by the identities $\min_{\{Q:~ g(Q) \leq 0\}} f(Q) = \min_{Q} \sup_{s \geq 0} \{f(Q) + s \cdot g(Q)\}$ and $\left[ A \right]_{+} = \max_{\mu \in [0,1]} \mu A$ that \eqref{SW_RC_VR_SD_MAP} can also be written as
\begin{align}
E_{\mbox{\tiny e}}(\RF(\cdot)) = \min_{Q_{U}} \min_{Q_{V|U}} \max_{\mu \in [0,1]} \sup_{\sigma \geq 0} \{D(Q_{UV} \| P_{UV}) &+\mu \cdot (\RF(Q_{U}) - H_{Q}(U|V)) \nn \\
&~~~~~~~~+ \sigma \cdot (\RF(Q_{U}) - H_{Q}(U)) \},
\end{align}
such that $E_{\mbox{\tiny e}}(\RF(\cdot)) \geq \ee$ is equivalent to
\begin{align}
&\forall Q_{U},~ \forall Q_{V|U},~ \exists \mu \in [0,1],~ \exists \sigma\geq 0: \nn \\
&D(Q_{UV} \| P_{UV}) +\mu \cdot (\RF(Q_{U}) - H_{Q}(U|V)) + \sigma \cdot (\RF(Q_{U}) - H_{Q}(U)) \geq \ee,
\end{align}
or,
\begin{align}
&\forall Q_{U},~ \forall Q_{V|U},~ \exists \mu \in [0,1],~ \exists \sigma\geq 0: \nn \\ 
&\RF(Q_{U}) \geq  \frac{\mu \cdot H_{Q}(U|V) + \sigma \cdot H_{Q}(U) + \ee - D(Q_{UV} \| P_{UV})}{\mu+\sigma},
\end{align}
or that for any $Q_{U} \in \calP(\calU)$,
\begin{align} 
\RF(Q_{U}) 
&\geq \max_{Q_{V|U}} \min_{\mu \in [0,1]} \inf_{\sigma \geq 0}
\left\{ \frac{\mu \cdot H_{Q}(U|V) + \sigma \cdot H_{Q}(U) + \ee - D(Q_{UV} \| P_{UV})}{\mu+\sigma} \right\} \\
&= \max_{Q_{V|U}} \min_{\mu \in [0,1]} \min
\left\{H_{Q}(U), H_{Q}(U|V) + \frac{\ee - D(Q_{UV} \| P_{UV})}{\mu} \right\} \\
&= \max_{Q_{V|U}} \min
\left\{H_{Q}(U), \min_{\mu \in [0,1]} \left\{
H_{Q}(U|V) + \frac{\ee - D(Q_{UV} \| P_{UV})}{\mu} \right\} \right\} \\
\label{To_exp_1}
&= \max_{Q_{V|U}} \left\{   
\begin{array}{l l}
\min \{H_{Q}(U), H_{Q}(U|V) + \ee - D(Q_{UV} \| P_{UV})\}  &  \text{ $\ee \geq D(Q_{UV} \| P_{UV})$  }\\
-\infty  &  \text{ $\ee < D(Q_{UV} \| P_{UV})$  } 
\end{array} \right. \\
&= \max_{\{Q_{V|U}:~ D(Q_{UV} \| P_{UV}) \leq \ee\}} \min \{H_{Q}(U), H_{Q}(U|V) + \ee - D(Q_{UV} \| P_{UV})\} \\
&= \min \left\{ H_{Q}(U), \max_{\{Q_{V|U}:~ D(Q_{UV} \| P_{UV}) \leq \ee\}} \{ H_{Q}(U|V) + \ee - D(Q_{UV} \| P_{UV})\} \right\},
\end{align}
and the proof is complete.


\begin{thebibliography}{AA}
\bibitem{GoodCodes}
R. Ahlswede and G. Dueck, ``Good codes can be produced by a few permutations,'' {\it IEEE Trans. on Inform. Theory}, vol. 28, no. 3, pp. 430--443, May 1982.
	

	\bibitem{BargForney}
	A. Barg and G. D. Forney, Jr., ``Random codes: minimum distances and error exponents,'' {\it IEEE Trans. on Inform. Theory}, vol. 48, no. 9, pp. 2568--2573, September 2002.
	
	\bibitem{SW2}
	J. Chen, D.-K. He, A. Jagmohan, and L. A. Lastras-Mont\~ano, ``On the reliability function of variable--rate Slepian--Wolf coding,'' Entropy, vol. 19, 389, 2017. doi:10.3390/e19080389

\bibitem{C1982}
I. Csisz\'ar, ``Linear codes for sources and source networks: Error exponents, universal coding," {\it IEEE Trans. on Inform. Theory}, vol. 28, no. 4, pp. 585--592, July 1982.

\bibitem{CK1980}
I.~Csisz\'ar and J.~K\"orner, ``Towards a general theory of source networks," {\it IEEE Trans. on Inform. Theory}, vol. 26, no. 2, pp. 155--165, March 1980.

\bibitem{CKgraph}
I.~Csisz\'ar and J.~K\"orner, ``Graph decomposition: a new key to coding theorems," {\it IEEE Trans. on Inform. Theory}, vol. 27, no. 1, pp. 5--12, January 1981.
%
%

%
	
	
	\bibitem{CK11}
	I.~Csisz\'ar and J.~K\"orner, {\it Information Theory: Coding Theorems for 	Discrete Memoryless Systems}, Cambridge University Press, 2011.
	
	
	
	
	\bibitem{GAL76}
	R. G. Gallager, ``Source coding with side information and universal coding,'' LIDS-P-937, M.I.T., 1976.
	
	\bibitem{KW2011}
	B. G. Kelly and A. B. Wagner, ``Improved source coding exponents via Witsenhausen's rate," {\it IEEE Trans. on Inform. Theory}, vol. 57, no. 9, pp. 5616--5633, September 2011.
	
	\bibitem{KW2012}
	B. G. Kelly and A. B. Wagner, ``Reliability in source coding with side information," {\it IEEE Trans. on Inform. Theory}, vol. 58, no. 8, pp. 5086--5111, August 2012.
	
	\bibitem{LCV2017}
	J. Liu, P. Cuff, and S. Verd\'u, ``On $\alpha$--decodability and $\alpha$--likelihood decoder," in {\it Proc. 55th Ann. Allerton Conf. Comm. Control Comput.}, Monticello, IL, October 2017.
	
	\bibitem{MERHAV09}
	N. Merhav,``Statistical physics and information theory,''  
	{\it Foundations and Trends in Communications and Information Theory,} vol. 6, nos. 1–2, pp. 1--212, 2009. 
	
	%
	
%
	\bibitem{MERHAV2017}
	N. Merhav, ``The generalized stochastic likelihood decoder: random coding and expurgated bounds," {\it IEEE Trans. on Inform. Theory}, vol. 63, no. 8, pp. 5039--5051, August 2017. 
	See also a correction at {\it IEEE Trans. on Inform. Theory}, vol. 63, no. 10, pp. 6827--6829, October 2017.
	
	\bibitem{MERHAV_TYPICAL}
	N. Merhav, ``Error exponents of typical random codes," {\it IEEE Trans. on Inform. Theory}, vol. 64, no. 9, pp. 6223--6235, September 2018.
	
	\bibitem{MERHAV_GAUSS}
	N. Merhav, ``Error exponents of typical random codes for the colored Gaussian channel," {\it IEEE Trans. on Inform. Theory}, vol. 65, no. 12, pp. 8164--8179, December 2019.
	
	\bibitem{MERHAV_TRELLIS}
	N. Merhav, ``Error exponents of typical random trellis codes," {\it IEEE Trans. on Inform. Theory}, vol. 66, no. 4, pp. 2067--2077, April 2020.
	
	\bibitem{MERHAV_IID}
	N. Merhav, ``A Lagrange--dual lower bound to the error exponent of the typical random code," to appear in {\it IEEE Trans. on Inform. Theory}, vol. 66, no. 6, pp. 3456--3464, June 2020.	
	
	\bibitem{PRAD2014}
	A. Nazari, A. Anastasopoulos, and S. S. Pradhan, ``Error exponent for multiple--access channels: lower bounds," {\it IEEE Trans. on Inform. Theory}, vol. 60, no. 9, pp. 5095--5115, September 2014.
	
	\bibitem{OH1994}
	Y. Oohama and T. S. Han, ``Universal coding for the Slepian-Wolf data compression system and the strong converse theorem," {\it IEEE Trans. on Inform. Theory}, vol. 40, no. 6, pp. 1908--1919, November 1994.
	
	
	%
	%

	\bibitem{SW}
	D. Slepian and J. Wolf, ``Noiseless coding of correlated information sources,'' {\it IEEE Trans. on Inform. Theory}, vol. 19, no. 4, pp. 471--480, July 1973.

	\bibitem{SSG}
	A. Somekh--Baruch, J. Scarlett, and A. Guill\'en i F\`abregas, ``Generalized random Gilbert--Varshamov codes,'' {\it IEEE Trans. on Inform. Theory}, vol. 65, no. 5, pp. 3452--3469, May 2019.


	
	\bibitem{TMWG}
	R.~Tamir (Averbuch), N.~Merhav, N.~Weinberger, and A.~Guill\'en i F\`abregas,
	``Large deviations behavior of the logarithmic error probability of random codes,'' {\it IEEE Trans.\ on Inform.\ Theory}, vol.\ 66, no.\ 11, pp.\ 6635--6659, November 2020.
	
	\bibitem{TM_MMI}
	R.~Tamir (Averbuch) and N.~Merhav, ``The MMI decoder is asymptotically optimal for the typical random code and for the expurgated code,'' submitted for publication, July 2020.
	

	\bibitem{SW9}
	N. Weinberger and N. Merhav, ``Optimum tradeoffs between the error exponent and the excess--rate exponent of variable--rate Slepian--Wolf coding,'' {\it IEEE Trans. on Inform. Theory}, vol. 61, no. 4, pp. 2165--2190, April 2015.
	

		

	%
	

	
\end{thebibliography}
\end{document}